\documentclass[final]{IEEEtran}

\usepackage{cite}
\usepackage{bm}
\usepackage{amsmath,amsthm}
\usepackage{extarrows}
\usepackage{amssymb}
\usepackage{graphicx}
\usepackage{xcolor}
\usepackage{subfig}
\usepackage{hyperref}
\usepackage{algorithmic,algorithm}

\newcommand{\mb}{\bm}
\newcommand{\mr}{\mathrm}

\newcommand{\BE}{\begin{equation}}
\newcommand{\EE}{\end{equation}}
\newcommand{\BS}{\begin{subequations}}
\newcommand{\ES}{\end{subequations}}
\newcommand{\UH}{\mathsf{H}}   
\newcommand{\UT}{\mathsf{T}}   

\newcommand{\df}{\mathsf{df}}

\newcommand{\Mydef}{\overset{  \scriptscriptstyle \Delta  }{=}}

\newtheorem{theorem}{Theorem}

\newtheorem{assumption}{Assumption}

\newtheorem{remark}{Remark}
\newtheorem{lemma}{Lemma}

\newtheorem{claim}{Claim}

\newcommand{\opt}{\star}
\newcommand{\G}{\langle\bm{G}\rangle} 
\newcommand{\Gs}{\bar{G}} 

\newcommand{\T}{\mathcal{T}}
\newcommand{\Tr}{\mathsf{Tr}}

\newcommand{\Alg}{\textsf{PCA-EP }}

\graphicspath{{figures/}}

\newcommand{\nobracket}{}

\newcommand{\tmmathbf}[1]{\ensuremath{\boldsymbol{#1}}}
\newcommand{\tmop}[1]{\ensuremath{\operatorname{#1}}}

\begin{document}

\title{Spectral Method for Phase Retrieval: an Expectation Propagation Perspective}
\author{Junjie Ma,~Rishabh Dudeja,~Ji Xu,~Arian Maleki,~Xiaodong~Wang

\thanks{This work was supported in part by National Science Foundation (NSF) under grant CCF 1814803 and Office of Naval Research (ONR) under grant N000141712827.}

\thanks{J.~Ma was with the Department of Statistics and the Department of Electrical Engineering of Columbia University, New York, USA. He is now with the Institute of Computational Mathematics and Scientific/Engineering Computing, Academy of
Mathematics and Systems Science, Chinese Academy of Sciences, Beijing, China. (e-mail:majunjie@lsec.cc.ac.cn)}

\thanks{R.~Dudeja and Arian Maleki are with the Department of Statistics, Columbia University, New York, USA. (e-mail: rd2714@columbia.edu; arian@stat.columbia.edu)}

\thanks{J.~Xu is with the Department of Computer Science, Columbia University, New York, USA. (e-mail: jixu@cs.columbia.edu)}

\thanks{X.~Wang is with the Department of Electrical Engineering, Columbia University, New York, USA. (e-mail: xw2008@columbia.edu)}

}

\maketitle

\begin{abstract}
Phase retrieval refers to the problem of recovering a signal $\bm{x}_{\star}\in\mathbb{C}^n$ from its phaseless measurements $y_i=|\bm{a}_i^{\UH}\bm{x}_{\star}|$, where $\{\bm{a}_i\}_{i=1}^m$ are the measurement vectors. Spectral method is widely used for initialization in many phase retrieval algorithms. The quality of spectral initialization can have a major impact on the overall algorithm. In this paper, we focus on the model where $\bm{A}=[\bm{a}_1,\ldots,\bm{a}_m]^{\UH}$ has orthonormal columns, and study the spectral initialization under the asymptotic setting $m,n\to\infty$ with $m/n\to\delta\in(1,\infty)$. We use the expectation propagation framework to characterize the performance of spectral initialization for Haar distributed matrices. Our numerical results confirm that the predictions of the EP method are accurate for not-only Haar distributed matrices, but also for realistic Fourier based models (e.g. the coded diffraction model). The main findings of this paper are the following:
\begin{itemize}
\item There exists a threshold on $\delta$ (denoted as $\delta_{\mr{weak}}$) below which the spectral method cannot produce a meaningful estimate. We show that $\delta_{\mr{weak}}=2$ for the column-orthonormal model. In contrast, previous results by Mondelli and Montanari show that $\delta_{\mr{weak}}=1$ for the i.i.d. Gaussian model.
\item The optimal design for the spectral method coincides with that for the i.i.d. Gaussian model, where the latter was recently introduced by Luo, Alghamdi and Lu.
\end{itemize}
\end{abstract}




\section{Introduction}\label{Sec:model}

In many scientific and engineering applications, it is expensive or even impossible to measure the phase of a signal due to physical limitations \cite{Shechtman15}. Phase retrieval refers to algorithmic methods for reconstructing signals from magnitude-only measurements. The measuring process for the phase retrieval problem can be modeled as
\BE
y_i=\left|(\bm{Ax}_{\star})_i\right|,\quad i=1,2\ldots,m,
\EE
where $\bm{A}\in\mathbb{C}^{m\times n}$ ($m>n$) is the measurement matrix, $\bm{x}_{\star}\in\mathbb{C}^{n\times1}$ is the signal to be recovered, and $(\bm{Ax}_{\star})_i$ denotes the $i$th entry of the vector $\bm{Ax}_{\star}$. Phase retrieval has important applications in areas ranging from X-ray crystallography, astronomical imaging, and many others \cite{Shechtman15}.

Phase retrieval has attracted a lot of research interests since the work of Candes \textit{et al} \cite{Candes2013}, where it is proved that a convex optimization based algorithm can provably recover the signal under certain randomness assumptions on $\bm{A}$. However, the high computational complexity of the PhaseLift algorithm in \cite{Candes2013} prohibits its practical applications. More recently, a lot of algorithms were proposed as low-cost iterative solvers for the following nonconvex optimization problem (or its variants) \cite{CaLiSo15,ChenCandes17,Wang2016,Zhang2016reshaped}:
\BE\label{Eqn:loss}
\underset{\bm{x}\in\mathbb{C}^{m}}{\text{argmin}} \  \frac{1}{m}\sum_{i=1}^m \big ( |y_i|^2 - |(\bm{Ax})_i|^2  \big)^2.
\EE
These algorithms typically initialize their estimates using a spectral method \cite{Eetrapalli2013} and then refine the estimates via alternative minimization \cite{Eetrapalli2013}, gradient descend (and variants) \cite{CaLiSo15,ChenCandes17,Wang2016,Zhang2016reshaped} or approximate message passing \cite{MXM18}. Since the problem in \eqref{Eqn:loss} is nonconvex, the initialization step plays a crucial role for many of these algorithms.

Spectral methods are widely used for initializing local search algorithms in many signal processing applications. In the context of phase retrieval, spectral initialization was first proposed in \cite{Eetrapalli2013} and later studied in \cite{ChenCandes17,Wang2016,gao2016phaseless,Lu17,Mondelli2017,Lu2018}. To be specific, a spectral initial estimate is given by the leading eigenvector (after proper scaling) of the following data matrix \cite{Lu17}:
\BE\label{Eqn:data_matrix}
\bm{D}\Mydef \bm{A}^{\UH}\mr{diag}\{\T(y_1),\ldots,\T(y_m)\}\bm{A},
\EE
where $\mr{diag}\{a_1,\ldots,a_m\}$ denotes a diagonal matrix formed by $\{a_i\}_{i=1}^m$ and $\mathcal{T}:\mathbb{R}_+\mapsto\mathbb{R}$ is a nonlinear processing function. A natural measure of the quality of the spectral initialization is the cosine similarity \cite{Lu17} between the estimate and the true signal vector:
\BE\label{Eqn:rho_def}
P^2_\T(\bm{x}_{\star},\hat{\bm{x}})\Mydef \frac{|\bm{x}_{\star}^{\UH}\hat{\bm{x}}|^2}{\|\bm{x}_{\star}\|^2\|\hat{\bm{x}}\|^2},
\EE
where $\hat{\bm{x}}$ denotes the spectral estimate. The performance of the spectral initialization highly depends on the processing function $\T$. Popular choices of $\T$ include the ``trimming'' function $\T_{\mr{trim}}$ proposed in \cite{ChenCandes17} (the name follows \cite{Lu17}), the ``subset method'' in \cite{Wang2016}, $\T_{\mr{MM}}(y)$ proposed by Mondelli and Montanari in \cite{Mondelli2017}, and $\T_{\opt}$ recently proposed in \cite{Lu2018}:
\BS\label{Eqn:T_examples}
\begin{align}
\T_{\mr{trim}}(y)&=\delta{y}^2\cdot\mathbb{I}(\delta{y}^2<c_2^2),\label{Eqn:T_trunc}\\
\T_{\mr{subset}}(y)&=\mathbb{I}(\delta y^2>c_1),\\
\T_{\mr{MM}}(y)&=1-\frac{\sqrt{\delta}}{\delta y^2+\sqrt{\delta}-1},\\
\T_{\opt}(y)&=1-\frac{1}{\delta y^2}.\label{Eqn:T_Lu}
\end{align}
\ES
In the above expressions, $c_1>0$ and $c_2>0$ are tunable thresholds, and $\mathbb{I}(\cdot)$ denotes an indicator function that equals one when the condition is satisfied and zero otherwise.

The asymptotic performance of the spectral method was studied in \cite{Lu17} under the assumption that $\bm{A}$ contains i.i.d. Gaussian entries. The results of \cite{Lu17} unveil a phase transition phenomenon in the regime where $m,n\to\infty$ with $m/n=\delta\in(0,\infty)$ fixed. Specifically, there exists a threshold on the measurement ratio $\delta$: below this threshold the cosine similarity $P_{\T}^2(\hat{\bm{x}},\bm{x}_{\star})$ converges to zero (meaning that $\hat{\bm{x}}$ is not a meaningful estimate) and above it $P^2_{\T}(\hat{\bm{x}},\bm{x}_{\star})$ is strictly positive. Later, Mondelli and Montanari showed in \cite{Mondelli2017} that $\T_{\mr{MM}}$ defined in \eqref{Eqn:T_examples} minimizes the above-mentioned reconstruction threshold. Following \cite{Mondelli2017}, we will call the minimum threshold (over all possible $\T$) the \textit{weak threshold}. It is proved in \cite{Mondelli2017} that the weak thresholds are $0.5$ and $1$, respectively, for real and complex valued models. Further, \cite{Mondelli2017} showed that these thresholds are also information-theoretically optimal for weak recovery. Notice that $\T_{\mr{MM}}$ minimizes the reconstruction threshold, but  does not necessarily maximize $P^2_{\T}(\hat{\bm{x}},\bm{x}_{\star})$ when $\delta$ is larger than the weak threshold. The latter criterion is more relevant in practice, since intuitively speaking a larger $P_{\T}^2(\hat{\bm{x}},\bm{x}_{\star})$ implies better initialization, and hence the overall algorithm is more likely to succeed. This problem was recently studied by Luo, Alghamdi and Lu in \cite{Lu2018}, where it was shown that the function $\T_{\star}$ defined in \eqref{Eqn:T_examples} uniformly maximizes $P_{\T}^2(\hat{\bm{x}},\bm{x}_{\star})$ for an arbitrary $\delta$ (and hence also achieves the weak threshold).

Notice that the analyses of the spectral method in \cite{Lu17,Mondelli2017,Lu2018} are based on the assumption that $\bm{A}$ has i.i.d. Gaussian elements. This assumption is key to the random matrix theory (RMT) tools used in \cite{Lu17}. However, the measurement matrix $\bm{A}$ for many (if not all) important phase retrieval applications is a Fourier matrix\cite{millane1990phase}. In certain applications, it is possible to randomize the measuring mechanism by adding random masks \cite{Candes15_Diffraction}. Such models are usually referred to as \textit{coded diffraction patterns} (CDP) models \cite{Candes15_Diffraction}. In the CDP model, $\bm{A}$ can be expressed as
\BE\label{Eqn:CDP_model}
\bm{A}=
\begin{bmatrix}
\bm{F}\bm{P}_1\\
\bm{F}\bm{P}_2\\
\ldots\\
\bm{F}\bm{P}_L
\end{bmatrix},
\EE
where $\bm{F}^{n\times n}$ is a square discrete Fourier transform (DFT) matrix, and $\bm{P}_l=\mr{diag}\{e^{j\theta_{l,1}},\ldots,e^{j\theta_{l,n}}\}$ represents the effect of random masks. For the CDP model, $\bm{A}$ has orthonormal columns, namely,
\BE
\bm{A}^{\UH}\bm{A}=\bm{I}.
\EE

In this paper, we assume $\bm{A}$ to be an isotropically random unitary matrix (or simply Haar distributed). We study the performance of the spectral method and derive a formula to predict the cosine similarity $\rho^2_{\T}(\hat{\bm{x}},\bm{x}_{\star})$. We conjecture that our prediction is asymptotically exact as $m,n
\to\infty$ with $m/n=\delta\in(1,\infty)$ fixed. Based on this conjecture, we are able to show the following.
\begin{itemize}
\item There exists a threshold on $\delta$ (denoted as $\delta_{\mr{weak}}$) below which the spectral method cannot produce a meaningful estimate. We show that $\delta_{\mr{weak}}=2$ for the column-orthonormal model. In contrast, previous results by Mondelli and Montanari show that $\delta_{\mr{weak}}=1$ for the i.i.d. Gaussian model.
\item The optimal design for the spectral method coincides with that for the i.i.d. Gaussian model, where the latter was recently derived by Luo, Alghamdi and Lu.
\end{itemize}

Our asymptotic predictions are derived by analyzing an expectation propagation (EP) \cite{Minka2001} type message passing algorithm \cite{opper2005,ma2015turbo} that aims to find the leading eigenvector. A key tool in our analysis is a state evolution (SE) technique that has been extensively studied for the compressed sensing problem \cite{ma2015turbo,Ma2016,liu2016,Rangan17,Takeuchi2017,He2017}. Several arguments about the connection between the message passing algorithm and the spectral estimator is heuristic, and thus the results in this paper are not rigorous yet. Nevertheless, numerical results suggest that our analysis accurately predicts the actual performance of the spectral initialization under the practical CDP model. This is perhaps a surprising phenomenon, considering the fact that the sensing matrix in \eqref{Eqn:CDP_model} is still quite structured although the matrices $\{\bm{P}_l\}$ introduce certain randomness.

We point out that the predictions of this paper have been rigorously proved using RMT tools in \cite{Dudeja2020}. It is also expected that the same results can be obtained by using the replica method \cite{Takeda2006,Tulino2013,Vehkapera2014}, which is a non-rigorous, but powerful tool from statistical physics. Compared with the replica method, our method seems to be technically simpler and more flexible (e.g., we might be able to handle the case where the signal is known to be positive).

Finally, it should be noted that although the coded diffraction pattern model is much more practical than the i.i.d. Gaussian model, it is still far from practical for some important phase retrieval applications (e.g., X-ray crystallography). As pointed out in \cite{luke2017phase,elser2018benchmark}, designing random masks for X-ray crystallography would correspond to manipulating the masks at sub-nanometer resolution and thus physically infeasible. On the other hand, it has been suggested to use multiple illuminations for certain types of imaging applications (see \cite[Section 2.3]{fannjiang2020numerics} for discussions and related references). We emphasize that the results in this paper only apply to (empirically) the random-masked Fourier model, and the theory for the challenging non-random Fourier case is still open.

\textit{Notations:} $a^{\ast}$ denotes the conjugate of a complex number $a$. We use bold lower-case and upper case letters for vectors and matrices respectively. For a matrix $\bm{A}$, $\bm{A}^{\UT}$ and $\bm{A}^{\UH}$ denote the transpose of a matrix and its Hermitian respectively. $\mr{diag}\{a_1,a_2,\ldots, a_n\}$ denotes a diagonal matrix with $\{a_i\}_{i=1}^n$ being the diagonal entries. $\bm{x}\sim\mathcal{CN}(\mathbf{0},\sigma^2\bm{I})$ is circularly-symmetric Gaussian if $\mathrm{Re}(\bm{x})\sim\mathcal{N}(\mathbf{0},\sigma^2/2\bm{I})$, $\mathrm{Im}(\bm{x})\sim\mathcal{N}(\mathbf{0},\sigma^2/2\bm{I})$, and $\mathrm{Re}(\bm{x})$ is independent of $\mathrm{Im}(\bm{x})$. For $\bm{a},\bm{b}\in\mathbb{C}^m$, define $\langle\bm{a},\bm{b}\rangle=\sum_{i=1}^m a_i^{\ast}b_i$. For a Hermitian matrix $\bm{H}\in\mathbb{C}^{n\times n}$, $\lambda_1(\bm{H})$ and $\lambda_n(\bm{H})$ denote the largest and smallest eigenvalues of $\bm{H}$ respectively. For a non-Hermitian matrix $\bm{Q}$, $\lambda_1(\bm{Q})$ denotes the eigenvalue of $\bm{Q}$ that has the largest magnitude. $\Tr(\cdot)$ denotes the trace of a matrix. $\overset{\mr{P}}{\to}$ denotes convergence in probability.

\section{Asymptotic Analysis of the Spectral Method}
This section presents the main results of this paper. Our results have not been fully proved, we call them ``claims" throughout this section to avoid confusion. The rationales for our claims will be discussed in Section \ref{Sec:GEC}.
\subsection{Assumptions}\label{Sec:assumption}
In this paper, we make the following assumption on the sensing matrix $\bm{A}$.

\begin{assumption}\label{Ass:matrix}
The sensing matrix $\bm{A}\in\mathbb{C}^{m\times n}$ ($m> n$) is sampled uniformly at random from column orthogonal matrices satisfying $\bm{A}^{\UH}\bm{A}=\bm{I}$.
\end{assumption}

Assumption \ref{Ass:matrix} introduces certain randomness assumption about the measurement matrix $\bm{A}$. However, in practice, the measurement matrix is usually a Fourier matrix or some variants of it. One particular example is the coded diffraction patterns (CDP) model in \eqref{Eqn:CDP_model}. For this model, the only freedom one has is the ``random masks'' $\{\bm{P}_l\}$, and the overall matrix is still quite structured. In this regard, it is perhaps quite surprising that the predictions developed under Assumption \ref{Ass:matrix} are very accurate even for the CDP model. Please refer to Section \ref{Sec:numerical} for our numerical results.

We further make the following assumption about the processing function $\mathcal{T}:\mathbb{R}_+\mapsto\mathbb{R}$.

\begin{assumption}\label{Ass:T}
The processing function $\mathcal{T}:\mathbb{R}_+\mapsto\mathbb{R}$ satisfies $ \sup_{y\ge0}\,\mathcal{T}(y)=T_{\max}$, where $T_{\max}<\infty$.

\end{assumption}

%

As pointed out in \cite{Lu17}, the boundedness of $\T$ in Assumption \ref{Ass:T}-(I) is the key to the success of the truncated spectral initializer proposed in \cite{ChenCandes17}. (Following \cite{Lu17}, we will call it the trimming method in this paper.) Notice that we assume $T_{\max}=1$ without any loss of generality. To see this, consider the following modification of $\T$:
\[
\hat{\T}(y)\Mydef \frac{C+\T(y)}{C+T_{\max}},
\]
where $C>\max\{0,-T_{\max}\}$. It is easy to see that
\[
\begin{split}
&\bm{A}^{\UH}\mr{diag}\{\hat{\T}(y_1),\ldots,\hat{\T}(y_m)\}\bm{A}\\
&=\frac{C}{C+T_{\max}}\bm{I}+\frac{1}{C+T_{\max}}\bm{A}^{\UH}\bm{TA},
\end{split}
\]
where we used $\bm{A}^{\UH}\bm{A}=\bm{I}$. Clearly, the top eigenvector of $\bm{A}^{\UH}\bm{TA}$ is the same as that of $\bm{A}^{\UH}\mr{diag}\{\hat{\T}(y_1),\ldots,\hat{\T}(y_m)\}\bm{A}$, where for the latter matrix we have $\sup_{y\ge0}\, \hat{\T}(y)=1$.

\subsection{Asymptotic analysis}\label{Sec:main_results}
Let $\hat{\bm{x}}$ be the principal eigenvector of the following data matrix:
\BE
\bm{D}\Mydef \bm{A}^{\UH}\bm{T}\bm{A},
\EE
where $\bm{T}=\mr{diag}\{\T(y_1),\ldots,\T(y_m)\}.$
Following \cite{Lu17}, we use the squared cosine similarity defined below to measure the accuracy of $\hat{\bm{x}}$:
\BE\label{Eqn:rho_def}
P^2_\T(\hat{\bm{x}},\bm{x}_{\star})\Mydef \frac{|\bm{x}_{\star}^{\UH}\hat{\bm{x}}|^2}{\|\bm{x}_{\star}\|^2\|\hat{\bm{x}}\|^2}.
\EE

Our goal is to understand how $P_\T^2(\hat{\bm{x}},\bm{x}_{\star})$ behaves when $m,n\to\infty$ with a fixed ratio $m/n=\delta$. It seems possible to solve this problem by adapting the tools developed in \cite{Lu17}. Yet, we take a different approach which we believe to be technically simpler and more flexible. Specifically, we will derive a message passing algorithm to find the leading eigenvector of the data matrix. Central to our analysis is a deterministic recursion, called state evolution (SE), that characterizes the asymptotic behavior of the message passing algorithm. By analyzing the stationary points of the SE, we obtain certain predictions about the spectral estimator. This approach has been adopted in \cite{Bayati&Montanari12} to analyze the asymptotic performance of the LASSO estimator and in \cite{Montanari_pca} for the nonnegative PCA estimator, based on the approximate message passing (AMP) algorithm\cite{DoMaMo09,Bayati&Montanari11}. However, the SE analysis of AMP does not apply to the partial orthogonal matrix model considered in this paper.

Different from \cite{Montanari_pca}, our analysis is based on a variant of the expectation propagation (EP) algorithm \cite{Minka2001,opper2005}, called \textsf{PCA-EP} in this paper. Different from AMP, such EP-style algorithms could be analyzed via SE for a wider class of measurement matrices, including the Haar model considered here\cite{ma2015turbo,liu2016,Ma2016,fletcher2016,Rangan17,Takeuchi2017,Schniter16,He2017}. The derivations of \Alg and its SE analysis will be introduced in Section \ref{Sec:GEC}.
\vspace{5pt}

Our characterization of $P^2_\T(\hat{\bm{x}},\bm{x}_{\star})$ involves the following functions (for $\mu\in(0,1]$):
\BE\label{Eqn:psi_definitions}
\begin{split}
\psi_1(\mu) &\Mydef \frac{\mathbb{E}\left[\delta|{Z}_{\star}|^2G\right]}{\mathbb{E}[G]},\\
\psi_2(\mu) &\Mydef \frac{\mathbb{E}\left[G^2\right]}{\left(\mathbb{E}[G]\right)^2},\\
\psi_3(\mu) &\Mydef \frac{\sqrt{\mathbb{E}\left[\delta|{Z_{\star}}|^2G^2\right]}}{\mathbb{E}[G]},
\end{split}
\EE
where $Z_{\star}\sim\mathcal{CN}(0,1/\delta)$ and $G$ is a shorthand for
\BE
G(|Z_\star|,\mu)\Mydef \frac{1}{1/\mu-\T(|Z_\star|)}.
\EE
We note that $\psi_1$, $\psi_2$ and $\psi_3$ all depend on the processing function $\T$. However, to simplify notation, we will not make this dependency explicit. Claim \ref{Claim:correlation} below summarizes our asymptotic characterization of the cosine similarity $P_\T^2(\hat{\bm{x}},\bm{x}_{\star})$.

\begin{claim}[Cosine similarity]\label{Claim:correlation}
Define
\BE\label{Eqn:SNR_asym}
\rho^2_\T(\mu,\delta)\Mydef\frac{\left(\frac{\delta}{\delta-1}\right)^2-\frac{\delta}{\delta-1}\cdot\psi_2(\mu)}{\psi_3^2(\mu)-\frac{\delta}{\delta-1}\cdot\psi_2(\mu)},
\EE
and
\BE\label{Eqn:Lambda_def}
\Lambda(\mu) \Mydef \frac{1}{\mu}-\frac{\delta-1}{\delta}\cdot\frac{1}{\mathbb{E}[G(|Z_\star|,\mu)]},
\EE
where $Z_\star\sim\mathcal{CN}(0,1/\delta)$. Let
\BE\label{Eqn:mu_star_def}
\bar{\mu}(\delta)\Mydef \underset{\mu\in(0,1]}{\text{argmin}}\ \Lambda(\mu).
\EE
Then, as $m,n\to\infty$ with a fixed ratio $m/n=\delta>1$, we have
\BE\label{Eqn:PT_result}
\lim_{m\to\infty}P_\T^2(\hat{\bm{x}},\bm{x}_{\star})\overset{\mr{P}}{\rightarrow}\rho_\T^2(\delta),
\EE
where with a slight abuse of notations,
\BE\label{Eqn:SNR_asym2}
\rho^2_\T(\delta)\Mydef
\begin{cases}
\rho_\T^2\big(\hat{\mu}(\delta),\delta\big), & \text{if } \psi_1\left(\bar{\mu}(\delta)\right)\ge\frac{\delta}{\delta-1},\\
0 & \text{if } \psi_1\left(\bar{\mu}(\delta)\right)<\frac{\delta}{\delta-1},
\end{cases}
\EE
and $\hat{\mu}(\delta)$ is a solution to
\BE\label{Eqn:psi1_fixed}
\psi_1(\mu)=\frac{\delta}{\delta-1},\quad \mu\in(0,\bar{\mu}(\delta)].
\EE
\end{claim}

Claim \ref{Claim:correlation}, which is reminiscent of Theorem 1 in \cite{Lu17}, reveals a phase transition behavior of $P^2_\T({\hat{\bm{x}}},\bm{x}_{\star})$: $P^2_\T({\hat{\bm{x}}},\bm{x}_{\star})$ is strictly positive when $\psi_1(\bar{\mu}(\delta))>\frac{\delta}{\delta-1}$ and is zero otherwise. In the former case, the spectral estimator is positively correlated with the true signal and hence provides useful information about $\bm{x}_{\star}$; whereas in the latter case $P^2_\T({\hat{\bm{x}}},\bm{x}_{\star})\to0$, meaning that the spectral estimator is asymptotically orthogonal to the true signal and hence performs no better than a random guess. Claim \ref{Claim:correlation} is derived from the analysis of an EP algorithm whose fixed points correspond to the eigenvector of the matrix $\bm{D}$. The main difficulty for making this approach rigorous is to prove the EP estimate converges to the leading eigenvector. Currently, we do not have a proof yet, but we will provide some heuristic arguments in Section \ref{Sec:GEC}.

\begin{remark}
For notational simplicity, we will write $\bar{\mu}(\delta)$ and $\hat{\mu}(\delta)$ as $\bar{\mu}$ and $\hat{\mu}$, respectively.
\end{remark}

Some remarks about the phase transition condition are given in Lemma \ref{Lem:PT_second} below. Item (i) guarantees the uniqueness of the solution to \eqref{Eqn:psi1_fixed}. Item (ii) is the actual intuition that leads to the conjectured phase transition condition. The proof of Lemma \ref{Lem:PT_second} can be found in Appendix \ref{App:proof_PT_sec}.

\begin{lemma}\label{Lem:PT_second}
(i) Eq. \eqref{Eqn:psi1_fixed} has a unique solution if $\psi_1\left(\bar{\mu}\right)>\frac{\delta}{\delta-1}$, where $\bar{\mu}$ is defined in \eqref{Eqn:mu_star_def}. (ii) $\psi_1\left(\bar{\mu}\right)>\frac{\delta}{\delta-1}$ if and only if there exists a $\hat{\mu}\in(0,1]$ such that
\[
\psi_1(\hat{\mu})=\frac{\delta}{\delta-1}\quad \text{and}\quad 0<\rho_\T^2(\hat{\mu},\delta)<1.
\]
The latter statement is equivalent to
\[
\psi_1(\hat{\mu})=\frac{\delta}{\delta-1}\quad \text{and}\quad \psi_2(\hat{\mu})<\frac{\delta}{\delta-1}.
\]
\end{lemma}

 \begin{figure*}[t]
\centering
\subfloat{\includegraphics[width=.33\textwidth]{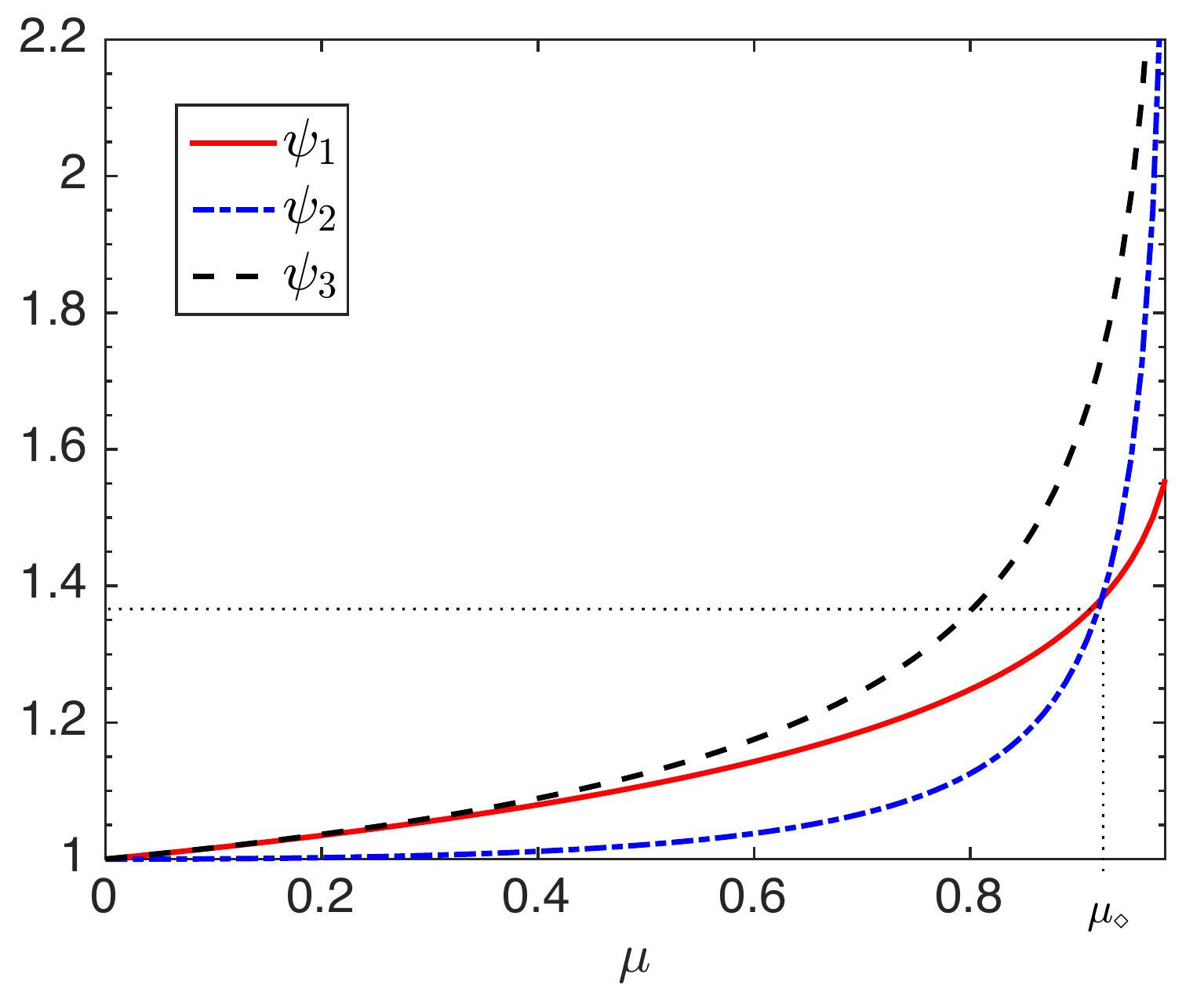}}
\subfloat{\includegraphics[width=.33\textwidth]{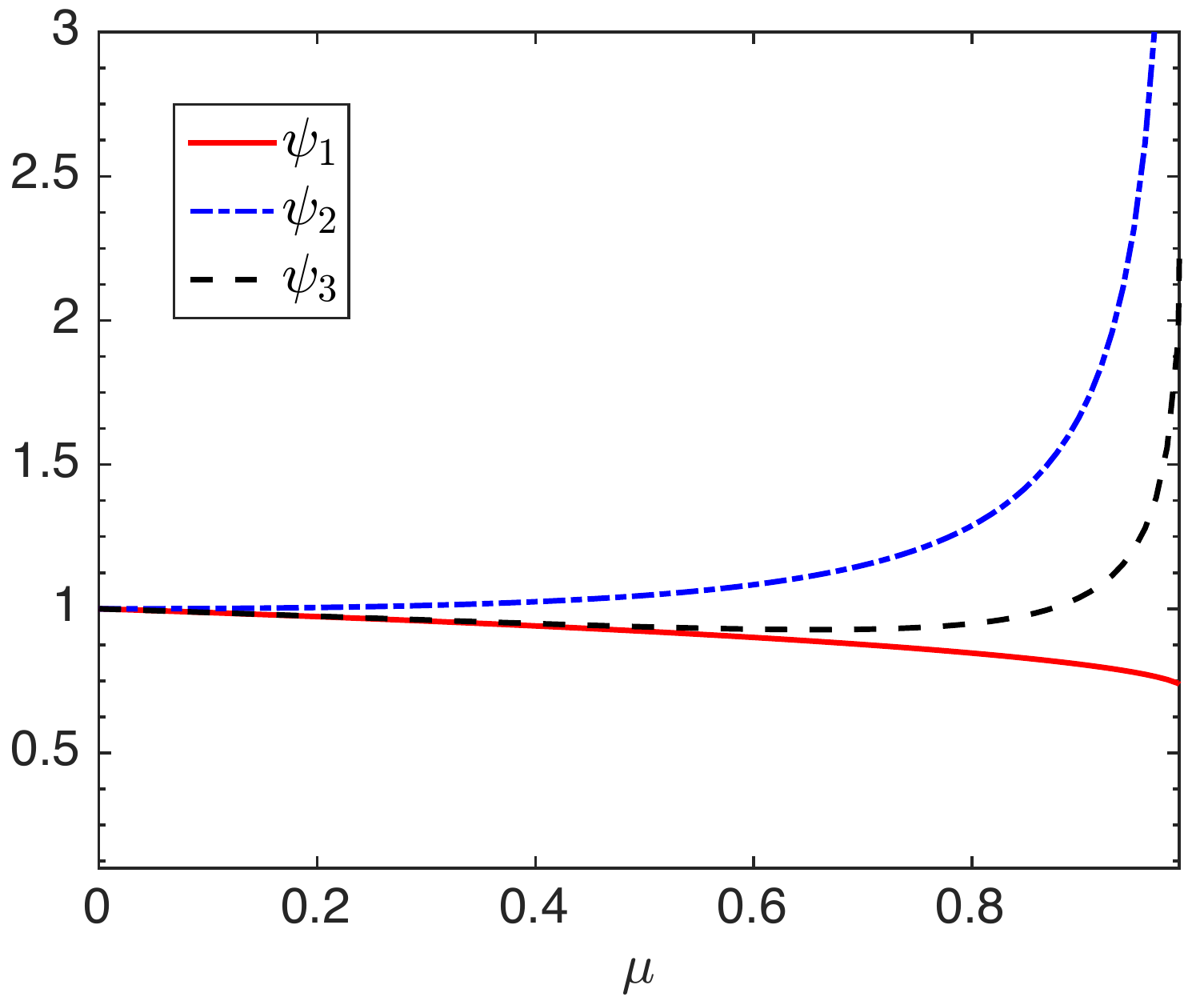}}
\subfloat{\includegraphics[width=.33\textwidth]{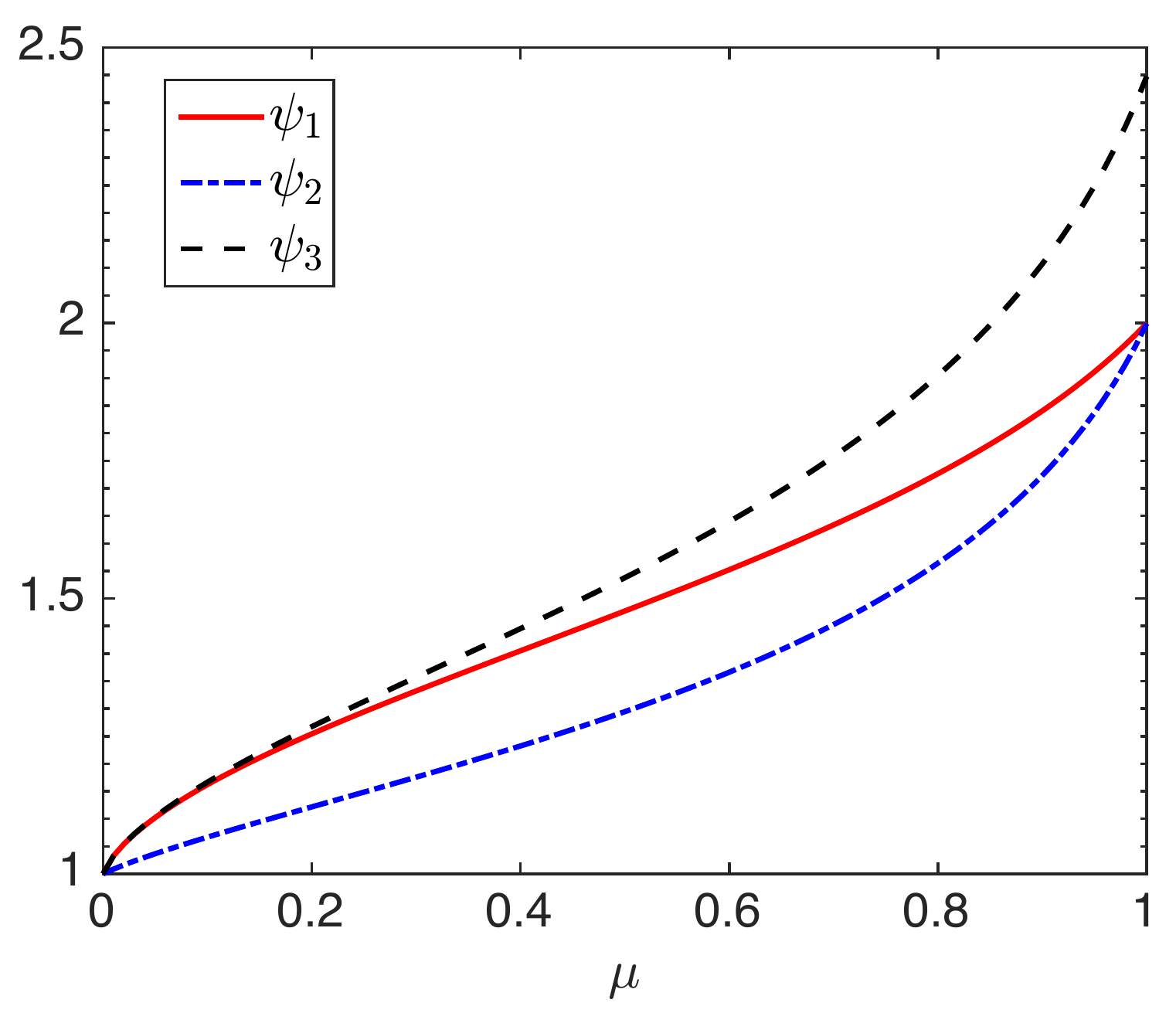}}
\caption{Examples of $\psi_1$, $\psi_2$ and $\psi_3$. \textbf{Left}: $\T=\T_{\mr{trim}}$ with $c_2=2$. \textbf{Middle:} $\T=\T_{\mr{trim}}$ with $c_2=0.8$. In both cases, we further scaled $\T_{\mr{trim}}$ so that $\sup_{y\ge0}\ \T(y)=1$. \textbf{Right:} $\T=\T_\opt$.}\label{Fig:psi_functions}
\end{figure*}

Fig.~\ref{Fig:psi_functions} plots $\psi_1$, $\psi_2$ and $\psi_3$ for various choices of $\T$. The first two subfigures employ the trimming function $\T_{\mr{trim}}$ in \eqref{Eqn:T_trunc} (with different truncation thresholds), and third subfigure uses the function $\T_\opt$ in \eqref{Eqn:T_Lu}.

Notice that for the first and third subfigures of Fig.~\ref{Fig:psi_functions}, $\psi_1$ and $\psi_2$ have  unique nonzero crossing points in $\mu\in(0,1]$, which we denote as $\mu_{\diamond}$. Further, $\psi_1>\psi_2$ for $\mu<\mu_\diamond$, and $\psi_1<\psi_2$ for $\mu>\mu_\diamond$ (for the first subfigure). Then by Lemma \ref{Lem:PT_second}-(ii), a phase transition happens when $\psi_1(\mu_{\diamond})>\frac{\delta}{\delta-1}$, or equivalently
\[
\delta> \frac{\psi_1(\mu_{\diamond})}{\psi_1(\mu_{\diamond})-1}\Mydef \delta_\T.
\]
The spectral estimator is positively correlated with the signal when $\delta>\delta_{\T}$, and perpendicular to it otherwise.

Now consider the second subfigure of Fig.~\ref{Fig:psi_functions}. Note that $\T$ only depends on $\delta$ through $\delta y^2$, and it is straightforward to check that $\psi_1$, $\psi_2$ and $\psi_3$ do not depend on $\delta$ for such $\T$. Hence, there is no solution to $\psi_1(\mu)=\frac{\delta}{\delta-1}$ for any $\delta\in[1,\infty)$. According to Lemma \ref{Lem:PT_second} and Claim \ref{Claim:correlation}, we have $P^2_T(\hat{\bm{x}},\bm{x}_\star)\to0$. The situation for $\T_\opt$ is shown in the figure on the right panel of Fig.~\ref{Fig:psi_functions}.

\subsection{Discussions}
The main finding of this paper (namely, Claim \ref{Claim:correlation}) is a characterization of the asymptotic behavior of the spectral method. As will be explained in Section \ref{Sec:GEC}, the intuition for Claim \ref{Claim:correlation} is that the fixed point of an expectation-propagation type algorithm (called \Alg in this paper) coincides with an eigenvector with a particular eigenvalue (see Lemma \ref{Lem:stationary}) of the data matrix which the spectral estimator uses. Furthermore, the asymptotic behavior \Alg could be conveniently analyzed via the state evolution (SE) platform. This motivates us to study the spectral estimator by studying the fixed point of \textsf{PCA-EP}. The difficulty here is to prove that the stationary points of the \textsf{PCA-EP} converge to the eigenvector corresponding to the largest eigenvalue. Currently, we do not have a rigorous proof for Claim \ref{Claim:correlation}.

We would like to mention two closely-related work \cite{Bayati&Montanari12,Montanari2016} where the approximate message passing (AMP) algorithm and the state evolution formalism was employed to study the high-dimensional asymptotics of the LASSO problem and the nonnegative PCA problem respectively. Unfortunately, neither approach is applicable to our problem. In \cite{Bayati&Montanari12}, the convexity of the LASSO problem is crucial to prove that the AMP estimate converges to the minimum of the LASSO problem. In contrast, the variational formulation of the eigenvalue problem is nonconvex. In \cite{Montanari2016}, the authors used a Gaussian comparison inequalities to upper bound the cosine similarity between the signal vector and the PCA solution and use the AMP estimates to lower bound the same quantity, and combining the lower and upper bound yields the desired result. For our problem, however, the sensing matrix is not Gaussian and the Gaussian comparison inequality is not immediately applicable.

Although we do not have a rigorous proof for Claim \ref{Claim:correlation} yet, we will present some heuristic arguments in Section \ref{Sec:heuristics}. Our main heuristic arguments for Claim \ref{Claim:correlation} is Lemma \ref{Lem:norm} and Lemma \ref{Lem:leading_eigen}, which establishes a connection between the extreme eigenvalues of the original matrix $\bm{D}$ (see \eqref{Eqn:data_matrix}) and the matrix $\bm{E}(\hat{\mu})$ (see definition in \eqref{Eqn:Alg_orth_final_a}) that originates from the \textsf{PCA-EP} algorithm. These two lemmas provide some interesting observations that might shed light on Claim \ref{Claim:correlation}. The detailed discussions can be found in Section \ref{Sec:heuristics}.

\section{Optimal Design of $\T$}

Claim 1 characterizes the condition under which the leading eigenvector $\hat{\bm{x}}$ is positively correlated with the signal vector $\bm{x}_{\star}$ for a given $\T$. We would like to understand how does $\T$ affect the performance of the spectral estimator. We first introduce the following definitions:
\BE\label{Eqn:delta_T}
\delta_{\T}\Mydef \inf\{\delta:\rho_\T^2(\delta)>0\},
\EE
and
\BE\label{Eqn:delta_weak}
\delta_{\mr{weak}}\Mydef \inf_{\T}\ \delta_{\T}.
\EE
Following \cite{Mondelli2017}, the latter threshold is referred to as the \textit{weak threshold}. We will answer the following questions:
\begin{itemize}
\item[(Q.1)] What is $\delta_{\mr{weak}}$ for the spectral method under a partial orthogonal model?
\item[(Q.2)] For a given $\delta>\delta_{\mr{weak}}$, what is the optimal $\T$ that maximizes $P^2_\T(\hat{\bm{x}},\bm{x}_{\star})$?
\end{itemize}

These questions have been addressed for an i.i.d. Gaussian measurement model. In \cite{Mondelli2017}, Mondelli and Montanari proved that $\delta_{\mr{weak}}=1$ for this model. They also proposed a function (i.e., $\T_{\mr{MM}}$ in \eqref{Eqn:T_examples}) that achieves the weak threshold. Very recently, \cite{Lu2018} proved that $\T_{\opt}$ given in \eqref{Eqn:T_examples} maximizes $P^2(\hat{\bm{x}},\bm{x}_{\star})$ for any fixed $\delta>\delta_{\mr{weak}}$, and is therefore \textit{uniformly optimal}. For the partial orthogonal measurement model, the above problems could also be addressed based on our asymptotic characterization of the spectral initialization. It it perhaps surprising that the same function is also uniformly optimal under the partial orthogonal sensing model considered in this paper. Note that although $\T_\star$ is optimal for both the i.i.d. Gaussian and the partial orthogonal models, the performances of the spectral method are different: for the former model $\delta_{\mr{weak}}=1$, whereas for the latter model $\delta_{\mr{weak}}=2$. Theorem \ref{Lem:optimality} below, whose proof can be found in Appendix \ref{App:optimality}, summarizes our findings.

\begin{theorem}[Optimality of $\T_\opt$]\label{Lem:optimality}
Suppose that Claim 1 is correct, then $\delta_{\mr{weak}}=2$ for the partial orthogonal measurement model. Further, for any $\delta>2$, $\T_\opt$ in \eqref{Eqn:T_examples} maximizes $\rho_{\T}^2(\delta)$ and is given by
\[
\rho_{\star}^2(\delta)=
\begin{cases}
0, &  \text{if }\delta< 2,\\
\frac{1-\hat{\mu}_\star}{1-\frac{1}{\delta}\hat{\mu}_\star}, & \text{if }\delta\ge2,
\end{cases}
\]
where $\hat{\mu}_\star$ is the unique solution to $\psi_1(\mu)=\frac{\delta}{\delta-1}$ (with $\T=\T_\opt$). Finally, $\rho_{\star}^2(\delta)$ is an increasing function of $\delta$.
\end{theorem}

\begin{remark}
The function that can achieve the weak threshold is not unique. For instance, the following function also achieves the weak threshold:
\[
\T(y)=1-\frac{1}{\delta y^2+\delta-2}.
\]
It is straightforward to show that $\psi_1(1)=\frac{\delta}{\delta-1}$ for any $\delta\ge2$. Further, some straightforward calculations show that
\[
\rho^2_\T(1,\delta)=\frac{\delta^2-2\delta}{\delta^2-2}\in(0,1),\quad \forall \delta>2.
\]
Hence, by Lemma \ref{Lem:PT_second} and Claim \ref{Claim:correlation}, the cosine similarity is strictly positive for any $\delta>2$.

\end{remark}

\section{An EP-based Algorithm for the Spectral Method}\label{Sec:GEC}
In this section, we present the rationale behind Claim \ref{Claim:correlation}. Our reasoning is based on analyzing an expectation propagation (EP) \cite{Minka2001,opper2005,fletcher2016,He2017} based message passing algorithm that aims to find the leading eigenvector of $\bm{D}$. Since our algorithm intends to solve the eigenvector problem, we will call it \Alg throughout this paper. The key to our analysis is a state evolution (SE) tool for such EP-based algorithm,  first conjectured in \cite{ma2015turbo,liu2016} (for a partial DFT sensing matrix) and \cite{Ma2016} (for generic unitarily-invariant sensing matrices), and later proved in \cite{Rangan17,Takeuchi2017} (for general unitarily-invariant matrices). For approximate message passing (AMP)  algorithms \cite{DoMaMo09}, state evolution (SE) has proved to be a powerful tool for performance analysis in the high dimensional limit \cite{Bayati&Montanari12,Montanari_pca,Weng2016,Zheng17}.

\subsection{The \Alg algorithm}\label{Sec:GEC_algorithm}
The leading eigenvector of $\bm{A}^{\UH}\bm{T}\bm{A}$ is the solution to the following optimization problem:
\BE\label{Eqn:eigen_opt}
\underset{\|\bm{x}\|=\sqrt{n}}{\max}\ \bm{x}^{\UH}\bm{A}^{\UH}\bm{T}\bm{A}\bm{x},
\EE
where $\bm{T}\Mydef \mr{diag}\{\T(y_1),\ldots,\T(y_m)\}$. The normalization $\|\bm{x}\|=\sqrt{n}$ (instead of the commonly-used constraint $\|\bm{x}\|=1$) is introduced for convenience.
\Alg is an EP-type message passing algorithm that aims to solve \eqref{Eqn:eigen_opt}. Starting from an initial estimate $\bm{z}^0$, \Alg proceeds as follows  (for $t\ge1$):
\BS\label{Eqn:Alg_orth_final}
\BE\label{Eqn:Alg_orth_final_a}
\Alg:\qquad\bm{z}^{t+1}=\underbrace{\left(\delta\bm{AA}^{\UH}-\bm{I}\right)\left(\frac{ \bm{G}}{\G}-\bm{I}\right)}_{\bm{E}(\mu)}\bm{z}^t,
\EE
where $\bm{G}$ is a diagonal matrix defined as
\BE\label{Eqn:G_def_first}
\bm{G}\Mydef \mr{diag}\left\{\frac{1}{\mu^{-1}-\T(y_1)},\ldots,\frac{1}{\mu^{-1}-\T(y_m)}\right\},
\EE
and
\BE
\G\Mydef \frac{1}{m}\Tr(\bm{G}).
\EE
In \eqref{Eqn:G_def_first}, $\mu\in(0,1]$ is a parameter that can be tuned.
At every iteration of \Alg, the estimate for the leading eigenvector is given by
\BE\label{Eqn:x_hat_def}
{\bm{x}}^{t+1}=\frac{\sqrt{n}}{\|\bm{A}^{\UH}\bm{z}^{t+1}\|}\cdot\bm{A}^{\UH}\bm{z}^{t+1}.
\EE
\ES
The derivations of \Alg can be found in Appendix \ref{App:derivations}. Before we proceed, we would like to mention a couple of points:
\begin{itemize}

\item The \Alg algorithm is a tool for performance analysis, not an actual numerical method for finding the leading eigenvector. Further, our analysis is purely based on the iteration defined in \eqref{Eqn:Alg_orth_final}, and the heuristic behind \eqref{Eqn:Alg_orth_final} is irrelevant for our analysis.

\item The \Alg iteration has a parameter: $\mu\in(0,1]$, which does not change across iterations. To calibrate \Alg with the eigenvector problem, we need to choose the value of $\mu$ carefully. We will discuss the details in Section \ref{Sec:GEC}.

\item From this representation, \eqref{Eqn:Alg_orth_final_a} can be viewed as a power method applied to $\bm{E}(\mu)$.
This observation is crucial for our analysis. We used the notation $\bm{E}(\mu)$ to emphasize the impact of $\mu$.

\item The two matrices involved in $\bm{E}(\mu)$ satisfy the following ``zero-trace'' property (referred to as ``divergence-free'' in \cite{Ma2016}):
\BE
\frac{1}{m}\Tr\left(\frac{ \bm{G}}{\G}-
\bm{I}\right)=0\quad\text{and}\quad
\frac{1}{m}\Tr(\delta \bm{AA}^{\UH}-\bm{I})=0.
\EE
This zero-trace property is the key to the correctness of the state evolution (SE) characterization.

\end{itemize}

\subsection{State evolution analysis}\label{Sec:SE}

State evolution (SE) was first introduced in \cite{DoMaMo09,Bayati&Montanari11} to analyze the dynamic behavior of AMP. However, the original SE technique for AMP only works when the sensing matrix $\bm{A}$ has i.i.d. entries. Similar to AMP, \Alg can also be described by certain SE recursions, but the SE for \Alg works for a wider class of sensing matrices (specifically, unitarily-invariant $\bm{A}$ \cite{Ma2016,Rangan17,Takeuchi2017}) that include the random partial orthogonal matrix considered in this paper.

Denote $\bm{z}_\star=\bm{Ax}_\star$. Assume that the initial estimate $\bm{z}^0\overset{d}{=}\alpha_0\bm{z}_\star+\sigma_0\bm{w}^0$, where $\bm{w}^0\sim\mathcal{CN}(\mathbf{0},1/\delta\bm{I})$ is independent of $\bm{z}_\star$. Then, intuitively speaking, \Alg has an appealing property that $\bm{z}^{t+1}$ in \eqref{Eqn:Alg_orth_final_a} is approximately
\BE\label{Eqn:z_AWGN_first}
\bm{z}^{t+1} \approx \alpha_z^{t+1} \bm{z}_{\star}+\sigma^{t+1}\bm{w}^{t+1},
\EE
where $\alpha_t$ and $\sigma_t$ are the variables defined in \eqref{Eqn:SE_final}, and $\bm{w}^{t+1}$ is an iid Gaussian vector. Due to this property, the distribution of $\bm{z}^t$ is fully specified by $\alpha_t$ and $\sigma_t$. Further, for a given $\T$ and a fixed value of $\delta$, the sequences $\{\alpha_t\}_{t\ge1}$ and $\{\sigma^2_t\}_{t\ge1}$ can be calculated recursively from the following two-dimensional map:
\BS\label{Eqn:SE_final}
\begin{align}
\alpha_{t+1} &= (\delta-1)\cdot\alpha_t\cdot \big( \psi_1(\mu)-1 \big),\label{Eqn:SE_final_a}\\
\sigma^2_{t+1} &= (\delta-1)\cdot\left[|\alpha_t|^2\cdot\big( \psi_3^2(\mu)-\psi_1^2(\mu)\big)+\sigma_t^2\cdot\big( \psi_2(\mu)-1 \big)\right],\label{Eqn:SE_final_b}
\end{align}
\ES
where the functions $\psi_1$, $\psi_2$ and $\psi_3$ are defined in \eqref{Eqn:psi_definitions}. To gain some intuition on the SE, we present a heuristic way of deriving the SE in Appendix \ref{App:SE_heuristics}. The conditioning technique developed in \cite{Rangan17,Takeuchi2017} is promising to prove Claim \ref{Lem:SE}, but we need to generalize the results to handle complex-valued nonlinear model with possibly non-continuous $\T$. Given the fact that the connection between \textsf{PCA-EP} and the spectral estimator is already non-rigorous, we did not make such an effort.
\begin{claim}[SE prediction]\label{Lem:SE}
Consider the \Alg algorithm in \eqref{Eqn:Alg_orth_final}.
Assume that the initial estimate $\bm{z}^0\overset{d}{=}\alpha_0\bm{z}_\star+\sigma_0\bm{w}^0$ where $\bm{w}^0\sim\mathcal{CN}(\mathbf{0},1/\delta\bm{I})$ is independent of $\bm{z}_\star$.
Then, almost surely, the following hold for $t\ge1$:
\[
\lim_{n\to\infty}\frac{\langle\bm{z}_{\star},\bm{z}^t\rangle}{\|\bm{z}_{\star}\|^2} =\alpha_t\quad \text{and}\quad \lim_{n\to\infty}\frac{\|\bm{z}^t-\alpha_t\bm{z}_{\star}\|^2}{\|\bm{z}_{\star}\|^2} =\sigma^2_t,
\]
where $\alpha_t$ and $\sigma^2_t$ are defined in \eqref{Eqn:SE_final}. Furthermore, almost surely we have
\BE\label{Eqn:rho_SE}
\begin{split}
\lim_{n\to\infty}P_\T^2(\bm{x}^{t},\bm{x}) &=\lim_{n\to\infty}\frac{|\langle\bm{x}_{\star},\bm{x}^{t}\rangle|^2}{\|\bm{x}_{\star}\|^2\|\bm{x}_{t}\|^2}=\frac{|\alpha_t|^2}{|\alpha_t|^2+\frac{\delta-1}{\delta}\cdot \sigma_t^2},
\end{split}
\EE
where $\bm{x}^{t}$ is defined in \eqref{Eqn:x_hat_def}.
\end{claim}

\subsection{Connection between \Alg and the PCA problem}
\label{Sec:rationale}
Lemma \ref{Lem:stationary} below shows that any nonzero stationary point of \Alg in \eqref{Eqn:Alg_orth_final} is an eigenvector of the matrix $\bm{D}=\bm{A}^{\UH}\bm{TA}$. This is our motivation for analyzing the performance of the spectral method through \textsf{PCA-EP}.
\begin{lemma}\label{Lem:stationary}
Consider the \Alg algorithm in \eqref{Eqn:Alg_orth_final} with $\mu\in(0,1]$. Let $\bm{z}^{\infty}$ be an arbitrary stationary point of \Alg. Suppose that $\bm{A}^{\UH}\bm{z}^{\infty}\neq\textbf{0}$, then $(\bm{x}^{\infty},\lambda(\mu))$ is an eigen-pair of $\bm{D}=\bm{A}^{\UH}\bm{TA}$, where the eigenvalue $\lambda(\mu)$ is given by
\BE\label{Eqn:eigen_stationary}
\lambda(\mu)=\frac{1}{\mu}-\frac{\delta-1}{\delta}\frac{1}{\G}.
\EE
\end{lemma}
\begin{proof}
We introduce the following auxiliary variable:
\BE\label{Eqn:p_def}
\bm{p}^t\Mydef \left(\frac{\bm{G}}{\G}-\bm{I}\right)\bm{z}^t.
\EE
When $\bm{z}^{\infty}$ is a stationary point of \eqref{Eqn:Alg_orth_final_a}, we have
\BE\label{Eqn:GEC_fixed_orth}
\begin{split}
\bm{p}^{\infty}&=\left(\frac{1}{\G}(1/\mu\bm{I}-\bm{T})^{-1}-\bm{I}\right)(\delta\bm{AA}^{\UH}-\bm{I})\bm{p}^{\infty}.
\end{split}
\EE
Rearranging terms, we have
\[
\begin{split}
&\left(\frac{1}{\G}(1/\mu\bm{I}-\bm{T})^{-1}\right)\bm{p}^{\infty}\\
&=\left(\frac{1}{\G}(1/\mu\bm{I}-\bm{T})^{-1}-\bm{I}\right)\delta\bm{AA}^{\UH}\bm{p}^{\infty}.
\end{split}
\]
Multiplying $\bm{A}^{\UH}\G(1/\mu\bm{I}-\bm{T})$ from both sides of the above equation yields
\[
\begin{split}
\bm{A}^{\UH}\bm{p}^{\infty}&=\bm{A}^{\UH}\left(\bm{I}-\G(1/\mu\bm{I}-\bm{T})\right)\delta\bm{AA}^{\UH}\bm{p}^{\infty}\\
&=\delta(1-\mu^{-1}\G)\bm{A}^{\UH}\bm{p}^{\infty}+\delta \G\cdot\bm{A}^{\UH}\bm{TAA}^{\UH}\bm{p}^{\infty}.
\end{split}
\]
After simple calculations, we finally obtain
\BE\label{Eqn:GEC_fixed_eigen}
\bm{A}^{\UH}\bm{T}\bm{A}\left(\bm{A}^{\UH}\bm{p}^{\infty}\right)=\left(\frac{1}{\mu}-\frac{\delta-1}{\delta}\frac{1}{\G}\right)\left(\bm{A}^{\UH}\bm{p}^{\infty}\right).
\EE
In the above, we have identified an eigen-pair for the matrix $\bm{A}^{\UH}\bm{T}\bm{A}$. To complete our proof, we note from \eqref{Eqn:Alg_orth_final_a} and \eqref{Eqn:p_def} that $\bm{z}^{\infty}=(\delta\bm{AA}^{\UH}-\bm{I})\bm{p}^{\infty}$ and so
\[
\begin{split}
\bm{A}^{\UH}\bm{z}^{\infty}&=\bm{A}^{\UH}(\delta\bm{AA}^{\UH}-\bm{I})\bm{p}^{\infty}\\
&=(\delta-1)\bm{A}^{\UH}\bm{p}^{\infty}.
\end{split}
\]
Hence,
\BE
{\bm{x}}^{\infty}\Mydef\frac{\sqrt{n}}{\|\bm{A}^{\UH}\bm{z}^{\infty}\|}\cdot\bm{A}^{\UH}\bm{z}^{\infty}=\frac{\sqrt{n}}{\|\bm{A}^{\UH}\bm{p}^{\infty}\|}\cdot\bm{A}^{\UH}\bm{p}^{\infty}
\EE
is also an eigenvector.
\end{proof}

\begin{remark}\label{Rem:empri_div}
Notice that the eigenvalue identified in \eqref{Eqn:eigen_stationary} is closely related to the function $\Lambda(\cdot)$ in \eqref{Eqn:Lambda_def}. In fact, the only difference between \eqref{Eqn:eigen_stationary} and \eqref{Eqn:Lambda_def} is that the normalized trace $\G$ in \eqref{Eqn:eigen_stationary} is replaced by $\mathbb{E}[G(|Z_\star|,\mu)]$, where $Z_\star\sim\mathcal{CN}(0,1/\delta)$. Under certain mild regularity conditions on $\T$, it is straightforward to use the weak law of large numbers to prove $\G\overset{p}{\to}\mathbb{E}[G(|Z_\star|,\mu)]$.
\end{remark}

Lemma \ref{Lem:stationary} shows that the stationary points of the \Alg algorithm are eigenvectors of $\bm{A}^{\UH}\bm{TA}$. Since the asymptotic performance of \Alg can be characterized via the SE platform, it is conceivable that the fixed points of the SE describe the asymptotic performance of the spectral estimator. However, the answers to the following questions are still unclear:
\begin{itemize}
\item Even though $\bm{x}^{\infty}$ is an eigenvector of $\bm{A}^{\UH}\bm{TA}$, does it correspond to the largest eigenvalue?
\item The eigenvalue in \eqref{Eqn:eigen_stationary} depends on $\mu$, which looks like a free parameter. How should $\mu$ be chosen?
\end{itemize}
In the following section, we will discuss these issues and provide some heuristic arguments for Claim \ref{Claim:correlation}.

\subsection{Heuristics about Claim \ref{Claim:correlation}}\label{Sec:heuristics}

The SE equations in \eqref{Eqn:SE_final} have two sets of solutions (in terms of $\alpha$, $\sigma^2$ and $\mu$):
\BS\label{Eqn:SE_fixed}
\begin{alignat}{3}
&\text{Uninformative solution:}\quad &&\alpha=0,\ \sigma^2\neq0, &\qquad\\
& &&\psi_2(\mu)=\frac{\delta}{\delta-1},&\label{Eqn:SE_uninformative}
\intertext{and}
&\text{Informative solution:}\quad &&\frac{|\alpha|^2}{\sigma^2}=\frac{\frac{\delta}{\delta-1}-\psi_2({\mu})}{\psi_3^2({\mu})-\left(\frac{\delta}{\delta-1}\right)^2},&\qquad\label{Eqn:SE_informative}\\
& &&\psi_1(\mu)=\frac{\delta}{\delta-1}.&
\end{alignat}
\ES

\begin{remark}
 Both solutions in \eqref{Eqn:SE_fixed} do not have a constraint on the norm of the estimate (which corresponds to a constraint on $|\alpha|^2+\sigma^2$). This is because we ignored the norm constraint in deriving our \Alg algorithm in Appendix \ref{App:derivations_a}. If our derivations explicitly take into account the norm constraint, we would get an additional equation that can specify the individual values of $|\alpha|$ and $\sigma$. However, only the ratio $|\alpha|^2/\sigma^2$ matters for our phase transition analysis. Hence, we have ignored the constraint on the norm of the estimate.
\end{remark}

\begin{remark}
 $\alpha=\sigma=0$ is also a valid fixed point to \eqref{Eqn:SE_final}. However, this solution corresponds to the all-zero vector and is not of interest for our purpose.
\end{remark}

For the uninformative solution, we have $\alpha=0$ and hence $P_\T^2(\bm{x}^{\infty},\bm{x}_\star)\to0$ according to \eqref{Eqn:rho_SE}. On the other hand, the informative solution (if exists) corresponds to an estimate that is positively correlated with the signal, namely, $P_\T^2({\bm{x}^{\infty}},\bm{x}_\star)>0$. Recall from \eqref{Eqn:Alg_orth_final} that $\mu\in(0,1]$ is a parameter of the \Alg algorithm. Claim \ref{Claim:correlation} is obtained based on the following heuristic argument: \textit{$P_\T^2(\hat{\bm{x}},\bm{x}_\star)>0$ if and only if there exists a $\hat{\mu}\in(0,1]$ so that the informative solution is valid}. More specifically, there exists a $\hat{\mu}$ so that the following conditions hold:
\BE\label{Eqn:heuristic_rho2}
\frac{|\alpha|^2}{\sigma^2}=\frac{\frac{\delta}{\delta-1}-\psi_2({\hat{\mu}})}{\psi_3^2({\mu})-\left(\frac{\delta}{\delta-1}\right)^2}>0\quad\text{and}\quad \psi_1(\hat{\mu})=\frac{\delta}{\delta-1}.
\EE
{Further, we have proved that (see Lemma \ref{Lem:PT_second}) the above condition is equivalent to the phase transition presented in Claim \ref{Claim:correlation}.}
Note that $\sigma^2\ge0$ (and so $|\alpha|^2/\sigma^2\ge0$) is an implicit condition in our derivations of the SE. Hence, for a valid informative solution, we must have $|\alpha|^2/\sigma^2>0$.

We now provide some heuristic rationales about Claim \ref{Claim:correlation}. Our first observation is that if the informative solution exists and we set $\mu$ in \Alg properly, then $\alpha_{t}$ and $\sigma_{t+1}^2$ neither tend to zero or infinity. A precise statement is given in Lemma \ref{Lem:norm} below. The proof is straightforward and hence omitted.

\begin{lemma}\label{Lem:norm}
Suppose that $\psi_1(\bar{\mu})>\frac{\delta}{\delta-1}$, where $\bar{\mu}$ is defined in \eqref{Eqn:mu_star_def}. Consider the \Alg algorithm with $\mu$ set to $\hat{\mu}$, where $\hat{\mu}\in(0,\bar{\mu}]$ is the unique solution to
\[
\psi_1(\mu)=\frac{\delta}{\delta-1}.
\]
Let $|\alpha_0|<\infty$ and $0<\sigma^2_0<\infty$. Then, $\{\alpha_t\}_{t\ge1}$ and $\{\sigma^2_t\}_{t\ge1}$ generated by \eqref{Eqn:SE_final} converge and
\[
0<|\alpha_\infty|^2+\sigma^2_\infty<\infty.
\]
\end{lemma}

From Claim \ref{Lem:SE}, we have almost surely (see Lemma \ref{Lem:norm})
\BE\label{Eqn:Lem_norm}
\lim_{t\to\infty}\lim_{m\to\infty}\frac{\|\bm{z}^t\|^2}{m} =\frac{|\alpha_\infty|^2+\sigma^2_\infty}{\delta}\in(0,\infty).
\EE
Namely, the norm $\|\bm{z}^t\|^2/m$ does not vanish or explode, under a particular limit order. Now, suppose that for a finite and deterministic problem instance we still have
\BE\label{Eqn:Z_norm_finite}
\lim_{t\to\infty}\frac{\|\bm{z}^t\|^2}{m}=C\in(0,\infty),
\EE
for some constant $C$. Recall that \Alg can be viewed as a power iteration applied to the matrix $\bm{E}(\hat{\mu})$ (see \eqref{Eqn:Alg_orth_final}). Hence, \eqref{Eqn:Z_norm_finite} implies
\[
|\lambda_1(\bm{E}(\hat{\mu}))|=1,
\]
where $|\lambda_1(\bm{E}(\hat{\mu}))|$ denotes the spectral radius of $\bm{E}(\hat{\mu})$. Based on these heuristic arguments, we conjecture that $\lim_{m\to\infty}|\lambda_1(\bm{E}(\hat{\mu})|=1$ almost surely.
Further, we show in Appendices \ref{App:w_corr} and \ref{App:Cov_convergence} that \Alg algorithm converges in this case. Clearly, this implies that $(1,\bm{z}^{\infty})$ is an eigen-pair of $\bm{E}(\hat{\mu})$, where $\bm{z}^{\infty}$ is the limit of the estimate $\bm{z}^t$. Combining the above arguments, we have the following conjecture
\BE\label{Eqn:D2_outlier}
\lim_{m\to\infty}\lambda_1(\bm{E}(\hat{\mu}))=1.
\EE

If \eqref{Eqn:D2_outlier} is indeed correct, then the following lemma shows that the largest eigenvalue of $\bm{D}=\bm{A}^{\UH}\bm{TA}$ is $\Lambda(\hat{\mu})$. A proof of the lemma can be found in Appendix \ref{App:eigen_max}.

\begin{lemma}\label{Lem:leading_eigen}
Consider the matrix $\bm{E}(\hat{\mu})$ defined in \eqref{Eqn:Alg_orth_final_a}, where $\hat{\mu}\in(0,1]$. Let $\lambda_1\left(\bm{E}(\hat{\mu})\right)$ be the eigenvalue of $\bm{E}(\hat{\mu})$ that has the largest magnitude and $\hat{\bm{z}}$ is the corresponding eigenvector. If $\lambda_1\left(\bm{E}(\hat{\mu})\right)=1$ and $\bm{A}^{\UH}\hat{\bm{z}}\neq\textbf{0}$, then
\BE\label{Eqn:lambda_D_heuristic}
\lambda_1(\bm{A}^{\UH}\bm{T}\bm{A})=\frac{1}{\hat{\mu}}-\frac{\delta-1}{\delta}\frac{1}{\G}.
\EE
\end{lemma}

Recall from Lemma \ref{Lem:stationary} that if \Alg converges, then the stationary estimate is an eigenvector of $\bm{D}=\bm{A}^{\UH}\bm{TA}$, but it is unclear whether it is the leading eigenvector. Our heuristic arguments in this section and Lemma \ref{Lem:leading_eigen} imply that \Alg indeed converges to the leading eigenvector of $\bm{D}$. Therefore, we conjecture that the fixed points of the SE characterize the asymptotic behavior of the spectral estimator.

\begin{figure*}[hbpt]
\centering
\includegraphics[width=.95\textwidth]{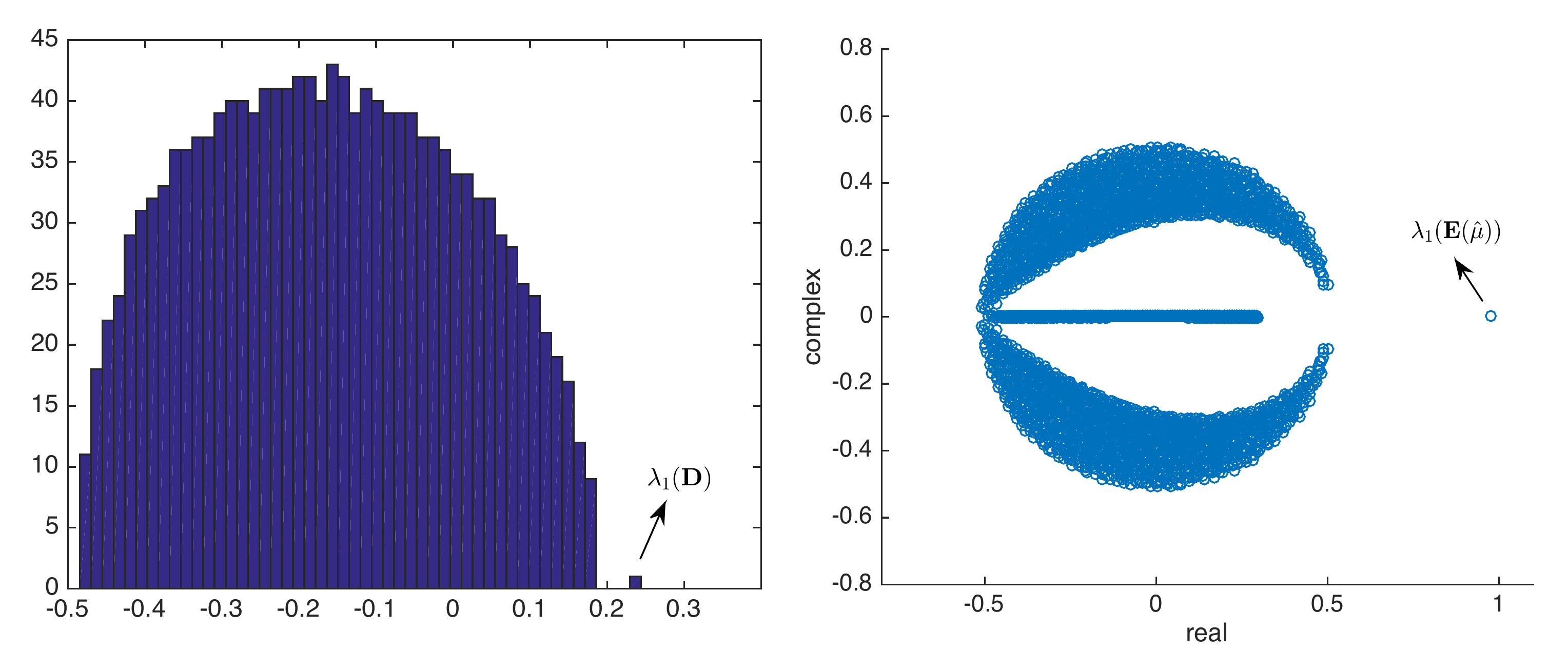}
\caption{\textbf{Left:} Histogram of the eigenvalues of $\bm{D}=\bm{A}^{\UH}\bm{TA}$. \textbf{Right:} Scatter plot of the eigenvalues of $\bm{E}(\hat{\mu})$. $\hat{\mu}$ is set to the unique solution to $\psi_1(\mu)=\frac{\delta}{\delta-1}$. $\T=\T_{\mr{MM}}$. $\delta=5$. $n=1500$. The signal $\bm{x}_\star$ is randomly generated from an i.i.d. zero Gaussian distribution.}\label{Fig:histograms}
\end{figure*}

\begin{figure*}[hbpt]
\centering
\subfloat{\includegraphics[width=.33\textwidth]{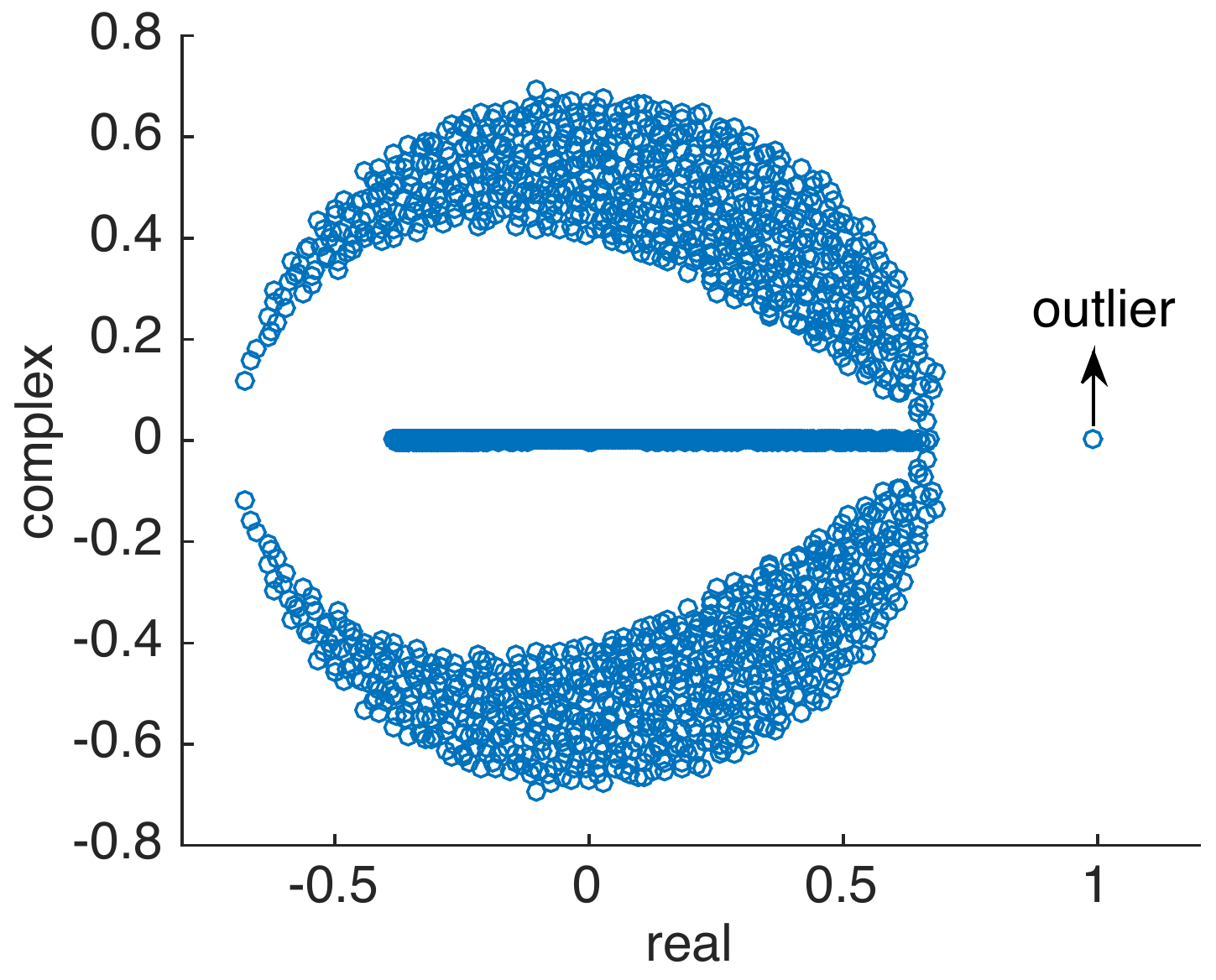}}
\hfil
\subfloat{\includegraphics[width=.33\textwidth]{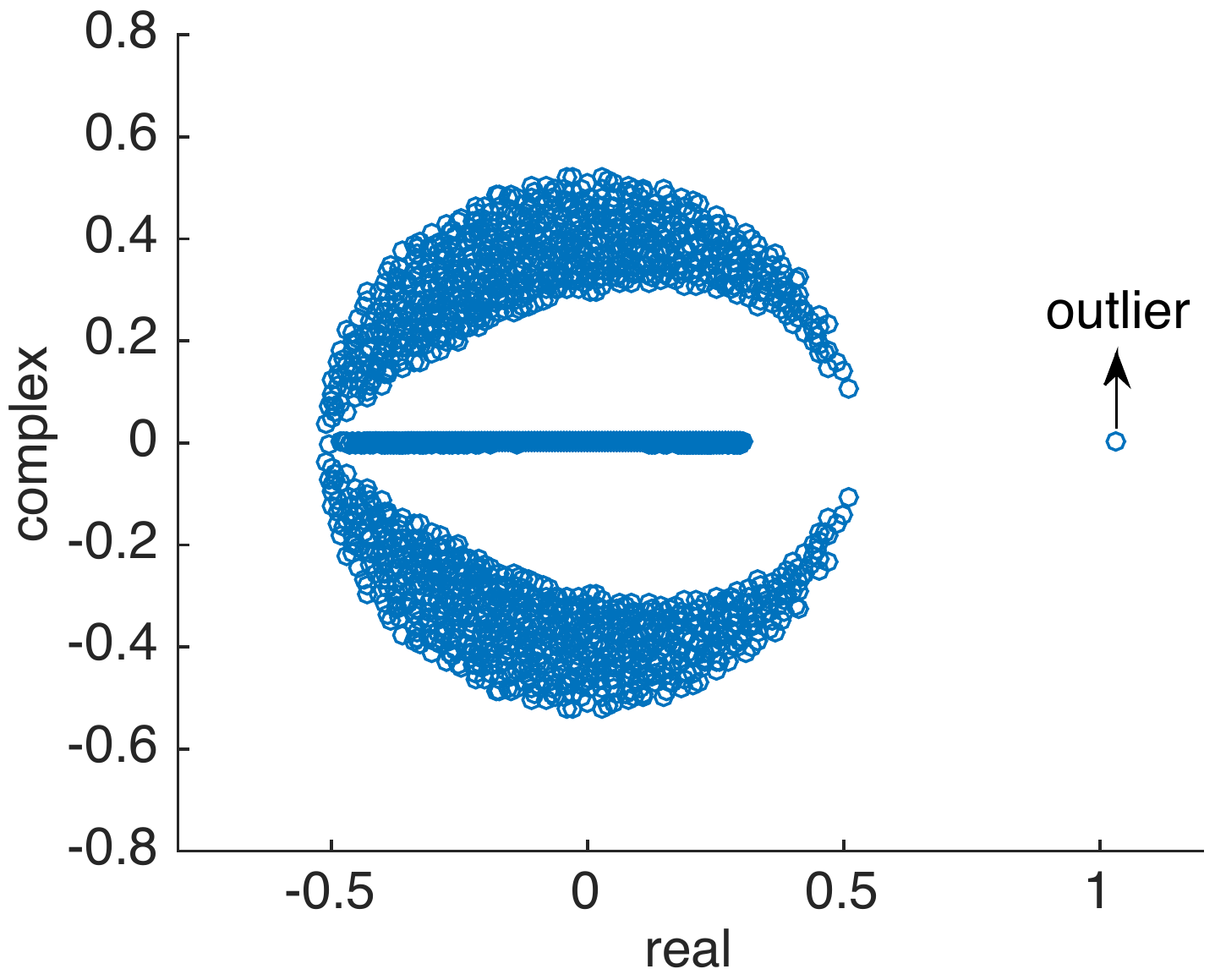}}
\hfil
\subfloat{\includegraphics[width=.33\textwidth]{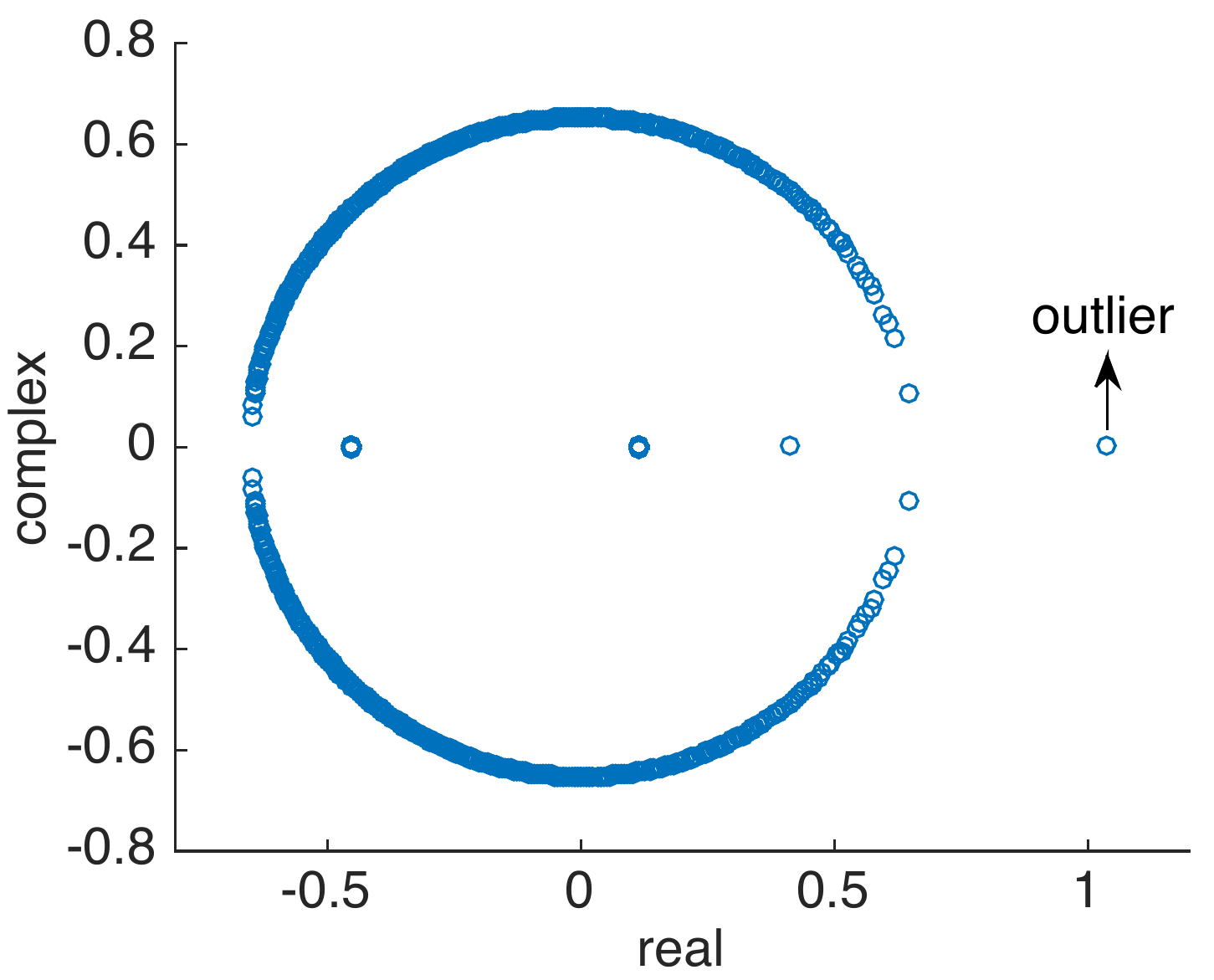}}
\caption{Scatter plot of the eigenvalues of $\bm{D}_2(\mu_1)$ in the complex plane. $\bm{x}_{\star}$ are sampled from an i.i.d. Gaussian distribution. $n=700$, $\delta=5$. \textbf{Left:} $\T_{\opt}$, \textbf{Middle:} $\T_{\mr{MM}}$, \textbf{Right:} $\T_{\mr{subset}}$ with $c_2=1.5$. }\label{Fig:scatter}
\end{figure*}

The above argument is not rigorous. However, our numerical results suggest that the conclusions are correct. An example is shown in Fig.~\ref{Fig:histograms}. Here, the figure on the right  plots the eigenvalues of a random realization of $\bm{E}(\hat{\mu})$ with $n=1500$. We can see that an outlying eigenvalue pops out of the bulk part of the spectrum. Also, the outlying eigenvalue is close to one. Fig. \ref{Fig:scatter} further plots the eigenvalues of $\bm{E}(\hat{\mu})$ for three other choices of $\T$. We see that all of the results seem to support our conjecture: there is an outlying eigenvalue (at one on the real line), although the shape of the bulk parts depends on the specific choice of $\T$.

\subsection{{On the domain of $\psi_1$, $\psi_2$, $\psi_3$, and $\Lambda$}}

In the previous discussions, we assumed that $\psi_1(\mu)$, $\psi_2(\mu)$, $\psi_3(\mu)$ and $\Lambda(\mu)$ are all defined on $(0,1]$. This assumption was implicitly used to derive the \Alg algorithm. Specifically, the Gaussian pdf for obtaining \eqref{Eqn:SEP_post_2} in the derivation of \Alg is not well defined if $\mu\notin(0,1]$. Nevertheless, the final \Alg algorithm in \eqref{Eqn:Alg_orth_final} is well defined, as long as $G(\mu)$ is well defined. We have assumed $\T(y)\le T_{\max}=1$. Let us further assume that $\T$ is bounded from below:
\[
T_{\min}\le\T(y)\le1,
\]
where $T_{\min}\in\mathbb{R}$. Under this assumption, $G(y,\mu)=\frac{1}{\mu^{-1}-\T(y)}$
is well defined as long as $\frac{1}{\mu}\notin (T_{\min},1)$. On the other hand, we only focused on the domain $\mu\in(0,1]$ (or $1/\mu\in[1,\infty)$) in our previous discussions.
In particular, we conjectured that $P_\T^2(\bm{x}_\star,\hat{\bm{x}})>0$ if and only if there exists an informative solution (see \eqref{Eqn:SE_fixed}) for $1/\mu\in[1,\infty)$. A natural question is what if the SE equations in \eqref{Eqn:SE_fixed} do not have informative solutions for  $1/\mu\in[1,\infty)$, but do have such a solution for $1/\mu\in(-\infty,T_{\min})$? To be specific, suppose that $1/\hat{\mu}\in(-\infty,T_{\min})$ satisfies the following conditions:
\[
\frac{\frac{\delta}{\delta-1}-\psi_2(\hat{\mu})}{\psi_3^2({\mu})-\left(\frac{\delta}{\delta-1}\right)^2}>0\quad\text{and}\quad \psi_1(\hat{\mu})=\frac{\delta}{\delta-1}.
\]
Then, one might ask what happens if we consider a \Alg algorithm by setting $\mu$ to such a $\hat{\mu}$?
It turns out that, based on arguments similar to those presented in Section \ref{Sec:heuristics}, we can obtain
\BE\label{Eqn:D_min}
\lambda_{n}(\bm{D})=\Lambda(\hat{\mu}),\quad 1/\hat{\mu}\in(-\infty,T_{\min}).
\EE
Namely, our method can provide a conjecture about the minimum eigenvalue of $\bm{D}$. To see this, we note that using exactly the same arguments as those in Section \ref{Sec:heuristics}, we claim (heuristically) that
\BE\label{Eqn:E1_largest}
\lim_{m\to\infty}\lambda_1(\bm{E}(\hat{\mu}))=1.
\EE
Lemma \ref{Lem:minimum_eigen} below further establishes a connection between the extremal eigenvalues of $\bm{D}$ and $\bm{E}(\hat{\mu})$. Its proof is postponed to Appendix \ref{App:eigen_min}.
\begin{lemma}\label{Lem:minimum_eigen}
Consider the matrix $\bm{E}(\hat{\mu})$ defined in \eqref{Eqn:Alg_orth_final_a}, where $1/\hat{\mu}\in(-\infty,T_{\min})$. Let $\lambda_1\left(\bm{E}(\hat{\mu})\right)$ be the eigenvalue of $\bm{E}(\hat{\mu})$ that has the largest magnitude and $\hat{\bm{z}}$ is the corresponding eigenvector. If $\lambda_1\left(\bm{E}(\hat{\mu})\right)=1$ and $\bm{A}^{\UH}\hat{\bm{z}}\neq\textbf{0}$, then
\BE\label{Eqn:lambda_D_heuristic}
\lambda_n(\bm{A}^{\UH}\bm{T}\bm{A})=\frac{1}{\hat{\mu}}-\frac{\delta-1}{\delta}\frac{1}{\G},
\EE
where $\lambda_n(\bm{D})$ denotes the minimum eigenvalue of $\bm{D}$.
\end{lemma}

A numerical example is shown in Fig.~\ref{Fig:histograms_min}. This figure is similar to Fig.~\ref{Fig:histograms}, but with the processing function replaced by $\T(y)=3-\T_{\mr{MM}}(y)$. Under this setting, $\psi_1(\hat{\mu})=\frac{\delta}{\delta-1}$ has a solution in the domain $1/\mu\in(-\infty,T_{\min}]$. Further, $\psi_1(\hat{\mu})>\psi_2(\hat{\mu})$. By the above heuristic arguments, we should have $\bm{E}(\hat{\mu})=1$ and $\lambda_n(\bm{D})=\Lambda(\hat{\mu})\approx0.71$. These conjectures about the extremal eigenvalues of $\bm{D}$ and $\bm{E}(\hat{\mu})$ seem to be close to the empirical eigenvalues given in Fig.~\ref{Fig:histograms_min}.

\begin{figure*}[hbpt]
\centering
\includegraphics[width=.95\textwidth]{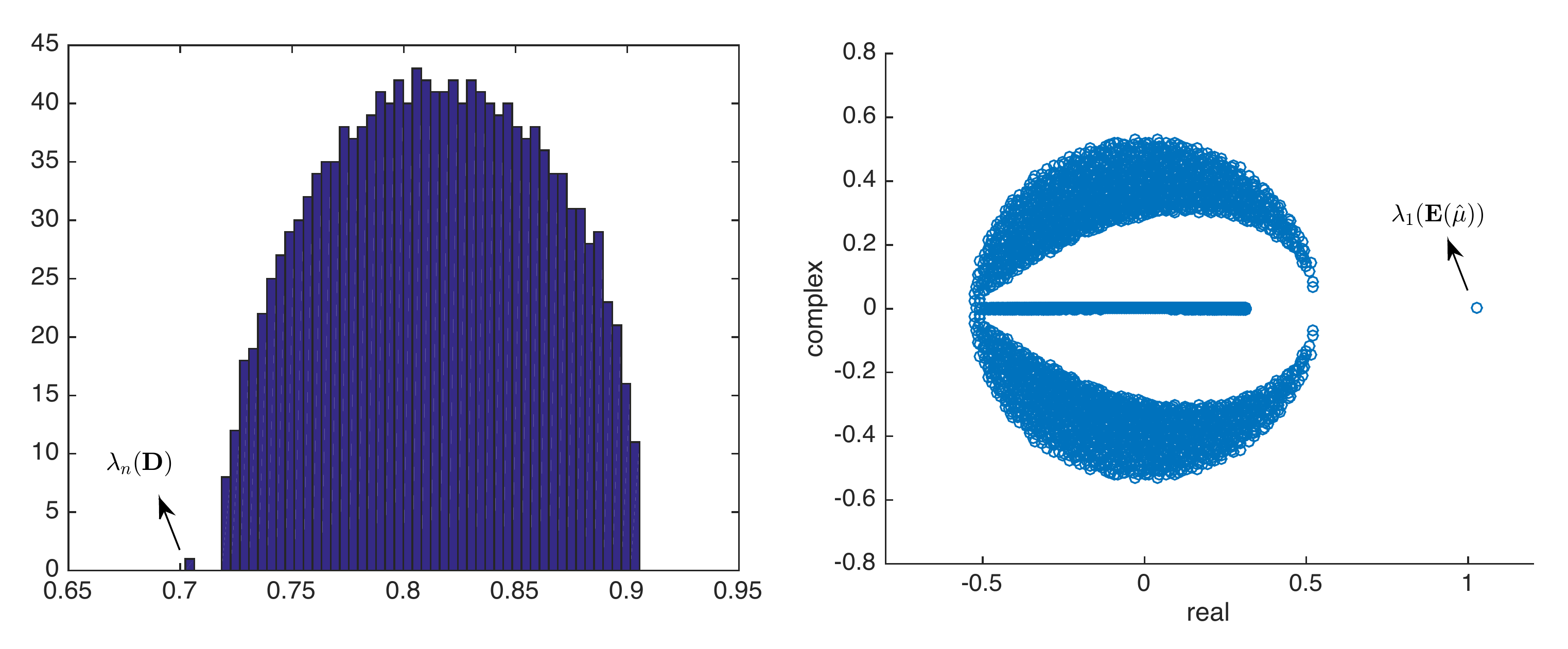}
\caption{An example where the smallest eigenvalue of $\bm{D}$ separates from the bulk spectrum. $\T(y)=3-\T_{\mr{MM}}(y)$. $\T(y)$ is further scaled by a positive constant to make it satisfy $\sup_{y\ge0}\T(y)=1$. \textbf{Left:} Histogram of the eigenvalues of $\bm{D}=\bm{A}^{\UH}\bm{TA}$. \textbf{Right:} Scatter plot of the eigenvalues of $\bm{E}(\hat{\mu})$, where $\hat{\mu}\approx 2.718$ is the solution to $\psi_1(\mu)=\frac{\delta}{\delta-1}$ in the domain $1/\mu\in(-\infty,T_{\min}]$. $n=1500$ and $\delta=5$.}\label{Fig:histograms_min}
\end{figure*}

\section{Numerical results}\label{Sec:numerical}

\subsection{Accuracy of our predictions}\label{Sec:nume_a}
We first provide some numerical results to verify the accuracy of the predictions in Claim \ref{Claim:correlation}. Following \cite{Mondelli2017}, our simulations are conducted using the image ($820  \times 1280$) shown in Fig. \ref{Fig:test_imag}. For ease of implementation, we reduced the size of the original image by a factor of $20$ (for each dimension). The length of the final signal vector is $2624$. We will compare the performances of the spectral method under the following models of $\bm{A}$:
\begin{itemize}
\item \textit{Random Haar model:} $\bm{A}$ is a subsampled Haar matrix;
\item \textit{Coded diffraction patterns (CDP):}
\[
\bm{A}=
\begin{bmatrix}
\bm{F}\bm{P}_1\\
\bm{F}\bm{P}_2\\
\ldots\\
\bm{F}\bm{P}_L
\end{bmatrix},
\]
where $\bm{F}$ is a two-dimensional DFT matrix and $\bm{P}_l=\mr{diag}\{e^{j\theta_{l,1}},\ldots,e^{j\theta_{l,n}}\}$ consists of i.i.d. uniformly random phases;
\item \textit{Partial DFT model:}
\[
\bm{A}=\bm{F}\bm{S}\bm{P},
\]
where, with slight abuse of notations, $\bm{F}\in\mathbb{C}^{m\times m}$ is now a unitary DFT matrix, and $\bm{S}\in\mathbb{R}^{m\times n}$ is a random selection matrix (which consists of randomly selected columns of the identity matrix), and finally $\bm{P}$ is a diagonal matrix comprised of i.i.d. random phases.
\end{itemize}

Fig.~\ref{Fig:emprical} plots the cosine similarity of the spectral method with various choices of $\T$. Here, the leading eigenvector is computed using a power method. In our simulations, we approximate ${\T}_{\star}$ by the function $1-\left(\delta y^2 + 0.01\right)^{-1}$. For $\T_{\mr{MM}}$ and $\T_\star$, the data matrix $\bm{A}^{\UH}\bm{TA}$ can have negative eigenvalues. To compute the largest eigenvector, we run the power method on the modified data matrix $\bm{A}^{\UH}\bm{TA}+\epsilon\bm{I}$ for a large enough $\epsilon$. In our simulations, we set $\epsilon$ to $10$ for $\T_{\mr{MM}}$ and $\epsilon$ to $50$ for $\T_\star$. The maximum number of power iterations is set to 10000. Finally, following \cite{Mondelli2017}, we measure the images from the three RGB color-bands using independent realizations of $\bm{A}$. For each of the three measurement vectors, we compute the spectral estimator $\hat{\bm{x}}$ and measure the cosine similarity $P_\T(\bm{x}_\star,\hat{\bm{x}})$ and then average the cosine similarity over the three color-bands. Finally, we further average the cosine similarity over 5 independent runs. Here, the lines show the simulation results and markers represent our predictions given in Claim \ref{Claim:correlation}.

From Fig.~\ref{Fig:emprical}, we see that the empirical cosine similarity between the spectral estimate and the signal vector match very well with our predictions, for both of the random Haar model and the two DFT matrix based models (i.e., the CDP model and the partial DFT model). Furthermore, the function $\T_\opt$ yields the best performance among the various choices of $\T$, which is consistent with our theory.

\begin{figure}[htbp]
\begin{center}
\includegraphics[width=.48\textwidth]{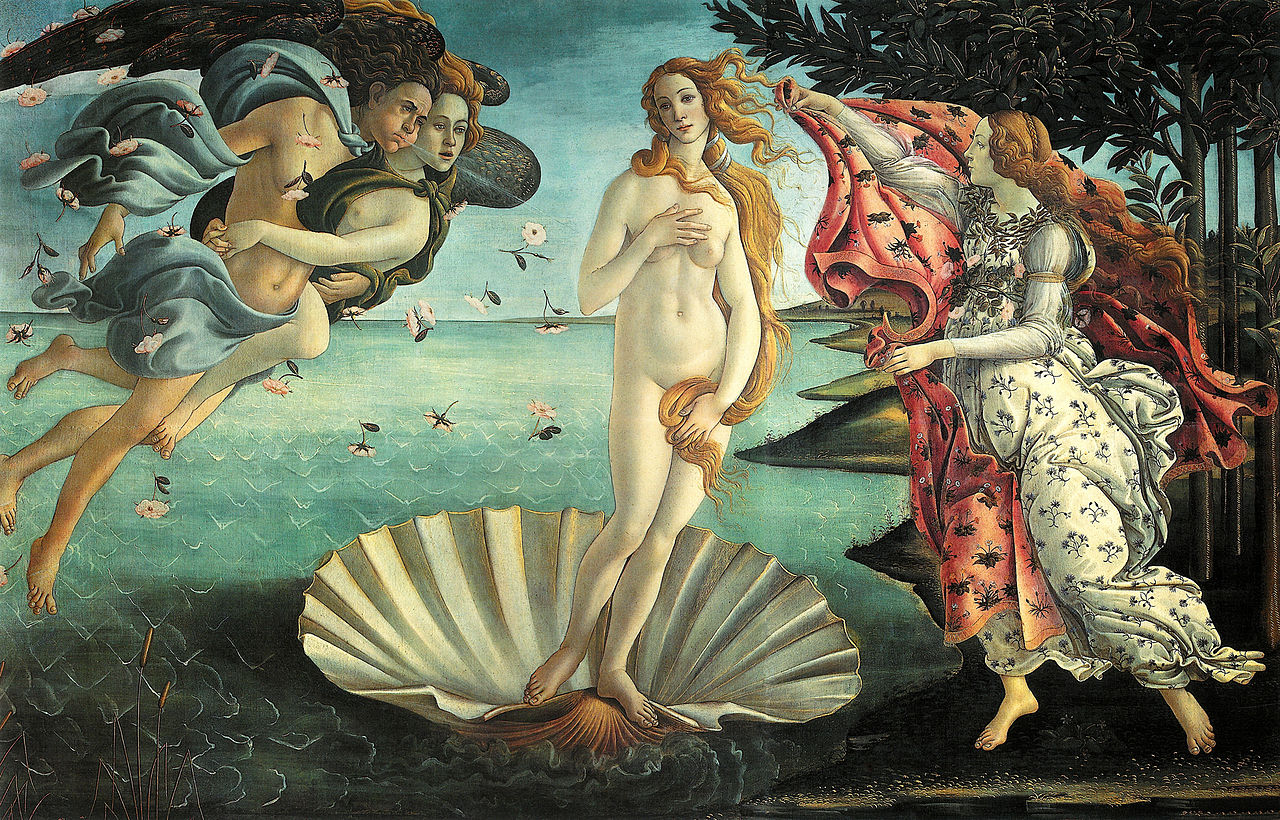}
\caption{The image from \cite{Mondelli2017}.}\label{Fig:test_imag}
\end{center}
\end{figure}

\begin{figure}[htbp]
\begin{center}
\includegraphics[width=.5\textwidth]{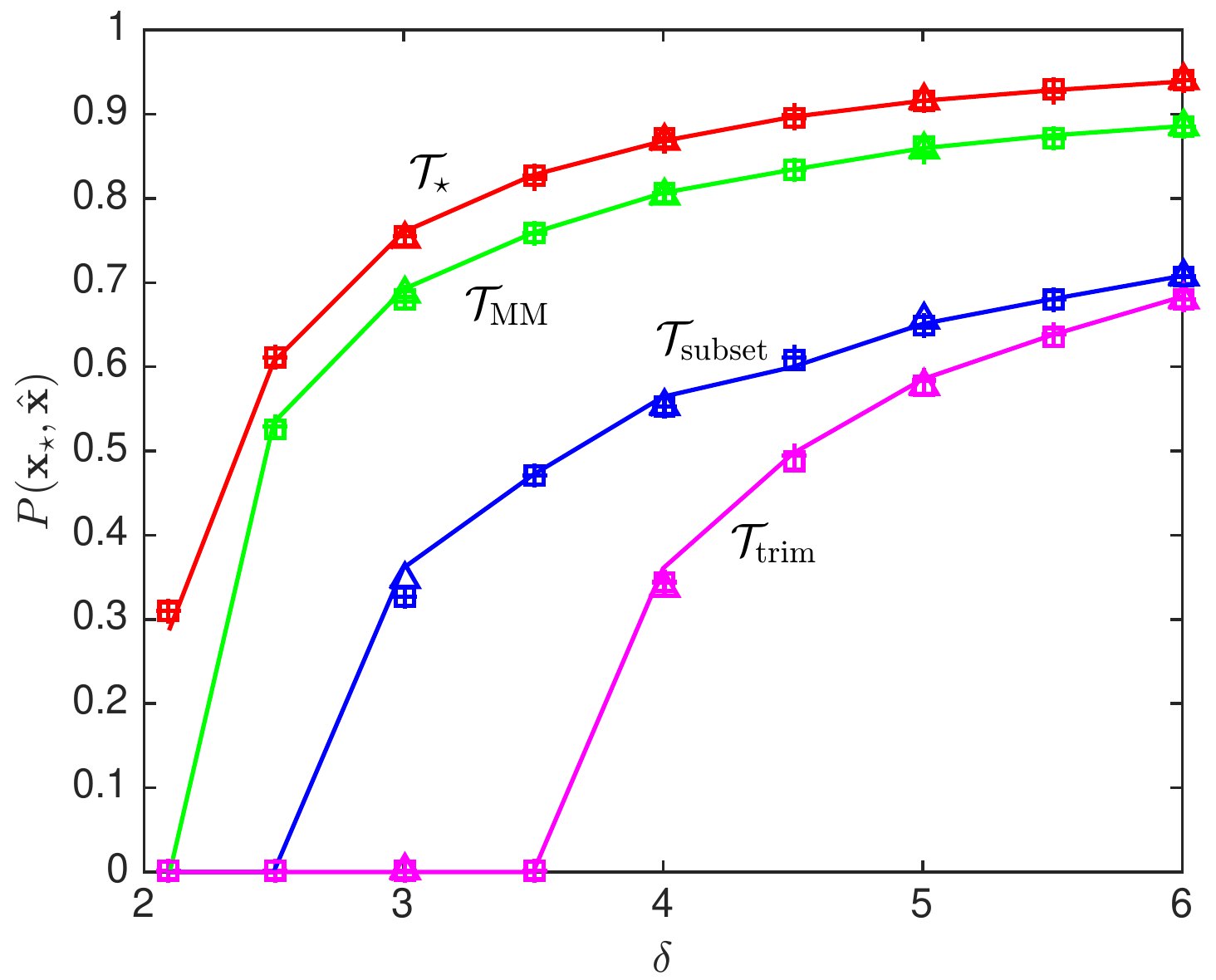}
\caption{Comparison of various spectral initialization methods. \textbf{Solid lines:} theoretical predictions. \textbf{Marker $+$: } simulation results for the partial DFT model. \textbf{Marker $\triangle$:} simulation results for the CDP model. \textbf{Marker $\square$:} simulation results for the Haar random matrix model. For simulations with the partial DFT model and the Haar model, the grid of $\delta$ is $[2.1, 2.5:0.5:6]$; for simulations with the CDP models, the grid of $\delta$ is $[3:1:6]$. The thresholds for $\T_{\mr{trim}}$ and $\T_{\mr{subset}}$ are set to $c_1=2$ and $c_2=1.5$ (under the normalization $\|\bm{x}_\star\|=\sqrt{n}$), respectively.}\label{Fig:emprical}
\end{center}
\end{figure}

\subsection{State evolution of \Alg}
Finally, we present some simulation results to show the accuracy of the state evolution characterization of the \Alg algorithm. In our simulations, we use the partial DFT matrix model introduced in Section \ref{Sec:nume_a}. The signal vector $\bm{x}_\star$ is randomly generated from an i.i.d. Gaussian distribution. The results are very similar when $\bm{x}_\star$ is replaced by the image shown in Fig.~\ref{Fig:test_imag}. We consider an \Alg algorithm with $\mu=\hat{\mu}$ where $\hat{\mu}$ is the solution to $\psi_1(\mu)=\frac{\delta}{\delta-1}$.

Fig.~\ref{Fig:SE} compares the empirical and theoretical predictions of two quantities: (i) the cosine similarity $P_\T(\bm{x}_\star,\bm{x}^t)$, where $\bm{x}^t$ is the estimate produced by \Alg (see definition in \eqref{Eqn:x_hat_def}) and $\bm{x}_\star$ is the true signal vector, and (ii) the cosine similarity between two consecutive ``noise terms'' (see \eqref{Eqn:z_AWGN_first}). The asymptotic prediction of the two quantities are given by \eqref{Eqn:rho_SE} and \eqref{Eqn:W_recursive} respectively. As can be see from Fig.~\ref{Fig:SE}, our theoretical predictions accurately characterized both quantities. Further, the correlation $P(\bm{w}^t,\bm{w}^{t+1})\to1$ as $t\to\infty$, as analyzed in Appendix \ref{App:Cov_convergence}. This implies that the estimate $\bm{x}^t$ converges; see discussions in Appendix \ref{App:w_corr}.

\begin{figure}[htbp]
\begin{center}
\includegraphics[width=.5\textwidth]{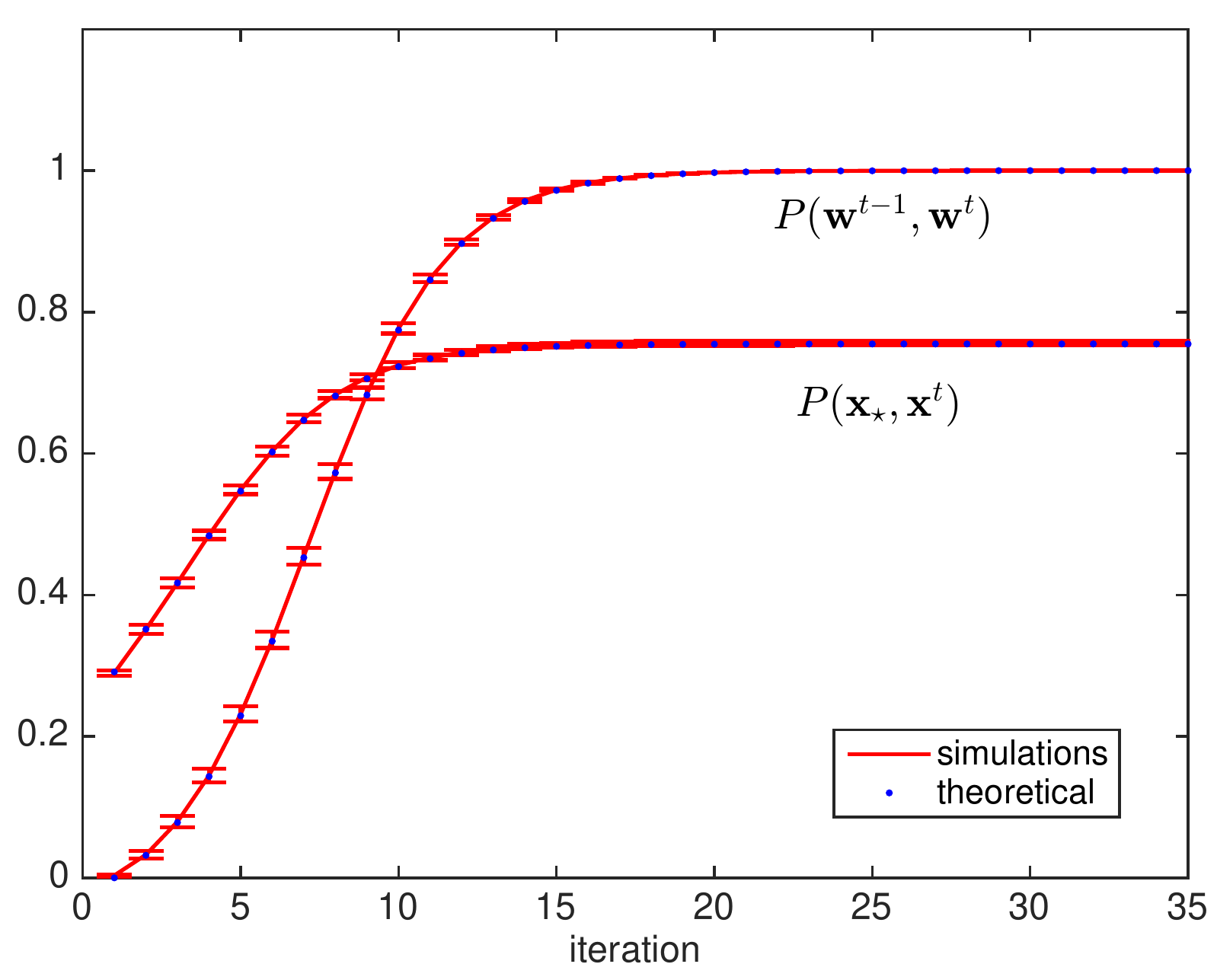}
\caption{Comparison between empirical results and theoretical predictions. $n=30000$. $\delta=3$. $\T=\T_\opt$. $\mu=\hat{\mu}$ where $\hat{\mu}$ is the unique solution to $\psi_1(\mu)=\frac{\delta}{\delta-1}$. $\alpha_0=0.2$ and $\sigma^2_0=1$. 10 independent realizations.}\label{Fig:SE}
\end{center}
\end{figure}

\begin{figure}[htbp]
\begin{center}
\includegraphics[width=.5\textwidth]{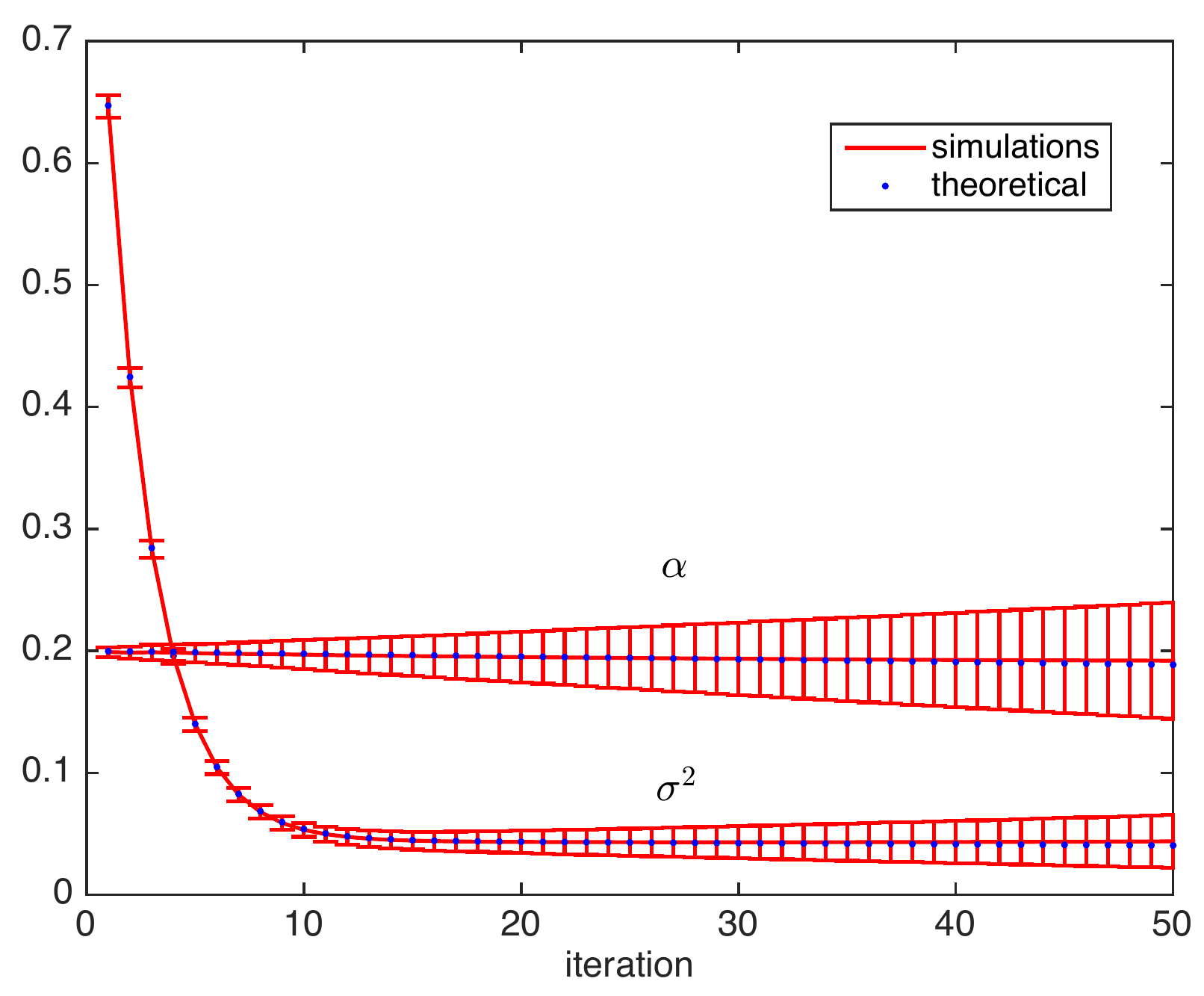}
\caption{Comparison between empirical results and theoretical predictions for $\{\alpha_t\}$ and $\{\sigma^2_t\}$. The settings are the same as those of \eqref{Fig:SE}. Notice that in theory we should have $\alpha_t=\alpha_0$ (the dotted line). However, the theoretical predictions of $\alpha_t$ shown in the figure deviates slightly from $0.2$. This is due to the numerical error incurred in computing $\hat{\mu}$ (namely, $\psi_1(\hat{\mu})$ is not exactly equal to $\delta/(\delta-1)$). This small error accumulates over iterations.}\label{Fig:SE_alpha}
\end{center}
\end{figure}

Finally, Fig.~\ref{Fig:SE_alpha} depicts the empirical and theoretical versions of $\alpha_t$ and $\sigma^2_t$ (see \eqref{Eqn:SE_final}). The SE predictions for these quantities are less accurate and still exhibit some fluctuations among different realizations. This can be explained as follows. As discussed in Section \ref{Sec:heuristics}, we conjecture that the spectral radius of $\bm{E}(\hat{\mu})$ converges to one as $m,n\to\infty$. However, for large but finite-sized instances, the spectral radius of $\bm{E}(\hat{\mu})$ will be slightly larger or smaller than one. Since \Alg can be viewed as a power method applied to the matrix $\bm{E}(\hat{\mu})$, as long as the spectral radius is not exactly one, the norm of $\bm{z}^t$ will keep shrinking or growing as $t$ increases. As a consequence, the mismatch between the predicted and simulated $\alpha_t,\sigma_t^2$ will accumulate as $t\to\infty$. Nevertheless, we believe that the characterization in Claim \ref{Lem:SE} is still correct. Namely, for any finite $t$, such mismatch vanishes as $m,n\to\infty$. As a comparison, the cosine similarity $P_\T(\bm{x}_\star,\bm{x}^t)$ in some sense normalizes such effect and matches excellently with the SE predictions.

\section{Conclusions}
In this paper, we studied a spectral method for phase retrieval under a practically-relevant partial orthogonal model. By analyzing the fixed points of an expectation propagation (EP) style algorithm, we are able to derive a formula to characterize the angle between the spectral estimator and the true signal vector. We conjecture that our prediction is exact in the asymptotic regime where $n,m=\delta n\to\infty$ and provide simulations to support our claim. Based on our asymptotic analysis, we found that the optimal $\T$ is the same as that for an i.i.d. Gaussian measurement model.

\appendices

\section{Proof of Lemma \ref{Lem:PT_second}}\label{App:proof_PT_sec}

\subsection{Proof of part (i)}\label{App:proof_PT_sec_a}
We first show that there exists at least one solution to $\psi_1(\mu)=\frac{\delta}{\delta-1}$ in $(0,\bar{\mu})$, if $\psi_1(\hat{\mu})>\frac{\delta}{\delta-1}$. Then, we prove that the solution is unique.

The existence of the solution follows from the continuity of $\psi_1$, and the following fact
\[
\psi_1(0)=\frac{\mathbb{E}[\delta|Z_\star|^2G(Y,0)]}{\mathbb{E}[G(Y,0)]}=1<\frac{\delta}{\delta-1},\quad\forall \delta>1,
\]
together with the hypothesis
\[
\psi_1(\bar{\mu})>\frac{\delta}{\delta-1}.
\]
We next prove the uniqueness of the solution. To this end, we introduce the following function
\BE\label{Eqn:f_def}
F(\mu)\Mydef \frac{1}{\mu}-\frac{1}{\mathbb{E}\left[\delta|Z_\star|^2G(\mu)\right]}.
\EE
Further, recall that  (cf. \eqref{Eqn:Lambda_def})
\BE\label{Eqn:Lambda_def2}
\Lambda(\mu) = \frac{1}{\mu}-\frac{\delta-1}{\delta}\cdot\frac{1}{\mathbb{E}\left[G(Y,\mu)\right]},
\EE
where $Y=|Z_\star|$ and $Z_\star\sim\mathcal{CN}(0,1/\delta)$.
From the definition of $\psi_1$ in \eqref{Eqn:psi_definitions}, it is straightforward to verify that
\[
\psi_1(\mu)=\frac{\delta}{\delta-1}\Longleftrightarrow F(\mu) =\Lambda(\mu),\quad\forall \mu\in(0,1],
\]
Hence, to prove that $\psi_1(\mu)=\frac{\delta}{\delta-1}$ cannot have more than one solution, it suffices to show that  $F(\mu) =\Lambda(\mu)$ cannot have more than one solution. To show this, we will prove that
\begin{itemize}
\item $\Lambda(\mu)$ is strictly decreasing on $(0,\bar{\mu})$;
\item$F(\mu)$ is strictly increasing on $(0,1)$.
\end{itemize}

We first prove the monotonicity of $\Lambda(\mu)$. From \eqref{Eqn:Lambda_def2}, we can calculate the derivative of $\Lambda(\mu)$:
\[
\Lambda'(\mu)=\frac{1}{\mu^2}\frac{\delta-1}{\delta}\left(\psi_2(\mu)-\frac{\delta}{\delta-1}\right).
\]
Recall that (cf. \eqref{Eqn:mu_star_def}) $\bar{\mu}$ is defined as
\BE\label{Eqn:mu_star_def2}
\bar{\mu}\Mydef\underset{\mu\in(0,1]}{\mr{argmin}} \quad \Lambda(\mu).
\EE
Further, Lemma \ref{Lem:T_increase} shows that $\psi_2(\mu)$ is strictly increasing on $(0,1)$. Two cases can happen: if $\psi_2(1)>\frac{\delta}{\delta-1}$, then $\bar{\mu}$ is the unique solution to $\psi_2(\mu)=\frac{\delta}{\delta-1}$; otherwise, $\bar{\mu}=1$. For both cases, it is easy to see that $\Lambda(\mu)$ is strictly decreasing on $(0,\bar{\mu})$.

It remains to prove the monotonicity of $F(\mu)$. To show this, we calculate its derivative:
\[
F'(\mu)=\frac{1}{\mu^2}\cdot\left( \frac{\mathbb{E}[\delta|Z_\star|^2G^2]-\left( \mathbb{E}[\delta|Z_\star|^2G] \right)^2}{\left( \mathbb{E}[\delta|Z_\star|^2G] \right)^2} \right).
\]
Hence, to show that $F(\mu)$ is increasing, we only need to show
\[
\mathbb{E}[\delta|Z_\star|^2G^2]> \left(\mathbb{E}[\delta|Z_\star|^2G] \right)^2.
\]
The above inequality can be proved using the association inequality in Lemma \ref{Lem:association}. To see this, we rewrite it as
\[
\mathbb{E}[\delta|Z_\star|^2G^2]\cdot\mathbb{E}[\delta|Z_\star|^2]> \mathbb{E}[\delta|Z_\star|^2G]\cdot\mathbb{E}[\delta|Z_\star|^2G].
\]
It is easy to see that the above inequality follows from Lemma \ref{Lem:association} with $B\Mydef \delta|Z_\star|^2$, $A\Mydef  \T$, $f(A)\Mydef \frac{1}{\mu^{-1}- A}$, and $g(A)\Mydef \frac{1}{\mu^{-1}- A}$. Clearly, $f(A)$ and $g(A)$ are increasing functions for $\mu\in(0,1]$, and the conditions for Lemma \ref{Lem:association} are satisfied.
\subsection{Proof of part (ii)}\label{App:proof_PT_1}
Define
\BE\label{Eqn:rho_bar}
\begin{split}
\theta^2_\T(\hat{\mu},\delta) & \Mydef \frac{\delta-1}{\delta}\cdot\frac{\rho_\T^2(\hat{\mu},\delta)}{1-\rho_\T^2(\hat{\mu},\delta)}\\
&=\frac{\frac{\delta}{\delta-1}-\psi_2(\hat{\mu})}{\psi_3^2(\hat{\mu})-\left(\frac{\delta}{\delta-1}\right)^2}
\end{split}
\EE
where the second step follows from the definition of $\rho_\T^2(\mu,\delta)$ (cf. \eqref{Eqn:SNR_asym}).
Hence, when $\delta>1$, we have
\[
0<\rho_\T^2(\mu,\delta)<1\Longleftrightarrow \theta_\T^2(\mu,\delta) >0.
\]
\textit{Hence, our problem becomes proving that $\psi_1(\bar{\mu})>\frac{\delta}{\delta-1}$ if and only if there exists at least one $\hat{\mu}\in(0,1]$ such that}
\[
\theta_\T^2(\hat{\mu},\delta)>0\quad \text{and}\quad \psi_1(\hat{\mu})=\frac{\delta}{\delta-1}.
\]

First, suppose $\psi_1(\bar{\mu})>\frac{\delta}{\delta-1}$ holds. We have proved in part (i) of this lemma that there is a unique solution to $\psi_1(\hat{\mu})=\frac{\delta}{\delta-1}$, where $\hat{\mu}<\bar{\mu}$. Further, from our discussions in part (i), the condition $\hat{\mu}<\bar{\mu}$ leads to $\psi_2(\hat{\mu})<\frac{\delta}{\delta-1}$. Also, from Lemma \ref{Lem:auxiliary1} (in Appendix \ref{App:aux}), we have
\[
\psi_3(\hat{\mu})>\psi_1(\hat{\mu})=\frac{\delta}{\delta-1}.
\]
Hence, both the numerator and denominator of $\rho_\T^2(\hat{\mu},\delta)$ are positive, and so $\rho_\T^2(\hat{\mu},\delta)>0$. This proved one direction of our claim.

To prove the other direction of the claim, suppose there exists a $\hat{\mu}$ such that $\theta_\T^2(\hat{\mu},\delta)>0$ and $ \psi_1(\hat{\mu})=\frac{\delta}{\delta-1}$. Again, by lemma \ref{Lem:auxiliary1}, the denominator of $\rho_\T(\hat{\mu},\delta)$ is positive. Hence (under the condition of $\psi_1(\hat{\mu})=\frac{\delta}{\delta-1}$), $\rho_\T^2(\hat{\mu},\delta)>0\Longleftrightarrow \psi_2(\hat{\mu})<\frac{\delta}{\delta-1}$. By the definition of $\bar{\mu}$ and connection between $\Lambda$ and $\psi_2$, we further have $\psi_2(\hat{\mu})<\frac{\delta}{\delta-1}\Longleftrightarrow\hat{\mu}<\bar{\mu}$. This means that there exists a $\hat{\mu}<\bar{\mu}$ such that $\psi_1(\hat{\mu})=\frac{\delta}{\delta-1}$, or equivalently $F(\hat{\mu})=\Lambda(\hat{\mu})$. Finally, by the monotonicity of $F(\cdot)$ and the strict monotonicity of $\Lambda(\cdot)$, and the fact that $\hat{\mu}<\bar{\mu}$, we must have $F(\bar{\mu})>\Lambda(\bar{\mu})$, or equivalently $\psi_1(\bar{\mu})>\frac{\delta}{\delta-1}$; see details in Section \ref{App:proof_PT_sec_a}.

Finally, our proof above implied that the condition
\[
\psi_1(\hat{\mu})=\frac{\delta}{\delta-1}\quad \text{and}\quad 0<\rho_\T^2(\hat{\mu},\delta)<1,
\]
is equivalent to
\[
\psi_1(\hat{\mu})=\frac{\delta}{\delta-1}\quad \text{and}\quad \psi_2(\hat{\mu})<\frac{\delta}{\delta-1}.
\]
\section{Optimality of $\T_\opt$}\label{App:optimality}
Denote the asymptotic cosine similarity achieved by $\T_{\star}$ as $\rho^2_\opt(\delta)$. We will discuss a few properties of $\rho^2_\opt(\delta)$ and then prove that no other $\T$ performs better than $\T_{\star}$. Our proof follows a similar strategy as that in \cite{Lu2018}.

We proceed in the following steps:
\begin{enumerate}
\item We first show that $\delta_{\mr{weak}}=2$, namely no $\T$ can work for $\delta<2$;
\item We further show that $\rho^2_{\star}(\delta)$ is strictly positive for any $\delta>2$;
\item Let $\mathcal{S}$ be a set of $\T$ for which the asymptotic cosine similarity is strictly positive for $\delta>2$. Clearly, we only need to consider functions in $\mathcal{S}$. We prove that the cosine similarity for any $\T\in\mathcal{S}$ cannot be larger than $\rho_{\star}^2(\delta)$. Restricting to $\mathcal{S}$ simplifies our discussions.
\end{enumerate}

\subsection{Weak threshold}\label{App:optimality_a}
\textit{We first prove that $\delta_\T$ is lower bounded by 2. }Namely, if $\delta<2$, then $\rho_\T^2(\delta)=0$ for any $\T$. According to Claim \ref{Claim:correlation}, if $\rho_\T^2(\delta)>0$, we must have $\psi_1(\bar{\mu})\ge\frac{\delta}{\delta-1}$. Further, Lemma \ref{Lem:PT_second} shows that there is a unique solution to the following equation (denoted as $\hat{\mu}$)
\[
\psi_1(\hat{\mu})=\frac{\delta}{\delta-1},\quad \hat{\mu}\in(0,\bar{\mu}].
\]
In Section \ref{App:proof_PT_sec_a} we have proved that
\[
\psi_2(\hat{\mu})<\frac{\delta}{\delta-1}\Longleftrightarrow \hat{\mu}\in(0,\bar{\mu}].
\]
Hence,
\[
\psi_1(\hat{\mu})>\psi_2(\hat{\mu}).
\]
From the definitions in \eqref{Eqn:psi_definitions} and noting $G(y,\mu)>0$ for any $\mu\in(0,1)$ and $y\ge0$, we can rewrite the condition $\psi_1(\hat{\mu})>\psi_2(\hat{\mu})$ as
\BE\label{Eqn:week_proof1}
\mathbb{E}\left[G^2_1\right]<\mathbb{E}\left[G_1\right]\cdot\mathbb{E}[\delta|Z_\star|^2G_1],
\EE
where we denoted $G_1\Mydef G(|Z_\star|,\hat{\mu})$. Further, applying the Cauchy-Schwarz inequality yields
\BE\label{Eqn:week_proof2}
\begin{split}
(\mathbb{E}[\delta|{Z}_{\star}|^2G_1])^2&\le \mathbb{E}[\delta^2|{Z}_{\star}|^4]\cdot\mathbb{E}[G^2_1]\\
&=2\cdot\mathbb{E}[G^2_1],
\end{split}
\EE
where the second step (i.e., $\mathbb{E}[\delta^2|{Z}_{\star}|^4]=2$) follows from the definition ${Z}_{\star}\sim\mathcal{CN}(0,1/\delta)$ and direct calculations of the fourth order moment of $|Z_\star|$. Combining \eqref{Eqn:week_proof1} and \eqref{Eqn:week_proof2} yields
\[
(\mathbb{E}[\delta|{Z}_{\star}|^2G_1])^2\le 2\mathbb{E}\left[G_1\right]\cdot\mathbb{E}[\delta|Z_\star|^2G_1],
\]
which further leads to
\BE\label{Eqn:week_proof3}
\frac{\mathbb{E}[\delta|{Z}_{\star}|^2G_1]}{\mathbb{E}\left[G_1\right]}\le2.
\EE
On the other hand, the condition $\psi_1(\hat{\mu})=\frac{\delta}{\delta-1}$ gives us
\BE\label{Eqn:week_proof4}
\frac{\mathbb{E}[\delta|{Z}_{\star}|^2G_1]}{\mathbb{E}\left[G_1\right]}=\frac{\delta}{\delta-1}.
\EE
Combining \eqref{Eqn:week_proof3} and \eqref{Eqn:week_proof4} leads to $\frac{\delta}{\delta-1}\le2$, and so $\delta\ge2$. This completes the proof.\vspace{5pt}

\textit{We now prove that $\delta_{\mr{}}=2$ can be achieved by $\T_{\opt}$.} When $\delta>2$, we have $\frac{\delta}{\delta-1}\in(1,2)$. Since $\T_\star$ is an increasing function, by Lemma \ref{Lem:T_increase}, both $\psi_1(\mu)$ and $\psi_2(\mu)$ are increasing functions on $\mu\in(0,1)$. Further, it is straightforward to show that $\psi_1(0)=\psi_2(0)=1$ and $\psi_1(1)=\psi_2(1)=2$. Further, Lemma \ref{Lem:psi1_psi2} (in Appendix \ref{App:aux}) shows that for $\T=\T_\star$
\[
\psi_2(\mu)<\psi_1(\mu),\quad\forall \mu\in(0,1).
\]
The above facts imply that when $\delta>2$ we have
\[
\psi_1(\bar{\mu})>\frac{\delta}{\delta-1},
\]
where $\bar{\mu}$ is the unique solution to
\[
\psi_2(\bar{\mu})=\frac{\delta}{\delta-1}.
\]
Then, using Lemma \ref{Lem:PT_second} we proved $\rho^2_\star(\delta)>0$ for $\delta>2$.
\subsection{Properties of $\rho^2_\opt(\delta)$}
Since $\delta_{\mr{weak}}=2$, we will focus on the regime $\delta>2$ in the rest of this appendix.

For notational brevity, we will put a subscript $_\star$ to a variable (e.g., $\hat{\mu}_\star$) to emphasize that it is achieved by $\T=\T_\opt$. Further, for brevity, we use the shorthand $G_\star$ for $G(|Z_\star|,\mu)$. Let $\rho^2_\opt(\delta)$ be the function value of achieved by $\T_\opt$. Further, for convenience, we also define (see \eqref{Eqn:rho_bar})
\BE\label{Eqn:Opt_8}
\begin{split}
\theta^2_\opt(\delta)&\Mydef \frac{\delta-1}{\delta}\cdot\frac{\rho^2_\star(\delta)}{1-\rho^2_\star(\delta)}\\
&=\frac{\frac{\delta}{\delta-1}-\frac{\mathbb{E}[G^2_\opt]}{\mathbb{E}^2[G_\opt]}}{\frac{\mathbb{E}[\delta|Z_\star|^2{G}_\opt^2]}{\mathbb{E}^2[G_\opt]}-\left(\frac{\delta}{\delta-1}\right)^2},
\end{split}
\EE
where the last equality is from the definition in \eqref{Eqn:SNR_asym}.
We next show that $P^2_\opt(\delta)$ can be expressed compactly as
\BE\label{Eqn:rho_optimal}
\theta^2_\opt(\delta)=\frac{1}{\hat{\mu}_\opt}-1,
\EE
where $\hat{\mu}_\opt$ is the unique solution to $\psi_1(\mu)=\frac{\delta}{\delta-1}$ in $(0,1)$. Then, from \eqref{Eqn:Opt_8}, it is straightforward to obtain
\[
{\rho}^2_\opt(\delta)=\frac{1-\hat{\mu}_\star}{1-\frac{1}{\delta}\hat{\mu}_\star}.
\]

For $\T_\opt=1-\frac{1}{\delta |Z_\star|^2}$, the function $G_\opt(|Z_\star|,\hat{\mu}_\star)$ (denoted as $G_\star$ hereafter) is given by
\BE\label{Eqn:Opt_3}
G_\opt(|Z_\star|,\hat{\mu}_\star)=\frac{1}{\hat{\mu}_\star^{-1}-\T_\opt(|Z_\star|)}=\frac{\hat{\mu}_\star\delta|Z_\star|^2}{(1-\hat{\mu}_\opt)\delta|Z_\star|^2+\hat{\mu}_\opt},
\EE
where $\hat{\mu}_\opt\in(0,1]$ is the unique solution to
\BE\label{Eqn:Opt_4}
\psi_1(\hat{\mu}_\star)=\frac{\mathbb{E}[\delta|Z_\star|^2G_\opt]}{\mathbb{E}[G_\opt]}=\frac{\delta}{\delta-1}.
\EE
The existence and uniqueness of $\hat{\mu}_\star$ (for $\delta>2$) is guaranteed by the monotonicity of $\psi_1$ under $\T_\opt$ (see Lemma \ref{Lem:T_increase}). Our first observation is that $G_{\opt}$ in \eqref{Eqn:Opt_3} satisfies the following relationship:
\BE\label{Eqn:Opt_5}
\mu_{\opt}\delta|Z_\star|^2-(1-\hat{\mu}_\opt)\delta|Z_\star|^2G_\opt=\hat{\mu}_\opt G_\opt.
\EE
Further, multiplying both sides of \eqref{Eqn:Opt_5} by $G_\star$ yields
\BE\label{Eqn:Opt_6}
\mu_{\opt}\delta|Z_\star|^2G_\opt-(1-\hat{\mu}_\opt)\delta|Z_\star|^2G^2_\opt=\hat{\mu}_\opt G^2_\opt.
\EE
Taking expectations over \eqref{Eqn:Opt_5} and \eqref{Eqn:Opt_6}, and noting $\mathbb{E}[\delta|Z_\star|^2]=1$, we obtain
\BS\label{Eqn:Opt_7}
\begin{align}
\hat{\mu}_\opt-(1-\hat{\mu}_\opt)\mathbb{E}\left[\delta|Z_\star|^2G_\opt\right]&=\hat{\mu}_\opt \mathbb{E}[G_\opt],\label{Eqn:Opt_7a}\\
\hat{\mu}_\opt\cdot\mathbb{E}\left[\delta|Z_\star|^2G_\opt\right]-(1-\hat{\mu}_\opt)\mathbb{E}\left[\delta|Z_\star|^2G^2_\opt\right]&=\hat{\mu}_\opt \mathbb{E}\left[G^2_\opt\right].\label{Eqn:Opt_7b}
\end{align}
\ES
Substituting \eqref{Eqn:Opt_7b} into \eqref{Eqn:Opt_8}, and after some calculations, we have\textcolor{red}{}
\BE\label{Eqn:Opt_9}
\begin{split}
\theta^2_\opt(\delta)&=\frac{1-\hat{\mu}_\opt}{\hat{\mu}_\opt}\cdot \frac{\frac{\mathbb{E}[\delta|Z_\star|^2G^2_\opt]}{\mathbb{E}^2[G_\opt]}-\frac{\delta}{\delta-1}\cdot\frac{\hat{\mu}_\opt}{1-\hat{\mu}_\opt}\cdot\left(\frac{1}{\mathbb{E}[G_\opt]}-1\right)}{\frac{\mathbb{E}[\delta|Z_\star|^2G^2_\opt]}{\mathbb{E}^2[G_\opt]}-\left(\frac{\delta}{\delta-1}\right)^2},
\end{split}
\EE
where we have used the identity $\psi_1(\hat{\mu}_\star)=\mathbb{E}[\delta|Z_\star|^2G_\star]/\mathbb{E}[G_\star]=\delta/(\delta-1)$.
From \eqref{Eqn:Opt_9}, to prove $\theta^2_\opt(\delta)=\hat{\mu}_\opt^{-1}-1$, we only need to prove
\BE
\frac{\hat{\mu}_\opt}{1-\hat{\mu}_\opt}\cdot\left(\frac{1}{\mathbb{E}[G_\opt]}-1\right)=\frac{\delta}{\delta-1},
\EE
which can be verified by combining \eqref{Eqn:Opt_4} and \eqref{Eqn:Opt_7a}.

Before leaving this section, we prove the monotonicity argument stated in Theorem \ref{Lem:optimality}. From \eqref{Eqn:rho_optimal}, to prove that $\rho^2_\opt(\delta)$ (or equivalently $\theta^2_\star(\delta)$) is an increasing function of $\delta$, it suffices to prove that $\hat{\mu}_\opt(\delta)$ is a decreasing function of $\delta$. This is a direct consequence of the following facts: (1) $\hat{\mu}_\opt(\delta)$ is the unique solution to $\psi_1(\mu)=\frac{\delta}{\delta-1}$ in $(0,1)$, and (2) $\psi_1(\mu)$ is an increasing function of $\mu$. The latter follows from Lemma \ref{Lem:T_increase} ($\T_\star(y)=1-\frac{1}{\delta y^2}$ is an increasing function).
\subsection{Optimality of $G_\opt$}
In the previous section, we have shown that the weak threshold is $\delta_{\mr{weak}}=2$. Consider a \textit{fixed} $\delta$ (where $\delta>2$) and our goal is to show that $\rho^2_\T(\delta)\le\rho_\star^2(\delta)$ for any $\T$.

We have proved $\rho^2_\star(\delta)>0$ (the asymptotic cosine similarity) for $\delta>2$. Hence, we only need to consider $\T$ satisfying $\rho^2_\T(\delta)>0$ (in which case we must have $\psi_1(\bar{\mu})>\frac{\delta}{\delta-1}$), since otherwise $\T$ is already worse than $\T_\star$. In Lemma \ref{Lem:PT_second}, we showed that the phase transition condition $\psi_1(\bar{\mu})>\frac{\delta}{\delta-1}$ can be equivalently reformulated as
\[
\exists \hat{\mu}\in(0,1],\ 0<\rho_\T^2(\hat{\mu},\delta)<1\quad\text{and}\quad \psi_1(\hat{\mu})=\frac{\delta}{\delta-1}.
\]
Also, from \eqref{Eqn:rho_bar} we see that
\[
0<\rho_\T^2(\hat{\mu},\delta)<1\Longleftrightarrow \theta_\T^2(\hat{\mu},\delta)>0.
\]
From the above discussions, the problem of optimally designing $\T$ can be formulated as
\BE\label{Eqn:Opt_1}
\begin{split}
\sup_{\T,\ \hat{\mu}\in(0,1)}\quad & \theta_\T^2(\delta,\hat{\mu})\\
s.t.
\quad & \theta_\T^2(\delta,\hat{\mu})>0,\\
\quad &\psi_1(\hat{\mu})=\frac{\delta}{\delta-1}.
\end{split}
\EE
In the above formulation, $\hat{\mu}\in(0,\mu)$ is treated as a variable that can be optimized. In fact, for a given $\T$, there cannot exist more than one $\hat{\mu}\in(0,1)$ such that $\psi_1(\hat{\mu})=\frac{\delta}{\delta-1}$ and $\theta_\T^2(\delta,\hat{\mu})>0$ hold simultaneously (from Lemma \ref{Lem:PT_second}). There can no such $\hat{\mu}$, though. In such cases, it is understood that $\theta_\T^2(\delta,\hat{\mu})=0$.

Substituting in \eqref{Eqn:psi_definitions} and after straightforward manipulations, we can rewrite \eqref{Eqn:Opt_1} as
\BE\label{Eqn:Opt_2}
\begin{split}
\sup_{G(\cdot)>0} \quad &\frac{\frac{\delta}{\delta-1}-\frac{\mathbb{E}[G^2]}{\mathbb{E}^2[G]}}{\frac{\mathbb{E}[\delta|Z_\star|^2{G}^2]}{\mathbb{E}^2[G]}-\left(\frac{\delta}{\delta-1}\right)^2}>0\\
s.t.\quad &\frac{\mathbb{E}[\delta|Z_\star|^2{G}]}{\mathbb{E}[G]}=\frac{\delta}{\delta-1},
\end{split}
\EE
where ${G}(y,\hat{\mu})$ is
\BE\label{Eqn:Opt_25}
{G}(y,\hat{\mu})= \frac{1}{\hat{\mu}^{-1}-\T(y)}.
\EE
Note that $G(y,\hat{\mu})\ge0$ for $\hat{\mu}\in(0,1]$.

At this point, we notice that the function to be optimized has been changed to the nonnegative function $G(\cdot)$ (with $\hat{\mu}$ being a parameter). Hence, the optimal $G(\cdot)$ is clearly not unique. In the following, we will show that the optimal value of the objective function cannot be larger than that $\rho_\star(\delta)$.

Consider $G(\cdot)$ be an arbitrary feasible function (satisfying $\mathbb{E}[G]=1$), and let $\theta^2(\delta)$ (or simply $\theta^2$) be the corresponding function value of the objective in \eqref{Eqn:Opt_2}. We now prove that $\theta(\delta)\le \theta_\opt(\delta)$ for any $\delta>2$. First, note that scaling the function $\hat{G}(\cdot)$ by a positive constant does not change the objective function and the constraint of the problem. Hence, without loss of generality and for simplicity of discussions, we assume
\[
\mathbb{E}[\hat{G}]=1.
\]
Since $\theta^2$ is the objective function value achieved by $\hat{G}$, by substituting the definition of $\hat{G}$ into \eqref{Eqn:Opt_2}, we have
\BE
\frac{\frac{\delta}{\delta-1}-\mathbb{E}[\hat{G}^2]}{\mathbb{E}[\delta|Z_\star|^2\hat{G}^2]-\left(\frac{\delta}{\delta-1}\right)^2}=\theta^2.
\EE
Some straightforward manipulations give us
\BE\label{Eqn:Opt_10}
\mathbb{E}\left[ (\theta^2\delta|Z_\star|^2+1)\hat{G}^2 \right]=\theta^2\left(\frac{\delta}{\delta-1}\right)^2
+\frac{\delta}{\delta-1}.
\EE
We assume $\theta^2>0$, since otherwise it already means that $G$ is worse than $G_\opt$ (note that $\theta_\opt^2=\hat{\mu}_\opt^{-1}-1$ is strictly positive). Hence, $\theta^2\delta|Z_\star|^2+1>0$, and we can lower bound $\mathbb{E}\left[ (\theta^2\delta|Z_\star|^2+1)\hat{G}^2 \right]$ by
\BE\label{Eqn:Opt_11}
\begin{split}
\mathbb{E}\left[ (\theta^2\delta|Z_\star|^2+1)\hat{G}^2 \right] &\ge \frac{\left(\mathbb{E}[\delta|Z_\star|^2\hat{G}]\right)^2}{ \mathbb{E}\left[\left(\frac{\delta|Z_\star|^2}{\sqrt{ \theta^2\delta|Z_\star|^2+1 }}\right)^2\right] }\\
&=\frac{\left(\frac{\delta}{\delta-1}\right)^2}{\mathbb{E}\left[\frac{\delta^2|Z_\star|^4}{\theta^2
\delta|Z_\star|^2+1}\right]},
\end{split}
\EE
where the first line follows from the Cauchy-Swarchz inequality $\mathbb{E}[X^2]\ge\mathbb{E}^2[XY]/\mathbb{E}[Y^2]$, and the second equality is due to the constraint $\mathbb{E}[|Z_\star|^2\hat{G}]=\frac{\delta}{\delta-1}$. Combining \eqref{Eqn:Opt_10} and \eqref{Eqn:Opt_11} yields
\BE\label{Eqn:Opt_12}
\mathbb{E}\left[\frac{\delta^2|Z_\star|^4}{\theta^2
\delta|Z_\star|^2+1}\right]\ge\frac{1}{\theta ^2
+\frac{\delta-1}{\delta}}.
\EE
Further, we note that
\BE\label{Enq:Opt_13}
\begin{split}
&\mathbb{E}\left[\frac{\delta^2|Z_\star|^4}{\theta^2
\delta|Z_\star|^2+1}\right]=\frac{1}{\theta^2}\left(\mathbb{E}[\delta|Z_\star|^2]-
\mathbb{E}\left[\frac{\delta|Z_\star|^2}{\theta^2
\delta|Z_\star|^2+1}\right]\right)\\
&=\frac{1}{\theta^2}\left(1-
\mathbb{E}\left[\frac{\delta|Z_\star|^2}{\theta^2
\delta|Z_\star|^2+1}\right]\right).
\end{split}
\EE
Then, substituting \eqref{Enq:Opt_13} into \eqref{Eqn:Opt_12}  gives us
\BE\label{Enq:Opt_14}
\mathbb{E}\left[\frac{\delta|Z_\star|^2}{\theta^2
\delta|Z_\star|^2+1}\right]\le\frac{\frac{\delta-1}{\delta}}{\theta^2+\frac{\delta-1}{\delta}}.
\EE
Combining \eqref{Enq:Opt_13} and \eqref{Enq:Opt_14}, we can finally get
\BE\label{Eqn:Opt_15}
\frac{\mathbb{E}\left[\frac{\delta^2|Z_\star|^4}{\theta^2
\delta|Z_\star|^2+1}\right]}{\mathbb{E}\left[\frac{\delta|Z_\star|^2}{\theta^2
\delta|Z_\star|^2+1}\right]}\ge \frac{\delta}{\delta-1}.
\EE

It remains to prove
\BE\label{Eqn:Opt_16}
\theta^2<\theta^2_\opt=\frac{1}{\hat{\mu}_\opt}-1.
\EE
To this end, we note that substituting \eqref{Eqn:Opt_3} into  \eqref{Eqn:Opt_4} yields
\BE\label{Eqn:Opt_17}
\frac{\mathbb{E}\left[\frac{\delta^2|Z_\star|^4}{(\hat{\mu}_\opt^{-1}-1)
\delta|Z_\star|^2+1}\right]}{\mathbb{E}\left[\frac{\delta|Z_\star|^2}{(\hat{\mu}_\opt^{-1}-1)
\delta|Z_\star|^2+1}\right]}= \frac{\delta}{\delta-1}.
\EE
From Lemma \ref{Lem:T_increase}, the LHS of \eqref{Eqn:Opt_15} (which is $\psi_1(1/(1+\theta^2))$ under $\T_\opt$) is a strictly decreasing function of $\theta\in(0,\infty)$. Hence, combining \eqref{Eqn:Opt_15} and \eqref{Eqn:Opt_17} proves \eqref{Eqn:Opt_16}.

\section{Derivations of the \Alg algorithm}\label{App:derivations}
In this appendix, we provide detailed derivations for the \Alg algorithm, which is an instance of the algorithm proposed in \cite{fletcher2016,He2017,Schniter16,Meng18}. The \Alg algorithm is derived based on a variant of expectation propagation \cite{Minka2001}, referred to as \textit{scalar EP} in \cite{Ccakmak2018}. The scalar EP approximation was first mentioned in \cite[pp. 2284]{opper2005} (under the name of diagonally-restricted approximation) and independently studied in \cite{ma2015turbo}\footnote{The algorithm in \cite{ma2015turbo} is equivalent to scalar EP, but derived (heuristically) in a different way.}. An appealing property of scalar EP is that its asymptotic dynamics could be characterized by a state evolution (SE) procedure under certain conditions. Such SE characterization for scalar EP was first observed in \cite{ma2015turbo,liu2016,Ma2016} and later proved in \cite{Rangan17,Takeuchi2017}. Notice that the SE actually holds for more general algorithms that might not be derived from scalar EP \cite{Ma2016,Rangan17}.

For simplicity of exposition, we will focus on the real-valued setting in this appendix. We then
generalize the \Alg algorithm to the complex-valued case in a natural way
(e.g., replacing matrix transpose to conjugate transpose, etc).
\subsection{The overall idea}\label{App:derivations_a}
The leading eigenvector of $\bm{D}$ is a solution to the following problem:
\BE
\min_{\|\bm{x}\|=\sqrt{n}}\ -\bm{x}^{\UT} \bm{D}\bm{x},
\EE
where $\bm{D}$ is defined in \eqref{Eqn:data_matrix}, and the normalization
$\|\bm{x}\|=\sqrt{n}$ (instead of $\|\bm{x}\|=1$) is imposed for discussion
convenience. By introducing a Lagrange multiplier $\lambda$, we further transform the
above problem into an
unconstrained one:
\BE\label{Eqn:PCA_unconstrain}
\min_{\bm{x}\in\mathbb{R}^n}\ -\bm{x}^{\UT} \bm{D}\bm{x}+ \lambda \|
\bm{x}\|^2.
\EE
To yield the principal
eigenvector, $\lambda$ should be set to
\BE
\lambda = \lambda_1(\bm{D}),
\EE
where $\lambda_1(\bm{D})$ denotes the maximum eigenvalue of $\bm{D}$. Note
that the
unconstrained formulation is not useful algorithmically since $\lambda_1(\bm{D})$
is not known a priori. Nevertheless, based on the unconstrained reformulation, we will derive a set of self-consistent
equations that can provide useful information
about the eigen-structure of $\bm{D}$.
Such formulation of the maximum eigenvector problem using a Lagrange multiplier
has also been adopted in \cite{Kabashima2010}.

Following \cite{fletcher2016}, we introduce an auxiliary variable $\bm{z}=\bm{Ax}$ and reformulate \eqref{Eqn:PCA_unconstrain} as
\BE\label{Eqn:optimization_final}
\min_{\bm{x}\in\mathbb{R}^n,\bm{z}\in\mathbb{R}^m}\ \underbrace{-\sum_{a=1}^m |z_a|^2\cdot
\T(y_a)}_{f_2(\bm{z})}  + \underbrace{\lambda\| \bm{x}\|^2}_{f_1(\bm{x})}+\mathbb{I}(\bm{z}=\bm{Ax}).
\EE
Our first step is to construct a joint pdf of $\bm{x}\in\mathbb{R}^n$ and $\bm{z}\in\mathbb{R}^{m}$:
\BE\label{Eqn:joint_pdf}
\ell(\bm{x},\bm{z})=\frac{1}{Z}\underbrace{\exp(-\beta\cdot f_2(\bm{z}))}_{F_2(\bm{z})}\cdot \underbrace{\exp(-\beta\cdot f_1(\bm{x}))}_{F_1(\bm{x})}\cdot \mathbb{I}(\bm{z}=\bm{Ax}),
\EE
where $Z$ is a normalizing constant and $\beta>0$ is a parameter (the inverse temperature). The factor graph corresponding to the above pdf is shown in
Fig.~\ref{Fig:factor2}.
Similar to \cite{ArianThesis}, we derive \Alg based on the following steps:
\begin{itemize}
\item Derive an EP algorithm, referred to as \textsf{PCA-EP-$\beta$}, for the factor graph shown in Fig.~\ref{Fig:factor2};
\item Obtain \Alg as the zero-temperature limit (i.e., $\beta\to\infty$) of \textsf{PCA-EP-$\beta$}.
\end{itemize}
Intuitively, as $\beta\to\infty$, the pdf $\ell(\bm{x},\bm{z})$ concentrates around the minimizer of \eqref{Eqn:optimization_final}. The \Alg algorithm is a low-cost message passing algorithm that intends to find the minimizer. Similar procedure has also been used to derive an AMP algorithm for solving an amplitude-loss based phase retrieval problem \cite[Appendix A]{MXM18}.

We would like to point out that, for the PCA problem in \eqref{Eqn:optimization_final}, the resulting \textsf{PCA-EP}-$\beta$ algorithm becomes invariant to $\beta$ (the effect of $\beta$ cancels out). This is essentially due to the Gaussianality of the factors $F_1(\bm{x})$ and $F_2(\bm{z})$ defined in \eqref{Eqn:joint_pdf}, as will be seen from the derivations in the next subsections. Note that this is also the case for the AMP.S algorithm derived in \cite{MXM18}, which is an AMP algorithm for solving \eqref{Eqn:optimization_final}.
\subsection{Derivations of \textsf{PCA-EP-$\beta$}}\label{Sec:PCA_EP_beta}
As shown in Fig.~\ref{Fig:factor2}, the factor graph has three factor nodes, represented in the figure as $F_1(\bm{x})$, $F_2(\bm{z})$ and $\bm{z}=\bm{Ax}$ respectively.

Before we proceed, we first point out that the message from node $\bm{x}$ to node $\bm{z}=\bm{Ax}$ is equal to $F_1(\bm{x})=\exp(-\lambda\|\bm{x}\|^2)$, and is invariant to its incoming message. This is essentially due to the fact that $F_1(\bm{x})$ is a Gaussian pdf with identical variances, as will be clear from the scalar EP update rule detailed below. As a consequence, we only need to update the messages exchanged between node $\bm{z}=\bm{Ax}$ and node $F_2(\bm{z})$; see Fig.~\ref{Fig:factor2}. Also, due to the Gaussian approximations adopted in EP algorithms, we only need to track the mean and variances of these messages.

In the following discussions, the means and precisions (i.e., the reciprocal of variance) of $m_{1\leftarrow 2}(\bm{z})$ and $m_{1\to 2}(\bm{z})$ are denoted as $\bm{z}_{1\leftarrow 2}$, $\beta\rho_{1\leftarrow 2}$, and $\bm{z}_{1\to 2}$, $\beta\rho_{1\to 2}$ respectively.

\begin{figure*}[hbpt]
\centering
\includegraphics[width=.6\textwidth]{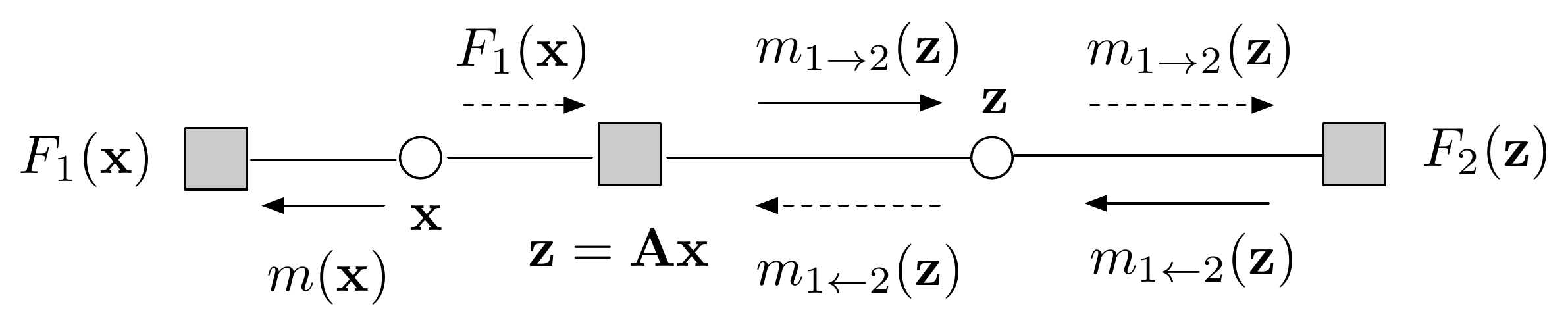}
\caption{Factor graph corresponding to the pdf in \eqref{Eqn:joint_pdf}. }\label{Fig:factor2}
\end{figure*}

\subsubsection{Message from $F_2(\mathbf{z})$ to $\mathbf{z}$}\label{Sec:EPs}

Let $m_{1\to2}^t(\bm{z})=\mathcal{N}(\bm{z}^t_{1\to2},1/(\beta\rho^t_{1\to2}\bm{I}))$ be the incoming message of node $F_2(\bm{z})$ at the $t$-th iteration. EP computes the outgoing message $m^t_{1\leftarrow2}(\bm{z})$ based on the following steps  \cite{Minka2001} (see also \cite{Rangan17}):\par\vspace{5pt}
\underline{(1) Belief approximation:}
The local belief at node $F_2(\bm{z})$ reads:
\BE
b_{2} ^t(\bm{z})\propto F_2(\bm{z})\cdot m_{1\to2}^t(\bm{z}).
\EE
For the general case where $F_2(\bm{z})$ is non-Gaussian, $b_{2}^t(\bm{z})$ is non-Gaussian. The first step of EP is to approximate $b^t_{2}(\bm{z})$ as a
Gaussian pdf based on moment matching:
\BE
 \hat{{b}}_2^t(\bm{z})=\mr{Proj}\big[ b_{2}^t(\bm{z})\big],
\EE
where, following the scalar EP approximation \cite{opper2005,ma2015turbo,liu2016,Ma2016,Rangan17,Takeuchi2017}, $\mr{Proj}\big[ b_{2}^t(\bm{z})\big]$ is given by
\BE\label{Eqn:scaler_EP_app}
\mr{Proj}\big[ b_{2}^t(\bm{z})\big]_i
= \mathcal{N}\Big(z_{i}; {z}^t_{2i},\frac{1}{m}\sum_{j=1}^m\frac{{v}^t_{2j}}{\beta}\Big),\quad \forall i=1,\ldots,m,
\EE
Here, ${z}^t_{2i}$ and ${v}^t_{2i}/\beta$ represent the marginal means and variance:
\BE
\begin{split}
{z}^t_{2i} &=\mathbb{E}[z_i],\\
{v}^t_{2i}/\beta&=\mathrm{var}[z_i],
\end{split}
\EE
where the expectations are taken w.r.t. the belief $b_2^t(\bm{z})=\frac{F_2(\bm{z})m_{1\to2}^t(\bm{z})}{\int F_2(\bm{z})m^t_{1\to2}(\bm{z}) \mr{d}\bm{z}}$. Using $F_2(\bm{z})=\exp\left(\beta \sum_{a=1}^m|z_a|^2\T(y_a)\right)$, it is straightforward to get the following closed form expressions for ${z}^t_{2i}$ and ${v}^t_{2i}$\footnote{This is under the condition that $\rho^t_{1\to2}>2\T(y_a)$, $\forall a=1,\ldots,m$. We assume that such condition is satisfied in deriving the \Alg algorithm.}
\BE\label{Eqn:SEP_post_2}
\begin{split}
{z}_{2i}^{t} &= \frac{\rho_{1\to2}^t}{\rho_{1\to2}^t-2\T(y_i)}\cdot{z}^t_{1\to2,i},\\
v_{2i}^t&= \frac{1}{\beta}\cdot\frac{1}{\rho_{1\to2}^t-2\T(y_i)}.
\end{split}
\EE

The approximation in \eqref{Eqn:scaler_EP_app} is based on the scalar EP. The difference between scalar EP and the conventional  EP will be discussed in Remark \ref{Rem:1} at the end of this subection.
\vspace{5pt}

\underline{(2) Message update:}
The outgoing message is computed as
\BE\label{Eqn:EP_outgoing}
m_{1\leftarrow 2}^{t}(\bm{z})\propto \frac{ \hat{{b}}^t_2(\bm{z})}{ m^{t}_{1\to 2}(\bm{z})}.
\EE
Since both the numerator and the denominator for the RHS of \eqref{Eqn:EP_outgoing} are Gaussian pdfs,
the resulting message is also Gaussian. The mean and
precision (i.e., the reciprocal of the variance) of $m^t_{1\to 2}(x_i)$ are respectively given by \cite{Minka2001}:
\BE\label{Eqn:EP_Gaussian_out2}
\begin{split}
\bm{z}_{1\leftarrow2}^t &= \frac{\beta \rho_2^t}{\beta \rho_2^t-\beta\rho^t_{1\to2}}\cdot{\bm{z}}^t_{2}-\frac{\beta \rho_{1\to2}^t}{\beta \rho_2^t-\beta\rho_{1\to2}^t}\cdot\bm{z}^t_{1\to2},\\
\beta\rho^t_{1\leftarrow 2} &= \beta\rho^t_{2} - \beta\rho^t_{1\to 2},
\end{split}
\EE
where
\[
\rho_2^t\Mydef\Big( \frac{1}{m}\sum_{j=1}^m v_{2j}^t \Big)^{-1}.
\]
Clearly, we can write \eqref{Eqn:EP_Gaussian_out2} into the following more compact form:
\BE
\begin{split}
\bm{z}^t_{1\leftarrow2} &= \frac{ \rho_2^t}{ \rho_2^t-\rho^t_{1\to2}}\cdot\bm{z}^t_{2}-\frac{ \rho^t_{1\to2}}{ \rho^t_2-\rho^t_{1\to2}}\cdot\bm{z}^t_{1\to2},\\
\rho_{1\leftarrow 2}^t &= \rho_{2}^t - \rho^t_{1\to 2}.
\end{split}
\EE

\subsubsection{Message from $\mathbf{z}$ to ${F}_1(\mathbf{z})$}

Let $m^t_{1\leftarrow 2}(\bm{z})=\mathcal{N}(\bm{z};\bm{z}^t_{1\leftarrow2},1/(\beta\rho^t_{1\leftarrow2}))$. The message $m^{t+1}_{1\to 2}(\bm{z})$ is calculated as \cite{Minka2001}
\[
m^{t+1}_{1\to 2}(\bm{z})\propto\frac{\text{Proj}\left[ \int_{\bm{x}} F_1(\bm{x})\mathbb{I}(\bm{z}=\bm{Ax}) m^t_{1\leftarrow2}(\bm{z})\mr{d}\bm{x} \right]}{m^t_{1\leftarrow2}(\bm{z})},
\]
where $F_1(x)=\exp(-\beta\lambda\|\bm{x}\|^2)$. In the above expression, $\text{Proj}(\cdot)$ denotes the scalar EP approximation in \eqref{Eqn:scaler_EP_app}, and to calculate the numerator we need to evaluate the following moments:
\BE
\begin{split}
{\bm{z}}^{t+1}_1&\Mydef \frac{\iint \bm{z}F_1(\bm{x})\mathbb{I}(\bm{z}=\bm{Ax})m^t_{1\leftarrow 2}(\bm{z})\mr{d}\bm{x}\mr{d}\bm{z}}{\iint F_1(\bm{x})\mathbb{I}(\bm{z}=\bm{Ax})m^t_{1\leftarrow 2}(\bm{z})\mr{d}\bm{x}\mr{d}\bm{z}},\\
\frac{{v}_1^{t+1}}{\beta}&\Mydef \frac{1}{m}\left(\frac{\iint\| \bm{z}\|^2F_1(\bm{x})\mathbb{I}(\bm{z}=\bm{Ax})m^t_{1\leftarrow 2}(\bm{z})\mr{d}\bm{x}\mr{d}\bm{z}}{\iint F_1(\bm{x})\mathbb{I}(\bm{z}=\bm{Ax})m^t_{1\leftarrow 2}(\bm{z})\mr{d}\bm{x}\mr{d}\bm{z}}-\|{\bm{z}}_1^{t+1}\|^2\right).
\end{split}
\EE
Using the definitions $F_1(x)=\exp(-\beta\lambda\|\bm{x}\|^2)$ and $m^t_{1\leftarrow 2}(\bm{z})=\mathcal{N}(\bm{z};\bm{z}^t_{1\leftarrow2},1/(\beta\rho^t_{1\leftarrow2}))$, it is straightforward to show that
\BE\label{Eqn:SEP_post_1}
\begin{split}
{\bm{z}}_1^{t+1} &=\bm{A}\left( \frac{2\lambda}{\rho^t_{1\leftarrow2}}\bm{I}+\bm{A}^{\UT}\bm{A} \right)^{-1} \bm{A}^{\UT}\bm{z}^t_{1\leftarrow2},\\
\frac{1}{\beta\rho_1^{t+1}} &= \frac{1}{m}\Tr\left(\bm{A}^{\UT}\left( 2\beta\lambda\bm{I}+\beta\rho^t_{1\leftarrow2}\bm{A}^{\UT}\bm{A} \right)^{-1}\bm{A}\right),
\end{split}
\EE
where $\rho_1^{t+1}=1/{v}_1^{t+1}$.
Finally, similar to \eqref{Eqn:EP_outgoing}, the mean/precision of the output message $m^{t+1}_{1\to2}(\bm{z})$ are given by
\BE\label{Eqn:EP_Gaussian_out1}
\begin{split}
\bm{z}_{1\to2}^{t+1} &=\frac{\rho_1^{t+1}}{\rho_1^{t+1}-\rho^t_{1\leftarrow2}}\cdot {\bm{z}}_1^{t+1}-\frac{\rho^{t}_{1\leftarrow2}}{\rho_1^{t+1}-\rho^t_{1\leftarrow2}}\cdot {\bm{z}}^t_{1\leftarrow2},\\
\rho_{1\to2}^{t+1} &=\rho_1^{t+1}-\rho^t_{1\leftarrow2}.
\end{split}
\EE

\vspace{5pt}
\begin{remark}[Difference between diagonal EP and scalar EP]\label{Rem:1}
Different from scalar EP, the conventional EP (referred to as diagonal EP in \cite{Ccakmak2018}) matches both mean and variance on the component-wise level \cite{Minka2001}, which seems to be a more natural treatment. For instance, in the diagonal EP approach, the projection operation in \eqref{Eqn:scaler_EP_app} becomes the following:
\begin{align}
\mr{Proj}\big[ b^t_{2}(\bm{z})\big]_i
= \mathcal{N}\Big(z_{i}; {z}^t_{2i},\frac{{v}^t_{2i}}{\beta}\Big).
\end{align}
For the specific problem considered in our paper, $F_2(\bm{z})=\exp\left(\beta\sum_{a=1}^m |z_a|^2\cdot \T(y_a)\right)$ can be viewed as a Gaussian message (up to constant scaling) of $\bm{z}$. In this case, the belief $b^t({\bm{z}})\propto F_2(\bm{z})\cdot\mathcal{N}(\bm{z};\bm{z}^t_{1\to2},1/(\beta\rho_{1\to2}^t)\bm{I})$ is a Gaussian pdf that has a diagonal covariance matrix. Hence, if we apply the diagonal EP approximation, then $\hat{b}^t(\bm{z})=\mathit{Proj}[b^t(\bm{z})]=b^t(\bm{z})$ and so $m^t_{1\leftarrow 2}(\bm{z})\propto \frac{ \hat{{b}}^t(\bm{z})}{ m^t_{1\to 2}(\bm{z})}=F_2(\bm{z})$. This means that the message $m^t_{1\leftarrow 2}(\bm{z})$ is invariant to the input, and hence all the messages in Fig.~\ref{Fig:factor2} remain constant.
\end{remark}

\subsubsection{Estimate of $\mathbf{x}$}
Let $m^{t+1}(\bm{x})$ be the message sent from node $\bm{z}=\bm{Ax}$ to node $\bm{x}$:
\[
m^{t+1}(\bm{x})\propto\frac{\text{Proj}\left[ \int_{\bm{z}} F_1(\bm{x})\mathbb{I}(\bm{z}=\bm{Ax}) m^t_{1\leftarrow2}(\bm{z})\mr{d}\bm{z} \right]}{F_1(\bm{x})}.
\]
The estimate of $\bm{x}$, denoted as $\bm{x}^{t+1}$, is given by the mean of the belief $b(\bm{x})$, where
\[
\begin{split}
&b(\bm{x})\propto F_1(\bm{x})\cdot m^{t+1}(\bm{x})\\
&=\text{Proj}\left[ \int_{\bm{z}} F_1(\bm{x})\mathbb{I}(\bm{z}=\bm{Ax}) m^t_{1\leftarrow2}(\bm{z})\mr{d}\bm{z} \right].
\end{split}
\]
Since $\text{Proj}(\cdot)$ is a moment-matching operation, we have
\BE
\begin{split}
\bm{x}^{t+1} &\propto \iint \bm{x} F_1(\bm{x})\mathbb{I}(\bm{z}=\bm{Ax}) m^t_{1\leftarrow2}(\bm{z})\mr{d}\bm{z}\mr{d}\bm{x}\\
&\propto\int_{\bm{x}} \bm{x}F_1(\bm{x})m^t_{1\leftarrow2}(\bm{Ax})\mr{d}\bm{x}.
\end{split}
\EE
Using $F_1(\bm{x})=\exp\left(-\beta\lambda\|\bm{x}\|^2\right)$, $m^t_{1\leftarrow2}(\bm{z})=\mathcal{N}(\bm{z};\bm{z}^t_{1\leftarrow2},1/(\beta\rho_{1\leftarrow2}^t)\bm{I})$, and after some simple calculations, we get
\BE\label{Eqn:x_final_estimate}
\bm{x}^{t+1}=\left(\frac{2\lambda}{\rho_{1\leftarrow2}^t}\bm{I}+\bm{A}^{\UT}\bm{A}\right)^{-1}\bm{A}^{\UT}\bm{z}^t_{1\leftarrow2}.
\EE

\subsection{Summary of \Alg}
For brevity, we make a few changes to our notations:
\[
\bm{z}^t_{1\to2} \Longrightarrow \bm{z}^t,\quad
\rho^t_{1\to2} \Longrightarrow \rho^t,\quad
\bm{z}^t_{1\leftarrow2} \Longrightarrow \bm{p}^t,\quad
\rho^t_{1\leftarrow 2}  \Longrightarrow \gamma^t.
\]
Combining \eqref{Eqn:SEP_post_2}, \eqref{Eqn:EP_Gaussian_out2}, \eqref{Eqn:SEP_post_1} and \eqref{Eqn:EP_Gaussian_out1}, and after some straightforward calculations, we can simplify the \Alg algorithm as follows:
\BS\label{Eqn:GEC_general}
\begin{align}
\bm{p}^{t} &= \frac{1}{1-\frac{1}{m}\Tr(\bm{G})} \left( \bm{G}- \frac{1}{m}\Tr(\bm{G})\cdot
\bm{I}\right)\bm{z}^t ,\label{Eqn:GEC_general_d}\\
\gamma^{t} &= \rho^t\cdot\left(\frac{1}{\frac{1}{m}\Tr(\bm{G})}-1\right),\label{Eqn:GEC_general_e}\\
\bm{z}^{t+1} &=\frac{1}{1-\frac{1}{m}\Tr(\bm{R})}\left( \bm{R}-\frac{1}{m}\Tr(\bm{R})\cdot\bm{I}
\right) \bm{p}^t \label{Eqn:GEC_general_a},\\
\rho^{t+1} &= \gamma^t\cdot\left( \frac{1}{\frac{1}{m}\Tr(\bm{R})}-1 \right),\label{Eqn:GEC_general_b}
\end{align}
where the matrices $\bm{G}$ and $\bm{R}$ are defined as
\begin{align}
\bm{G}&\Mydef\mr{diag}\left\{\frac{1}{1-2(\rho^t)^{-1} \T(y_1)},\ldots,\frac{1}{1-2(\rho^t)^{-1}
\T(y_m)}\right\} \label{Eqn:G_matrix_def},\\
\bm{R}&\Mydef \bm{A}\left( \frac{2\lambda}{\gamma^t}\bm{I}+\bm{A}^{\UH}\bm{A} \right)^{-1} \bm{A}^{\UH}.
\label{Eqn:R_def_origin}
\end{align}
\ES
The final output of $\bm{x}$ is given by (cf. \eqref{Eqn:x_final_estimate})
\BE\label{Eqn:x_final_estimate2}
\bm{x}^{t+1} = \left( \frac{2\lambda}{\gamma^t}\bm{I} +\bm{A}^{\UH}\bm{A}\right)^{-1}\bm{A}^{\UH}\bm{p}^t.
\EE

In the above algorithm, we did not include iteration indices for $\bm{G}$
and $\bm{R}$. The reason is that we will consider a version of \Alg where
$\rho^t$ and $\gamma^t$ are fixed during the iterative process. \vspace{5pt}

Below are a couple of points we would like to mention:
\begin{itemize}
\item
{From the above derivations, we see that the inverse temperature $\beta$ cancels out in the final \textsf{PCA-EP}-$\beta$ algorithm. This is essentially due to the Gaussianality of $F_1(\bm{x})$ and $F_2(\bm{z})$ (see \eqref{Eqn:optimization_final} and \eqref{Eqn:joint_pdf}). }
\item Although the derivations in Section \ref{Sec:PCA_EP_beta} focus on the real-valued case, the final \Alg algorithm described above is for the general complex-valued case. Notice that we generalized \Alg to the complex-valued case by replacing all matrix transpose to conjugate transpose.
\end{itemize}

\subsection{Stationary points of \Alg}
Suppose that $\bm{p}^{\infty}$, $\bm{z}^{\infty}$, $\bm{x}^{\infty}$, $\gamma^{\infty}$ and $\rho^{\infty}$ are the stationary values of the corresponding variables for the \Alg algorithm. In this section, we will show that $\bm{x}^{\infty}$ is an eigenvector of the data matrix $\bm{D}=\bm{A}^{\UH}\bm{T}\bm{A}$, where $\bm{T}\Mydef \mr{diag}\{\T(y_1),\ldots,\T(y_m)\}$. Our aim is to show
\BE\label{Eqn:eigen_to_prove}
\bm{A}^{\UH}\bm{T}\bm{A}\bm{x}^{\infty}=\lambda \bm{x}^{\infty}.
\EE

Suppose that $\gamma^{\infty}\neq0$. First, combining \eqref{Eqn:GEC_general_e} and \eqref{Eqn:GEC_general_b} yields
\BE\label{Eqn:stationary_div}
\begin{split}
&\frac{\rho^{\infty}}{\gamma^{\infty}} =\frac{1}{\frac{1}{m}\Tr(\bm{R})}-1=\left(\frac{1}{\frac{1}{m}\Tr(\bm{G})}-1\right)^{-1}\\
&\Longrightarrow
\frac{1}{m}\Tr(\bm{R})+\frac{1}{m}\Tr(\bm{G})=1.
\end{split}
\EE
Further,
\BE\label{Eqn:stationary_post_mean}
\begin{split}
\bm{Gz}^{\infty}&\overset{(a)}{=}\left(1-\frac{1}{m}\Tr(\bm{G})\right)\bm{p}^{\infty}+\frac{1}{m}\Tr(\bm{G})\bm{z}^{\infty}\\
&\overset{(b)}{=}\frac{1}{m}\Tr(\bm{R})\bm{p}^{\infty}+\left(1-\frac{1}{m}\Tr(\bm{R})\right)\bm{z}^{\infty}\\
&\overset{(c)}{=}\bm{Rp}^{\infty},
\end{split}
\EE
where step (a) follows from \eqref{Eqn:GEC_general_d}, step (b) from \eqref{Eqn:stationary_div} and step (c) from \eqref{Eqn:GEC_general_a}. We now prove \eqref{Eqn:eigen_to_prove} by showing
\BS
\begin{align}
\bm{A}^{\UH}\bm{T}\bm{A}\bm{x}^{\infty} &= \frac{\rho^{\infty}}{2}\bm{A}^{\UH}(\bm{G}-\bm{I})\bm{z}^{\infty} \label{Eqn:PCA_fix_temp1}\\
 \bm{x}^{\infty} &= \frac{\gamma^{\infty}}{2\lambda}\bm{A}^{\UH}(\bm{I}-\bm{R})\bm{p}^{\infty},\label{Eqn:PCA_fix_temp2}
\end{align}
and
\BE\label{Eqn:PCA_fix_temp3}
(\bm{G}-\bm{I})\bm{z}^{\infty}  = \frac{\gamma^{\infty}}{\rho^{\infty}}(\bm{I}-\bm{R})\bm{p}^{\infty}.
\EE
\ES

We next prove \eqref{Eqn:PCA_fix_temp1}. By using \eqref{Eqn:x_final_estimate2}, we have
\[
\begin{split}
\bm{A}^{\UH}\bm{T}\bm{A}\bm{x}^{\infty}&=\bm{A}^{\UH}\bm{T}\bm{A} \left( \frac{2\lambda}{\gamma^{\infty}}\bm{I} +\bm{A}^{\UH}\bm{A}\right)^{-1}\bm{A}^{\UH}\bm{p}^{\infty}\\
&\overset{(a)}{=}\bm{A}^{\UH}\bm{TRp}^{\infty}\\
&\overset{(b)}{=}\bm{A}^{\UH}\bm{TGz}^{\infty}\\
&\overset{(c)}{=}\frac{\rho^{\infty}}{2}\bm{A}^{\UH}(\bm{G}-\bm{I})\bm{z}^{\infty},
\end{split}
\]
where (a) is from the definition of $\bm{R}$ in \eqref{Eqn:R_def_origin}, (b) is due to \eqref{Eqn:stationary_post_mean}, and finally (c) is from the identity $\bm{TG}=\frac{\rho^{\infty}}{2}(\bm{G}-\bm{I})$ that can be verified from the definition of $\bm{G}$ in \eqref{Eqn:G_matrix_def}.

We now prove \eqref{Eqn:PCA_fix_temp2}. From \eqref{Eqn:R_def_origin}, we have
\[
\begin{split}
\bm{A}^{\UH}(\bm{I}-\bm{R})\bm{z}^{\infty} &=\left(\bm{A}^{\UH}-\bm{A}^{\UH}\bm{A}\left( \frac{2\lambda}{\gamma^{\infty}}\bm{I}+\bm{A}^{\UH}\bm{A} \right)^{-1} \bm{A}^{\UH}\right)\bm{z}^{\infty}\\
&=\left(\bm{I}-\bm{A}^{\UH}\bm{A}\left( \frac{2\lambda}{\gamma^{\infty}}\bm{I}+\bm{A}^{\UH}\bm{A} \right)^{-1}\right)\bm{A}^{\UH}\bm{z}^{\infty}\\
&=\frac{2\lambda}{\gamma^{\infty}}\left( \frac{2\lambda}{\gamma^{\infty}}\bm{I}+\bm{A}^{\UH}\bm{A} \right)^{-1}\bm{A}^{\UH}\bm{z}^{\infty}\\
&=\frac{2\lambda}{\gamma^{\infty}}\bm{x}^{\infty}
\end{split}
\]
where the last step is from the definition of $\bm{x}^{\infty}$ in \eqref{Eqn:x_final_estimate2}.

Finally, we prove \eqref{Eqn:PCA_fix_temp3}. From \eqref{Eqn:GEC_general_d} and \eqref{Eqn:GEC_general_a}, we have
\[
\begin{split}
(\bm{G}-\bm{I})\bm{z}^{\infty}&=\left(1-\frac{1}{m}\Tr(\bm{G})\right)\left(\bm{p}^{\infty}-\bm{z}^{\infty}\right),\\
(\bm{R}-\bm{I})\bm{p}^{\infty}&=\left(1-\frac{1}{m}\Tr(\bm{R})\right)\left(\bm{z}^{\infty}-\bm{p}^{\infty}\right).
\end{split}
\]
Then, \eqref{Eqn:PCA_fix_temp3} follows from these identities together with \eqref{Eqn:stationary_div}.
\subsection{\Alg with fixed tuning parameters}
We consider a version of \Alg where $\rho^t=\rho$ and $\gamma^t=\gamma$, $\forall t$, where $\rho>0$ and $\gamma\in\mathbb{R}$ are understood as tuning parameters. We further assume that $\rho$ and $\gamma$ are chosen such that the following relationship holds (cf. \ref{Eqn:stationary_div}):
\[
\frac{1}{m}\Tr(\bm{G})+\frac{1}{m}\Tr(\bm{R})=1.
\]

Under the above conditions, the \Alg algorithm in \eqref{Eqn:GEC_general} can be written into the following compact form:
\BE\label{Eqn:Alg_fixed}
\begin{split}
\bm{z}^{t+1} = \left( \frac{\bm{R}}{\frac{1}{m}\Tr(\bm{R})}-\bm{I}
\right)\left( \frac{\bm{G}}{\frac{1}{m}\Tr(\bm{G})}-
\bm{I}\right)\bm{z}^t,
\end{split}
\EE
where $\bm{G}$ and $\bm{R}$ are defined in \eqref{Eqn:G_matrix_def} and \eqref{Eqn:R_def_origin}, respectively. Further, there are three tuning parameters involved, namely, $\lambda$, $\rho$ and $\gamma$.

\subsection{\Alg for partial orthogonal matrix}

The \Alg algorithm simplifies considerably for partial orthogonal matrices satisfying $\bm{A}^{\UH}\bm{A}=\bm{I}$. To see this, note that $\bm{R}$ in \eqref{Eqn:R_def_origin} becomes (with $\gamma^t=\gamma$)
\BE
\bm{R}=\frac{1}{\frac{2\lambda}{\gamma}+1}\cdot\bm{AA}^{\UH}.
\EE
Since
\[
\frac{1}{m}\Tr(\bm{AA}^{\UH})=\frac{n}{m}=\frac{1}{\delta},
\]
we have
\[
\frac{\bm{R}}{\frac{1}{m}\Tr(\bm{R}) }= \delta\bm{AA}^{\UH}.
\]
Then, using the above identity and noting the constraint $\frac{1}{m}\Tr(\bm{R})+\frac{1}{m}\Tr(\bm{G})=1$ (cf. \eqref{Eqn:stationary_div}), we could write \eqref{Eqn:Alg_fixed} into the following update:
\BE\label{Eqn:Alg_fixed_orth}
\bm{z}^{t+1}=\left(\delta\bm{AA}^{\UH}-\bm{I}\right)\left(\frac{ \bm{G}}{\frac{1}{m}\Tr(\bm{G})}-
\bm{I}\right)\bm{z}^t,
\EE
where
\[
\bm{G}\Mydef\mr{diag}\left\{\frac{1}{1-2(\rho)^{-1} \T(y_1)},\ldots,\frac{1}{1-2(\rho)^{-1}
\T(y_m)}\right\}
\]
We will treat the parameter $\rho$ as a tunable parameter for \Alg. Finally, \textit{we note that \eqref{Eqn:Alg_fixed_orth} is invariant to a scaling of $\bm{G}$}. We re-define $\bm{G}$ into the following form:
\[
\bm{G}=\mr{diag}\left\{\frac{1}{\mu^{-1}- \T(y_1)},\ldots,\frac{1}{\mu^{-1}-
\T(y_m)}\right\},
\]
where $\mu$ is a tunable parameter. This form of $G$ is more convenient for certain parts of our discussions.

\section{Heuristic derivations of state evolution}\label{App:SE_heuristics}
The \Alg iteration is given by
\BS\label{Eqn:Alg_app}
\BE
\bm{z}^{t+1} =(\delta \bm{AA}^{\UH}-\bm{I})\left(\frac{ \bm{G}}{\G}-\bm{I}\right)\bm{z}^t,
\EE
with the initial estimate distributed as
\BE
\bm{z}^0\overset{d}{=}\alpha_0\bm{z}_\star+\sigma_0{\bm{w}}^0,
\EE
\ES
and ${\bm{w}}^0\sim\mathcal{CN}(\mathbf{0},1/\delta\mb{I})$ is independent of $\bm{z}_\star$.
An appealing property of \Alg is that $\bm{z}^t$ (for $t\ge1$) also satisfies the above ``signal plus white Gaussian noise'' property:
\[
\bm{z}^t\overset{d}{=}\alpha_t\bm{z}_\star+\sigma_0{\bm{w}}^t,
\]
where $\bm{w}^t$ is independent of $\bm{z}_\star$. In the following sections, we will present a heuristic way of deriving
\begin{itemize}
\item The mapping $(\alpha_t,\sigma_t^2)\mapsto(\alpha_{t+1},\sigma^2_{t+1})$;
\item The correlation between the estimate $\bm{x}^{t+1}$ and and true signal $\bm{x}_\star$:
\[
\rho(X_\star,X^{t+1})\Mydef\lim_{m\to\infty}\frac{\langle\bm{x}_\star,\bm{x}^{t+1}\rangle}{\|\bm{x}^{t+1}\|\|\bm{x}_\star\|};
\]
\item The correlation between two consecutive ``noise'' terms:
\[
\rho(W^t,W^{t+1})\Mydef\lim_{m\to\infty}\frac{\langle\bm{w}^t,\bm{w}^{t+1}\rangle}{\|\bm{w}^t\|\|\bm{w}^{t+1}\|}.
\]

\end{itemize}

\subsection{Derivations of $\{\alpha_t\}$ and $\{\sigma_t\}$}
We now provide a heuristic derivation for the SE recursion given in \eqref{Eqn:SE_final}. For convenience, we introduce an auxiliary variable:
\BE\label{Eqn:H_df_first}
\bm{p}^t \Mydef \underbrace{\left( \frac{\bm{G}}{\G}-
\bm{I}\right)\bm{z}^t }_{H_\df(\bm{z}^t,\bm{y};\mu)}
\EE
where $\G \Mydef \frac{1}{m}\Tr(\bm{G})$. Further, based on a heuristic concentration argument, we have
\[
\G\approx \Gs=\mathbb{E}\left[G\right],
\]
where
\[
G\Mydef \frac{1}{\mu^{-1}-\T(|Z_\star|)}
\]
and the expectation is taken over $Z_\star\sim\mathcal{CN}(0,1/\delta)$. The intuition of \Alg is that, in each iteration, $\bm{z}^t$ is approximately distributed as
\BE\label{Eqn:AWGN_t}
\bm{z}^t \overset{d}{=}\alpha_t\bm{z}_{\star}+\sigma_t\bm{w}^t,
\EE
where $\bm{w}^t\sim\mathcal{CN}(\mathbf{0},1/\delta\bm{I})$ is independent of $\bm{z}_{\star}$. Due to this property, the distribution of $\bm{z}^t$ is fully characterized by $\alpha_t$ and $\sigma_t$. State evolution (SE) refers to the map $(\alpha_t,\sigma_t^2)\mapsto(\alpha_{t+1},\sigma_{t+1}^2)$. It is possible to justify this property using the conditioning lemma developed in \cite{Rangan17,Takeuchi2017}. A rigorous proof is beyond the scope of this paper. Instead, we will only provide a heuristic way of deriving the SE maps.

We first decompose $\bm{p}^t$ into a vector parallel to $\bm{z}_{\star}$ and a vector perpendicular to it:
\BE\label{Eqn:xi_def}
\bm{p}^t = \underbrace{\frac{\langle \bm{z}_{\star},\bm{p}^t \rangle}{\|\bm{z}_{\star}\|^2}}_{\alpha_t^p} \cdot\bm{z}_{\star}+\underbrace{\left( \bm{p}^t -\frac{\langle \bm{z}_{\star},\bm{p}^t \rangle}{\|\bm{z}_{\star}\|^2}\cdot\bm{z}_{\star} \right)}_{\bm{\xi}},
\EE
where
\BE\label{Eqn:alpha_p}
\begin{split}
\alpha_t^p &\approx \delta\cdot\mathbb{E}[Z_{\star}^*P^t]\\
&=\delta\cdot\mathbb{E}\left[Z_{\star}^*\left( \frac{G}{\Gs}-1 \right)Z^t\right]\\
&\overset{(a)}{=}\delta\cdot\mathbb{E}\left[|Z_{\star}|^2\left( \frac{G}{\Gs}-1 \right)\right]\cdot\alpha_t\\
&\overset{(b)}{=}\big(\psi_1(\mu)-1\big)\cdot\alpha_t,
\end{split}
\EE
where step (a) is from our assumption $Z^t=\alpha_tZ_\star+W^t$ where $W^t$ is independent of $Z_\star$ and step (b) is from the definition of $\psi_1$ (see \eqref{Eqn:psi_definitions}).
Now, consider the update of $\bm{z}^{t+1}$ given in \eqref{Eqn:Alg_fixed}:
\BE\label{Eqn:AWGN_tplus1}
\begin{split}
\bm{z}^{t+1} &= (\delta\bm{AA}^{\UH}-\bm{I})\bm{p}^t\\
&=(\delta\bm{AA}^{\UH}-\bm{I})\left(\alpha_t^p\bm{z_{\star}} +\bm{\xi}  \right)\\
&=(\delta-1)\cdot\alpha_t^p \cdot \bm{z_{\star}} +\underbrace{ (\delta\bm{AA}^{\UH}-\bm{I})\bm{\xi} }_{\bm{w}^{t+1}},
\end{split}
\EE
where in the last step we used $\bm{AA}^{\UH}\bm{z}_\star=\bm{AA}^{\UH}\bm{A}\bm{x}_\star=\bm{Ax}_\star=\bm{z}_\star$. To derive the SE map, it remains to calculate the (average) variance of the ``effective noise'' $(\delta\bm{AA}^{\UH}-\bm{I})\bm{\xi}$:
\BE\label{Eqn:SE_heuristics}
\begin{split}
&\frac{1}{m}\mathbb{E}\|(\delta\bm{AA}^{\UH}-\bm{I})\bm{\xi}\|^2\\
&\overset{(a)}{=}\frac{\mr{tr}(\delta\bm{AA}^{\UH}-\bm{I})^2}{m}\cdot\mathbb{E}[|\Xi|^2]\\
&\overset{(b)}{=}(\delta-1)\cdot\mathbb{E}[|\Xi|^2]\\
&\overset{(c)}{=}(\delta-1)\cdot\left(\mathbb{E}\left[\left|H_{\df}\right|^2\right]-\left|\alpha_t^p\right|^2\cdot\mathbb{E}[|Z_{\star}|^2]\right)\\
&\overset{(d)}{=}(\delta-1)\cdot\left(\mathbb{E}\left[\left|H_{\df}\right|^2\right]-\frac{1}{\delta}\cdot\left(\psi_1(\mu)-1\right)^2|\alpha_t|^2\right)
\end{split}
\EE
where step (a) follows from the heuristic assumption that $\bm{\xi}$ is ``independent'' of $\bm{A}$\footnote{We do not expect them to be truly independent. However, we expect the SE maps derived under this assumption to be correct.}, step (b) is from $\bm{A}^{\UH}\bm{A}=\bm{I}$ and the definition $\delta=\frac{m}{n}$, step (c) is from the orthogonality between ${Z}_{\star}$ and $\Xi=H_{\df}-\alpha_t^p Z_{\star}$, and step (d) from \eqref{Eqn:alpha_p}.

Next, we simplify the map $(\alpha_t,\sigma^2_t)\mapsto\sigma^2_{t+1}$. Similar to step (c) of \eqref{Eqn:SE_heuristics}, we have
\BE\label{Eqn:SE_fixed3}
\begin{split}
&\mathbb{E}[|\Xi|^2]= \mathbb{E}\left[\left|H_{\df}\right|^2\right]-\left|\alpha_t^p\right|^2\cdot\mathbb{E}[|Z_{\star}|^2]\\
&\overset{(a)}{=}|\alpha_t|^2\cdot \mathbb{E}\left[ \left( \frac{G}{\Gs}-1 \right) ^2\left|Z_{\star}\right|^2 \right] +\frac{\sigma^2_t}{\delta}\cdot\mathbb{E}\left[ \left( \frac{G}{\Gs}-1 \right)^2  \right]-\frac{1}{\delta}\left| \alpha_t^p \right|^2\\
&=|\alpha_t|^2\cdot \left( \frac{\mathbb{E}\left[ |Z_{\star}|^2G^2\right]}{\Gs^2} +\frac{1}{\delta}-\frac{2\mathbb{E}\left[ |Z_{\star}|^2G \right]}{\Gs}\right)\\
&\quad+\frac{\sigma^2_t}{\delta}\left(\frac{\mathbb{E}\left[G^2\right]}{\Gs^2}-1\right)- \frac{1}{\delta}\left| \alpha_t^p \right|^2\\
&\overset{(b)}{=}|\alpha_t|^2\left(\frac{\psi_3^2(\mu)+1-2\psi_1(\mu)}{\delta}\right)+\frac{\sigma^2_t}{\delta}\left(\psi_2(\mu)-1\right) \\
&\quad-\frac{|\alpha_t|^2}{\delta}\cdot\big(\psi_1(\mu)-1\big)^2 \\
&= \frac{1}{\delta}\cdot\left[|\alpha_t|^2\cdot\big( \psi_3^2(\mu)-\psi_1^2(\mu)\big)+\sigma_t^2\cdot\big( \psi_2(\mu)-1 \big)\right],
\end{split}
\EE
where step (a) follows from the definition $Z^t =\alpha_t Z_{\star} +\sigma_t W$ with $W\sim\mathcal{CN}(0,1/\delta)$, and step (b) follows from the definitions in \eqref{Eqn:psi_definitions}, and (c) from \eqref{Eqn:psi_definitions} and \eqref{Eqn:alpha_p}.

From \eqref{Eqn:xi_def}, and \eqref{Eqn:AWGN_tplus1}-\eqref{Eqn:SE_fixed3}, and recalling our definitions of $\alpha_{t+1}$, $\sigma_{t+1}$ in \eqref{Eqn:AWGN_t}, we have
\BS\label{Eqn:SE}
\begin{align}
\alpha_{t+1} &= (\delta-1)\cdot\left(\psi_1(\mu)-1\right)\cdot\alpha_t,\label{Eqn:SE_a}\\
\sigma^2_{t+1}&=(\delta-1)\delta\cdot \mathbb{E}[|\Xi|^2]\label{Eqn:SE_b}\\
&=(\delta-1)\cdot\Big(|\alpha_t|^2\big( \psi_3^2(\mu)-\psi_1^2(\mu)\big)+\sigma_t^2\big( \psi_2(\mu)-1 \big)\Big).\nonumber
\end{align}
\ES
Notice that the extra scaling $\delta$ in \eqref{Eqn:SE_b} (compared with \eqref{Eqn:SE_heuristics}) is due to our assumption that $\mathbb{E}[|W|^2]=1/\delta$ rather than $\mathbb{E}[|W|^2]=1$. In what follows, we will express the SE maps in \eqref{Eqn:SE} using the functions $\psi_1$, $\psi_2$ and $\psi_3$ defined in \eqref{Eqn:psi_definitions}.
\subsection{Derivations of $\rho(X_\star,X^{t+1})$}

From \eqref{Eqn:AWGN_tplus1}, we have
\[
\bm{z}^{t+1} =(\delta-1)\cdot\alpha_t^p \cdot \bm{z_{\star}} + (\delta\bm{AA}^{\UH}-\bm{I})\bm{\xi},
\]
where $\bm{\xi}$ is assumed to be independent of $\bm{A}$. Recall that the final estimate $\bm{x}^{t+1}$ is defined as (cf. \eqref{Eqn:x_hat_def})
\[
\begin{split}
\bm{x}^{t+1}&\propto \bm{A}^{\UH}\bm{z}^{t+1}\\
&=(\delta-1)\cdot\alpha_t^p \cdot \bm{A}^{\UH}\bm{A}\bm{x_{\star}} + (\delta-1)\bm{A}^{\UH}\bm{\xi}\\
&=\underbrace{(\delta-1)\cdot\alpha_t^p}_{\alpha_{t+1}} \bm{x_{\star}} + (\delta-1)\bm{A}^{\UH}\bm{\xi}.
\end{split}
\]
Intuitively, $\bm{A}^{\UH}\bm{\xi}$ is a Gaussian noise term independent of $\bm{x}_\star$, and composed of i.i.d. entries. Its average variance is given by
\[
\begin{split}
\frac{1}{n}\mathbb{E}\left[|\bm{A}^{\UH}\bm{\xi}|^2\right] &=\frac{1}{n}\Tr\mathbb{E}\left[\bm{A}^{\UH}\bm{\xi\xi}^{\UH}\bm{A}\right]\\
&=\mathbb{E}[|\Xi|^2]\cdot \frac{1}{n}\Tr\mathbb{E}\left[\bm{A}^{\UH}\bm{A}\right]\\
&=\mathbb{E}[|\Xi|^2],
\end{split}
\]
where the first two steps are from our heuristic assumptions that $\bm{\xi}$ is independent of $\bm{A}$ and consists of i.i.d. entries. The cosine similarity is then given by
\[
\begin{split}
\frac{\langle\bm{x}_\star,\bm{x}^{t+1}\rangle}{\|\bm{x}^{t+1}\|\|\bm{x}_\star\|} &\approx \frac{\alpha_{t+1}}{\sqrt{|\alpha_{t+1}|^2+(\delta-1)^2\mathbb{E}[|\Xi|^2]}}\\
&= \frac{\alpha_{t+1}}{\sqrt{|\alpha_{t+1}|^2+\frac{\delta-1}{\delta} \cdot \sigma_{t+1}^2 }},
\end{split}
\]
where the last step is from \eqref{Eqn:SE_heuristics} and \eqref{Eqn:SE_b}.

\subsection{Derivations of $\rho(W^t,W^{t+1})$}\label{App:w_corr}
For notational convenience, we rewrite \eqref{Eqn:AWGN_t} as
\BE\label{Eqn:app_AWGN2}
\bm{z}^t\overset{d}{=}\alpha_t\bm{z}_\star+\tilde{\bm{w}}^t,
\EE
where we defined $\tilde{\bm{w}}^t\Mydef \sigma_t\bm{w}^t$.
\textit{In this section, we assume that $\alpha_0\in\mathbb{R}$ and so $\alpha_t\in\mathbb{R}$ for $t\ge1$}; see \eqref{Eqn:SE_final_a}. To analyze the behavior of the \Alg algorithm, we need to understand the evolution of the correlation between two consecutive estimates, or equivalently the correlation between two noise terms $\tilde{\bm{w}}^t$ and $\tilde{\bm{w}}^{t+1}$:
\BE\label{Eqn:rho_w_tilde}
\rho({\bm{w}}^t,{\bm{w}}^{t+1})=\rho(\tilde{\bm{w}}^t,\tilde{\bm{w}}^{t+1})\Mydef\frac{\langle \tilde{\bm{w}}^t,\tilde{\bm{w}}^{t+1} \rangle}{\|\tilde{\bm{w}}^t\|\|\tilde{\bm{w}}^{t+1}\|}.
\EE
We assume that this empirical correlation term converges as $m,n\to\infty$, namely,
\[
\rho(\tilde{W}^t,\tilde{W}^{t+1})\Mydef \lim_{m\to\infty}\rho(\tilde{\bm{w}}^t,\tilde{\bm{w}}^{t+1}).
\]
Suppose that $\rho(\tilde{\bm{w}}^t,\tilde{\bm{w}}^{t+1})\to1$ as $t\to\infty$, and the variables $\alpha_t$ and $\sigma^2_t$ converge to nonzero constants. Then, from \eqref{Eqn:app_AWGN2}, it is straightforward to show that $\bm{z}^t$ converges under an appropriate order, namely, $\lim_{t\to\infty}\lim_{m\to\infty}\frac{1}{m}\|\bm{z}^t-\bm{z}^{t+1}\|^2=0$. This implies $\lim_{t\to\infty}\lim_{n\to\infty}\frac{1}{n}\|\bm{x}^t-\bm{x}^{t+1}\|\to0$ (where $\bm{x}^t\propto \bm{A}^{\UH}\bm{z}^t$), since
\[
\|\bm{A}^{\UH}\bm{z}^t-\bm{A}^{\UH}\bm{z}^{t+1}\|\le \|\bm{A}\|\cdot \|\bm{z}^t-\bm{z}^{t+1}\|=\|\bm{z}^t-\bm{z}^{t+1}\|.
\]
Similarly, we can also show that $\frac{1}{n}\|\bm{x}^t-\bm{x}^{t+s}\|\to0$ for any $s\ge1$, under the above convergence order. This implies the convergence of $\{\bm{x}^t\}$.

From Claim \ref{Lem:SE}, as $m,n\to\infty$, almost surely we have
\BE\label{Eqn:W_norm_prod}
\frac{1}{m}\|\tilde{\bm{w}}^t\|\|\tilde{\bm{w}}^{t+1}\|\to \frac{\sigma_t\sigma_{t+1}}{\delta}.
\EE
Hence, to understand the evolution of $\rho(\tilde{\bm{w}}^t,\tilde{\bm{w}}^{t+1})$, it suffices to understand the evolution of $\frac{1}{m}\langle \tilde{\bm{w}}^t,\tilde{\bm{w}}^{t+1}\rangle$. We assume that this empirical correlation term converges as $m,n\to\infty$, namely,
\[
\frac{1}{m}\langle \tilde{\bm{w}}^t,\tilde{\bm{w}}^{t+1} \rangle\to\mathbb{E}\left[(\tilde{W}^t)^* \tilde{W}^{t+1}\right],
\]
where $\tilde{W}^t\sim\mathcal{CN}(0,\sigma^2_t/\delta)$ and $\tilde{W}^{t+1}\sim\mathcal{CN}(0,\sigma^2_{t+1}/\delta)$. To this end, we will derive a recursive formula for calculating $\mathbb{E}\left[(\tilde{W}^t)^* \tilde{W}^{t+1}\right]$, namely, the map $\mathbb{E}\left[(\tilde{W}^t)^* \tilde{W}^{t+1}\right]\mapsto\mathbb{E}\left[(\tilde{W}^{t+1})^* \tilde{W}^{t+2}\right].$

First, given $\mathbb{E}[(\tilde{W}^t)^* \tilde{W}^{t+1}]$, we can calculate $\mathbb{E}[(P^t)^* P^{t+1}]$ as
\BE\label{Eqn:w2p}
\begin{split}
&\mathbb{E}[(P^t)^* P^{t+1}] \overset{(a)}{=} \mathbb{E}\left[ \left(\frac{G}{\Gs}-1\right)^2(Z^t)^* Z^{t+1} \right]\\
&=\mathbb{E}\left[ \left(\frac{G}{\Gs}-1\right)^2\left(\alpha_t{Z}_\star+\tilde{W}^t\right)^*\left(\alpha_{t+1}{Z}_\star+\tilde{W}^{t+1}\right) \right]\\
&\overset{(b)}{=}\alpha_t\alpha_{t+1}\mathbb{E}\left[ \left(\frac{G}{\Gs}-1\right)^2|{Z}_\star|^2 \right] +\mathbb{E}\left[\left(\frac{G}{\Gs}-1\right)^2\right]\mathbb{E}[(\tilde{W}^t)^*\tilde{W}^{t+1}]\\
&\overset{(c)}{=}\frac{\alpha_t\alpha_{t+1}}{\delta}\left( \psi_3^2(\mu)+1-2\psi_1(\mu) \right)+\left(\psi_2(\mu)-1\right)\mathbb{E}[(\tilde{W}^t)^*\tilde{W}^{t+1}],
\end{split}
\EE
where step (a) is from \eqref{Eqn:H_df_first}, step (b) is due to the independence between $(\tilde{W}^t,\tilde{W}^{t+1})$ and $Z_\star$, and step (c) follows from the definitions of $\psi_1$, $\psi_2$ and $\psi_3$ in \eqref{Eqn:psi_definitions}.

From \eqref{Eqn:xi_def}, we have
\[
\begin{split}
P^t = \alpha_p^t Z_\star+\Xi^t,
\end{split}
\]
where $\Xi^t$ is orthogonal to $Z_\star$ (i.e., $\mathbb{E}[(Z_\star)^*\Xi]=0$).
Hence,
\BE\label{Eqn:p2xi}
\mathbb{E}[(\Xi^t)^*\Xi^{t+1}]=\mathbb{E}[(P^t)^*P^{t+1}]-\frac{\alpha_t^p\alpha_{t+1}^p}{\delta}.
\EE

Finally, we establish the relationship between $\mathbb{E}[(\Xi^t)^*\Xi^{t+1}]$ and $\mathbb{E}[(\tilde{W}^{t+1})^*\tilde{W}^{t+2}]$. This step is similar to our derivations of the SE map in \eqref{Eqn:SE_heuristics}.  Note that
\[
\tilde{\bm{w}}^{t+1}=(\delta\bm{AA}^{\UH}-\bm{I})\bm{\xi}^t.
\]
Then,
\BE\label{Eqn:xi2w}
\begin{split}
\mathbb{E}[(\tilde{W}^{t+1})^*\tilde{W}^{t+2}]&\approx\frac{1}{m}(\tilde{\bm{w}}^{t+1})^\UH\tilde{\bm{w}}^{t+2} \\
&=\frac{1}{m}\Tr\left( (\delta\bm{AA}^{\UH}-\bm{I})\bm{\xi}^{t+1}(\bm{\xi}^t)^\UH (\delta\bm{AA}^{\UH}-\bm{I})\right)\\
 &=\frac{1}{m}\Tr\left( (\delta\bm{AA}^{\UH}-\bm{I})^2\bm{\xi}^{t+1}(\bm{\xi}^t)^\UH \right)\\
 &\overset{(a)}{\approx} \frac{1}{m}\Tr\left( (\delta\bm{AA}^{\UH}-\bm{I})^2\right)\cdot\frac{1}{m}\Tr( \bm{\xi}^{t+1}(\bm{\xi}^t)^\UH)\\
 &\approx (\delta-1)\cdot\mathbb{E}[(\Xi^t)^*\Xi^{t+1}],
\end{split}
\EE
where step (a) follows from a heuristic assumption that $(\bm{\xi}^t,\bm{\xi}^{t+1})$ are ``independent of'' $(\delta\bm{AA}^\UH-\bm{I})$. Combining \eqref{Eqn:w2p}, \eqref{Eqn:p2xi} and \eqref{Eqn:xi2w} yields the following (see equation on the top of next page).
\begin{figure*}[!t]
\normalsize
\BE\label{Eqn:W_recursive}
\begin{split}
\mathbb{E}[(\tilde{W}^{t+1})^*\tilde{W}^{t+2}]&=(\delta-1)\cdot\left(\frac{\alpha_t\alpha_{t+1}}{\delta}\cdot\left( \psi_3^2(\mu)+1-2\psi_1(\mu) \right)+\left(\psi_2(\mu)-1\right)\cdot\mathbb{E}[(\tilde{W}^t)^*\tilde{W}^{t+1}]-\frac{\alpha_t^p\alpha_{t+1}^p}{\delta}  \right)\\
&=(\delta-1)\cdot\left(\frac{\alpha_t\alpha_{t+1}}{\delta}\cdot\Big(\psi_3^2(\mu)-\psi_1^2(\mu)\Big)+\left(\psi_2(\mu)-1\right)\cdot\mathbb{E}[(\tilde{W}^t)^*\tilde{W}^{t+1}]\right).
\end{split}
\EE
\hrulefill
\vspace*{4pt}
\end{figure*}

Finally, we derive the initial condition for the above recursion (i.e., $\mathbb{E}[(\tilde{W}^0)^*\tilde{W}^1]$). Recall that we assumed $\bm{z}^0=\alpha_0\bm{z}_\star+\tilde{\bm{w}}^0$, where $\tilde{\bm{w}}^0\sim\mathcal{CN}(\mathbf{0},\sigma^2_0/\delta\bm{I})$ is independent of $\bm{z}_\star$. Similar to the above derivations, $\tilde{\bm{w}}^1$ can be expressed as
\[
\tilde{\bm{w}}^1 = (\delta\bm{AA}^\UH-\bm{I})\bm{\xi}^0.
\]
Then,
\[
\begin{split}
\mathbb{E}[(\tilde{W}^0)^*\tilde{W}^1] &\approx \frac{1}{m}(\bm{w}^0)^\UH\bm{w}^1\\
&=\frac{1}{m}\Tr\left(  (\delta\bm{AA}^\UH-\bm{I})\bm{\xi}^0 (\tilde{\bm{w}}^0)^\UH\right)\\
&\overset{(a)}{\approx}\frac{1}{m}\Tr\left(  (\delta\bm{AA}^\UH-\bm{I})\right)\cdot\frac{1}{m}\Tr\left(\bm{\xi}^0 (\tilde{\bm{w}}^0)^\UH\right)\\
&=0,
\end{split}
\]
where step (a) is due to the heuristic assumption that $(\bm{w}^0,\bm{\xi}^0)$ are ``independent of'' $\delta\bm{AA}^\UH-\bm{I}$ (similar to \eqref{Eqn:xi2w}).

\subsection{Convergence of $\rho(W^t,W^{t+1})$}\label{App:Cov_convergence}
Suppose that the phase transition condition in Claim \ref{Claim:correlation} is satisfied. Specifically, $\psi_1(\bar{\mu})>\frac{\delta}{\delta-1}$, where $\bar{\mu}$ is defined in \eqref{Eqn:mu_star_def}. Let $\bar{\mu}\in(0,\bar{\mu}]$ be the unique solution to $\psi_1(\hat{\mu})=\frac{\delta}{\delta-1}$. Consider an \Alg algorithm with $\mu=\hat{\mu}$. Under this choice, it is straightforward to show that
\[
\alpha_t=\alpha_0,\quad\forall t\ge1,
\]
where $\{\alpha_t\}_{t\ge1}$ is generated according to \eqref{Eqn:SE_final_a}. Assume that $\alpha_0\neq0$. (We have assumed $\alpha_0\in\mathbb{R}$ in the previous section.) The recursion in \eqref{Eqn:W_recursive} becomes
\[
\begin{split}
&\mathbb{E}[(\tilde{W}^{t+1})^*\tilde{W}^{t+2}]=\frac{(\delta-1)\alpha_0^2}{\delta}\cdot\Big(\psi_3^2(\hat{\mu})-\psi_1^2(\hat{\mu})\Big)\\
&+(\delta-1)\left(\psi_2(\hat{\mu})-1\right)\cdot\mathbb{E}[(\tilde{W}^t)^*\tilde{W}^{t+1}].
\end{split}
\]
From Lemma \ref{App:proof_PT_sec}, we have $\psi_2(\hat{\mu})<\frac{\delta}{\delta-1}$, and so
\[
(\delta-1)\left(\psi_2(\hat{\mu})-1\right)<1.
\]
Further, Lemma \ref{Lem:auxiliary1} shows that $\psi_3^2(\hat{\mu})-\psi_1^2(\hat{\mu})>0$.
Hence, $\mathbb{E}[(\tilde{W}^{t+1})^*\tilde{W}^{t+2}]$ converges:
\BE\label{Eqn:Cor_E_infty}
\lim_{t\to\infty} \mathbb{E}[(\tilde{W}^{t+1})^*\tilde{W}^{t+2}] =\frac{\alpha_0^2}{\delta}\cdot\frac{\psi_3^2(\hat{\mu})-\psi_1^2(\hat{\mu})}{\frac{\delta}{\delta-1}-\psi_2(\hat{\mu})}=\frac{1}{\delta}\cdot\lim_{t\to\infty}\sigma_t^2,
\EE
where the second equality can be easily verified from \eqref{Eqn:SE_final_b}. Combining \eqref{Eqn:rho_w_tilde}, \eqref{Eqn:W_norm_prod}, and \eqref{Eqn:Cor_E_infty} leads to
\[
\begin{split}
\lim_{t\to\infty} \rho({W}^t,{W}^{t+1}) &=\lim_{t\to\infty} \rho(\tilde{W}^t,\tilde{W}^{t+1})\\
& =\frac{\lim_{t\to\infty}\mathbb{E}[(\tilde{W}^{t+1})^*\tilde{W}^{t+2}]}{\sigma^2/\delta}=1.
\end{split}
\]

\section{Connections between the extreme eigenvalues of $\mathbf{D}$ and $\mathbf{E}(\mu)$}

\subsection{Proof of Lemma \ref{Lem:leading_eigen}}\label{App:eigen_max}
Note that we assumed $\hat{\mu}\in(0,1]$ in Lemma \ref{Lem:leading_eigen}.
As mentioned in Remark \ref{Rem:empri_div}, we could replace the average trace $\G$ in $\bm{E}(\hat{\mu})$ by its asymptotic limit $\Gs=\mathbb{E}[G(|Z_\star|,\hat{\mu})]$, and consider the following matrix instead:
\[
\begin{split}
\bm{E}(\hat{\mu})&= \left(\delta\bm{AA}^{\UH}-\bm{I}\right)\left(\frac{ \bm{G}}{\Gs}-
\bm{I}\right)\\
&=\delta\bm{AA}^{\UH}\left(\frac{ \bm{G}}{\Gs}-
\bm{I}\right)-\left(\frac{ \bm{G}}{\Gs}-
\bm{I}\right),
\end{split}
\]
where throughout this appendix we denote
\[
\bm{G}\Mydef(1/\hat{\mu}\bm{I}-\bm{T})^{-1}.
\]
We next prove that if the spectral radius of $\bm{E}(\hat{\mu})$ is one, i.e.,
\[
|\lambda_1(\bm{E}(\hat{\mu})) |=1,
\]
then
\[
\lambda_1(\bm{D})\le\Lambda(\hat{\mu})=\frac{1}{\hat{\mu}}-\frac{\delta-1}{\delta}\cdot\frac{1}{\mathbb{E}[G(|Z_\star|,\hat{\mu})]} .
\]

The characteristic polynomial of $\bm{E}(\hat{\mu})$ is given by
\BE\label{Eqn:chara_poly}
\begin{split}
  f (t) &=  \det (t \tmmathbf{I} - \mb{E}(\hat{\mu}))\\
  &= \det \left( t \tmmathbf{I}  +\left(\frac{ \bm{G}}{\Gs}-
\bm{I}\right) -\delta\bm{AA}^{\UH}\left(\frac{ \bm{G}}{\Gs}-
\bm{I}\right)\right)\\
  & = \det \left( (t - 1)  \tmmathbf{I} + \frac{\tmmathbf{G}}{\Gs}-
  \delta\bm{AA}^{\UH} \left( \frac{\tmmathbf{G}}{\Gs} - \tmmathbf{I} \right) \ \right).
  \end{split}
\EE
Under the condition that $|\lambda_1(\bm{E}(\hat{\mu}) |=1$, the characteristic polynomial
$f (t)$ would have no root in $t \in (1, \infty)$:
\BE\label{Eqn:f_zero}
f (t) \neq 0, \quad \forall t > 1.
\EE
(Here we focus on real-valued $t$, which will be enough for our purpose.) Note that the diagonal matrix $(t - 1)  \tmmathbf{I} +
\frac{\tmmathbf{G}}{\Gs}$ is invertible for $t>1$ since $\bm{G}$ is a diagonal matrix with positive entries:
\BE
 G_{i,i} = \frac{1}{1/\hat{\mu} -  \T (y_i)} > 0.
 \EE
Further, $\Gs=\mathbb{E}[G(|Z_\star|,\hat{\mu})]$ is also positive.
Hence, $f(t)$ in \eqref{Eqn:chara_poly} can be rewritten as \eqref{Eqn:top_temporary} (see the top of the next page),
\begin{figure*}[t]
\BE\label{Eqn:top_temporary}
\begin{split}
f(t)&=\det\left((t-1)\bm{I}+\frac{\bm{G}}{\Gs}-\delta\bm{AA}^{\UH}\left(\frac{\bm{G}}{\Gs}-\bm{I}\right)\right)\\
&\overset{(a)}{=}\det\left((t-1)\bm{I}+\frac{\bm{G}}{\Gs}\right)\cdot
\det\left[\bm{I}-\delta\bm{AA}^{\UH}\left(\frac{\bm{G}}{\Gs}-\bm{I}\right)\left((t-1)\bm{I}+\frac{\bm{G}}{\Gs}\right)^{-1}\right]\\
&\overset{(b)}{=}\det\left((t-1)\bm{I}+\frac{\bm{G}}{\Gs}\right)\cdot\underbrace{\det\left[\bm{I}-\delta\bm{A}^{\UH}\left(\frac{\bm{G}}{\Gs}-\bm{I}\right)\left((t-1)\bm{I}+\frac{\bm{G}}{\Gs}\right)^{-1}\bm{A}\right]}_{g(t)},
\end{split}
\EE
\hrulefill
\vspace*{4pt}
\end{figure*}
where step (a) is from the identity $\det (\bm{PQ}) = \det (\bm{P}) \det
(\bm{Q})$ and step (b) from $\det (\tmmathbf{I} - \bm{PQ}) = \det
(\tmmathbf{I} - \tmmathbf{{QP}})$. Recall from \eqref{Eqn:f_zero} that $f (t) \neq 0$
for $t > 1$. Hence,
\BE
g (t) \neq 0, \quad t > 1.
\EE
Define
\BE\label{Eqn:B_def}
\tmmathbf{B} (t) \Mydef \delta \tmmathbf{A}^{\UH} \left( (t - 1)  \tmmathbf{I}
   + \frac{\tmmathbf{G}}{\Gs} \right)^{- 1} \left( \frac{\tmmathbf{G}}{\Gs}^{} -
   \tmmathbf{I} \right) \tmmathbf{A}, \quad t > 1.
   \EE
The condition $g (t) = \det (\tmmathbf{I} - \tmmathbf{B} (t))\neq0$ for $t > 1$ is equivalent to
\[ \lambda_i (\tmmathbf{B} (t)) \neq 1, \quad \forall i=1,\ldots,m,\quad\text{and}\quad t>1,
\]
where $\lambda_i (\tmmathbf{B} (t))$ denotes the $i^{\tmop{th}}$ largest
eigenvalue of $\tmmathbf{B} (t)$. It is straightforward to see that the entries of $\bm{B}(t)$ converge to zero as $t\to\infty$:
\BE\label{Eqn:B_eigen_neq1}
\lim_{t \rightarrow \infty} \tmmathbf{B} (t) = \tmmathbf{0} .
\EE
Hence, all the eigenvalues of $\tmmathbf{B} (t)$ converge to zero as $t \rightarrow \infty$.
At this point, we note that the eigenvalues of a matrix are continuous functions of it entries. (This observation has been used in \cite{Mondelli2017}, but for a different purpose.) Hence, as
$t$ varies in $[1,\infty]$, $\lambda_1 (\tmmathbf{B}
(t))$ can take any value in $[0, \lambda_1 (\tmmathbf{B} (1))]$. In other words, for any $C \in [0, \lambda_1 (\tmmathbf{B} (1))]$, there exists
  some $t_{\star} \in [1, \infty]$ such that $\lambda_1 (\tmmathbf{B}
  (t_{\star})) = C$. Hence, \eqref{Eqn:B_eigen_neq1} implies
\BE
 \lambda_1 (\tmmathbf{B} (t)) \neq 1, \quad\forall t > 1.
\EE
This means that the interval $[0, \lambda_1 (\tmmathbf{B} (1)))$ does not contain
$1$. Therefore, we must have
\BE\label{Eqn:B_1_lambda}
 \lambda_1 (\tmmathbf{B} (1)) \leqslant 1.
 \EE
From \eqref{Eqn:B_def} when $t = 1$, $\tmmathbf{B}
(t)$ simplifies into
\[
\begin{split}
  \tmmathbf{B} (1) &= \left.\left[ \delta \tmmathbf{A}^H \left( (t - 1)
  \tmmathbf{I} + \frac{\tmmathbf{G}}{\Gs} \right)^{- 1} \left(
  \frac{\tmmathbf{G}}{\Gs} - \tmmathbf{I} \right) \tmmathbf{A} \right]\right |_{t =
  1} \nobracket\\
  & =\delta \tmmathbf{} \tmmathbf{A}^{\UH} \left( \frac{\tmmathbf{G}}{\Gs}
  \right)^{- 1} \left( \frac{\tmmathbf{G}}{\Gs} - \tmmathbf{I} \right)
  \tmmathbf{A}\\
  & = \delta \tmmathbf{A}^{\UH} [\tmmathbf{I} - \Gs\tmmathbf{G}^{-
  1}] \tmmathbf{A}\\
  & =\delta\left(1-\frac{\bar{G}}{\hat{\mu}}\right)\bm{I}+\delta\bar{G}\bm{A}^{\UH}\bm{TA},
  \end{split}
\]
where the last step follows from the definition $\bm{G}=(1/\hat{\mu}\bm{I}-\bm{T})^{-1}$, and the assumption $\tmmathbf{A}^{\UH} \tmmathbf{A} = \tmmathbf{I}$. Hence, $\lambda_1 (\tmmathbf{B} (1)) \leqslant 1$ (see
\eqref{Eqn:B_1_lambda}) readily leads to
\[
\begin{split}
  \lambda_1 (\tmmathbf{A}^{\UH} \tmmathbf{T} \tmmathbf{A}) & \leqslant  \frac{1 -
  (\delta - \delta \Gs/\hat{\mu})}{\delta \Gs}\\
  & = \frac{1}{\hat{\mu}} -\frac{\delta - 1}{\delta} \cdot
  \frac{1}{\Gs}  =\Lambda(\hat{\mu}),
\end{split}
\]
where the last step is from \eqref{Eqn:Lambda_def}. Also, we assumed $1$ is an eigenvalue of $\bm{E}(\mu)$. By Lemma \ref{Lem:stationary}, this implies that $\Lambda(\hat{\mu})$ is an eigenvalue of $\bm{D}$. Summarizing, we have
\[
\lambda_1(\bm{A}^{\UH}\bm{TA})=\Lambda(\hat{\mu}).
\]
\subsection{Proof of Lemma \ref{Lem:minimum_eigen}}\label{App:eigen_min}

We only need to make a few minor changes to the proof in Appendix \ref{App:eigen_max}. When $1/\hat{\mu}\in(-\infty,T_{\min})$, we have
\[
G(y,\hat{\mu})=\frac{1}{\hat{\mu}^{-1} - \T (y)}<0,\quad \forall y\ge0,\ \frac{1}{\hat{\mu}}\in(-\infty,T_{\min}).
\]
Hence,
\[
\bar{G}=\mathbb{E}\left[\frac{1}{1/\hat{\mu}-\T(|Z_\star|)}\right]<0,\quad  \frac{1}{\hat{\mu}}\in(-\infty,T_{\min}).
\]
(On the other hand, $G(y,\hat{\mu})>0$ when $1/\hat{\mu}\in[1,\infty)$.)
Since $G(y,\hat{\mu})<0$ for arbitrary $y\ge0$, the diagonal matrix $\bm{G}/\Gs$ has positive entries and is invertible. Following the proof in Appendix \ref{App:eigen_max}, we still have
\BS\label{Eqn:App_D_min1}
\BE
\lambda_1(\bm{B}(1))\le1,
\EE
where
\BE\label{Eqn:App_D_min2}
\bm{B}(1)=(\delta-\delta\bar{G}/\hat{\mu})+\delta\bar{G}\bm{A}^{\UH}\bm{TA}.
\EE
\ES
Different from Appendix \ref{App:eigen_max}, here we have $\bar{G}<0$, and from \eqref{Eqn:App_D_min1} and \eqref{Eqn:App_D_min2} we have
\[
\lambda_n(\bm{A}^{\UH}\bm{TA})\ge \frac{1}{\hat{\mu}}  - \frac{\delta - 1}{\delta} \cdot
  \frac{1}{\Gs}=\Lambda(\hat{\mu}).
\]
Also, we assumed $1$ is an eigenvalue of $\bm{E}(\mu)$. By Lemma \ref{Lem:stationary}, this implies that $\Lambda(\hat{\mu})$ is an eigenvalue of $\bm{D}$. Summarizing, we have
\[
\lambda_n(\bm{A}^{\UH}\bm{TA})=\Lambda(\hat{\mu}).
\] 
\section{Auxiliary Lemmas}\label{App:aux}
In this appendix, we collect some results regarding the functions $\psi_1$, $\psi_2$ and $\psi_3$ defined in \eqref{Eqn:psi_definitions}.

\begin{lemma}[Auxiliary lemma]\label{Lem:auxiliary1}
For any $\mu\in[0,1]$, we have
\BE
\psi_3(\mu)\ge\psi_1(\mu),
\EE
Further, equality only holds at $\mu=0$.
\end{lemma}
\begin{IEEEproof}
From the definitions in \eqref{Eqn:psi_definitions}, it's equivalent to prove the following:
\[
\mathbb{E}\left[\delta|{Z}_{\star}|^2G^2(Y,\mu)\right] \ge \left(\mathbb{E}\left[\delta|{Z}_{\star}|^2G(Y,\mu)\right]\right)^2.
\]
This follows from the Cauchy-Schwarz inequality:
\[
\begin{split}
\mathbb{E}\left[\delta |{Z}_{\star}|^2G(Y,\mu) \right]&=\mathbb{E}\left[ \delta |{Z}_{\star}|G(Y,\mu)\cdot|{Z}_{\star}| \right]\\
&\le\sqrt{\mathbb{E}\left[ \delta|{Z}_{\star}|^2G^2(Y,\mu) \right]\cdot\mathbb{E}\left[\delta |{Z}_{\star}|^2\right]  }\\
&=\sqrt{\mathbb{E}\left[ \delta|{Z}_{\star}|^2G^2(Y,\mu) \right] },
\end{split}
\]
where the last equality is due to $\mathbb{E}\left[ \delta|{Z}_{\star}|^2\right]=1$.
\end{IEEEproof}

\begin{lemma}[Chebyshev's association inequality \cite{Boucheron2013}]\label{Lem:association}
Let $f$ and $g$ be nondecreasing real-valued functions. If $A$ is a real-valued random variable and $B$ is a nonnegative random variable, then
\[
\mathbb{E}[B]\cdot\mathbb{E}[Bf(A)g(A)]\ge \mathbb{E}[Bf(A)]\cdot \mathbb{E}[Bg(A)].
\]
\end{lemma}

\begin{lemma}[Monotonicity of $\psi_1$ and $\psi_2$]\label{Lem:T_increase}
The function $\psi_2$ in \eqref{Eqn:psi_definitions} is strictly increasing on $(0,1)$. Further, if $\T(y)$ is an increasing function, then $\psi_1(\mu)$ is strictly increasing on $(0,1)$.
\end{lemma}

\begin{IEEEproof}
Let $Z_\star\sim\mathcal{CN}(0,1/\delta)$. For brevity, we denote $G(|Z_\star|,\mu)$ as $G$. Further, the derivative $G'(|Z_\star|,\mu)$ (or simply $G'$) is with respect to $\mu$. From the definition of $\psi_2$ in \eqref{Eqn:psi_definitions}, it is straightforward to show that
\BE
\begin{split}
\psi_2'(\mu) &= 2\cdot\frac{\mathbb{E}[GG']\cdot\mathbb{E}[G]-\mathbb{E}[G^2]\cdot\mathbb{E}[G']}{\left(\mathbb{E}[G]\right)^3}\\
&\overset{(a)}{=}\frac{2}{\mu^2}\cdot\frac{\mathbb{E}[G^3]\cdot\mathbb{E}[G]-\mathbb{E}\left([G^2]\right)^2}{\left(\mathbb{E}[G]\right)^3},
\end{split}
\EE
where step (a) follows from
\[
\begin{split}
G'(Y,\mu)=\frac{1}{(1-\mu\T(Y))^2}=\frac{G^2}{\mu^2}.
\end{split}
\]
Hence, to prove the monotonicity of $\psi_2$, it remains to prove
\[
\mathbb{E}[G^3]\cdot\mathbb{E}[G]>\mathbb{E}\left([G^2]\right)^2,\quad \forall \mu\in(0,1),
\]
which is a direct consequence of Lemma \ref{Lem:association} (with $A=B=G$, $f(A)=g(A)=A$).

We next prove that $\psi_1$ is also an increasing function under the additional assumption that $\T(y)$ is an increasing function of $y$.
The derivative of $\psi_1$ is given by (cf. \ref{Eqn:psi_definitions}):
\BS\label{Eqn:App_Lu_2}
\begin{align}
\psi_1'(\mu) &=\frac{\mathbb{E}\left[\delta|Z_{\star}|^2G'\right]\cdot\mathbb{E}\left[G\right]-\mathbb{E}\left[\delta|Z_{\star}|^2G\right]\cdot\mathbb{E}\left[G'\right]}{\left(\mathbb{E}[G]\right)^2}\\
&=\frac{1}{\mu^2}\cdot\frac{\mathbb{E}[\delta|Z_\star|^2G^2]\cdot\mathbb{E}[G]-\mathbb{E}[\delta|Z_\star|^2G]\cdot\mathbb{E}[G^2]}{\left(\mathbb{E}[G]\right)^2}
\end{align}
\ES
Hence, we only need to prove
\[
\mathbb{E}[\delta|Z_\star|^2G^2]\cdot\mathbb{E}[G]>\mathbb{E}[\delta|Z_\star|^2G]\cdot\mathbb{E}[G^2]
\]
under the condition that $\T(y)$ is an increasing function. Again, this inequality follows from Lemma \ref{Lem:association} by letting $B = G$, $A=|Z_\star|$, $f(A)=A^2$ and $g(A)=G(A,\mu)$. Note that in this case $G(|Z_\star|,\mu)>0$ is an increasing function is $|Z_\star|$ for any $\mu\in(0,1)$. Hence, both $f(\cdot)$ and $g(\cdot)$ are increasing functions and the conditions in Lemma \ref{Lem:association} are satisfied.
\end{IEEEproof}

\begin{lemma}\label{Lem:psi1_psi2}
When $\T(y)=\T_\opt(y)\Mydef1-\frac{1}{\delta y^2}$, we have
\BE\label{Eqn:App_Lu_4}
\psi_1(\mu)>\psi_2(\mu),\quad\text{for }\mu\in(0,1).
\EE
\end{lemma}
\begin{IEEEproof}
From the definitions in \eqref{Eqn:psi_definitions}, \eqref{Eqn:App_Lu_4} is equivalent to
\BE\label{Eqn:App_Lu_5}
\mathbb{E}\left[\delta|{Z}_{\star}|^2G(Y,\mu)\right]\cdot\mathbb{E}[G(Y,\mu)]>\mathbb{E}[G^2(Y,\mu)],\quad\text{for }\mu\in(0,1),
\EE
where $Y=|Z_\star|$ and $Z_\star\sim\mathcal{CN}(0,1/\delta)$.
Since $\mathbb{E}[\delta|{Z}_{\star}|^2]=1$, we can write \eqref{Eqn:App_Lu_5} as
\BE\label{Eqn:App_Lu_6}
\mathbb{E}\left[\delta|{Z}_{\star}|^2G(Y,\mu)\right]\cdot\mathbb{E}[G(Y,\mu)]>\mathbb{E}[G^2(Y,\mu)]\cdot\mathbb{E}[\delta|{Z}_{\star}|^2].
\EE
We will use Lemma \ref{Lem:association} to prove \eqref{Eqn:App_Lu_6}. Specifically, we apply Lemma \ref{Lem:association} by specifying the random variables $A$ and $B$, and the functions $f$ and $g$ as follows:
\[
\begin{split}
A&=|Z_{\star}|,\\
B&=G(Y,\mu),\\
f(A)&=\delta A^2\cdot \frac{1}{G(A,\mu)},\\
g(A)&= G(A,\mu).
\end{split}
\]
It is easy to show that $B>0$ and $g(\cdot)$ is an increasing function; it only remains to prove that $f(\cdot)$ is an increasing function. From the definitions $G(y,\mu)=\left(1/\mu-\T(y)\right)^{-1}$ and $\T(y)=1-1/(\delta y^2)$, we have
\[
\begin{split}
f(A)&=\delta A^2\cdot \frac{1}{G(A,\mu)}\\
&=\delta A^2\cdot\left( \frac{1}{\mu}-\T(A) \right)\\
&=\delta A^2\cdot\left( \frac{1}{\mu}-\left( 1-\frac{1}{\delta A^2} \right) \right)\\
&=\frac{1-\mu}{\mu}\delta A^2+1,
\end{split}
\]
which is an increasing function of $A$ (which is defined on $\mathbb{R}_+$) when $\mu\in(0,1)$. This completes our proof.
\end{IEEEproof}

\bibliographystyle{IEEEtran}    
\bibliography{IEEEabrv,Phase_retrieval}         

\end{document}